\documentclass[a4paper,fleqn]{article}
\usepackage{a4wide}
\usepackage{amsmath,amssymb,amsthm}
\usepackage{thmtools}
\usepackage{enumerate}
\usepackage{abstract}
\usepackage{tikz}
\usetikzlibrary{decorations.pathreplacing} 
\usetikzlibrary{calc}                      
\usepackage{diagbox}
\usepackage{hyperref}
\usepackage[T1]{fontenc}
\usepackage{mathpazo}
\makeatletter\@ifpackageloaded{mathpazo}\@tempswatrue\@tempswafalse
\if@tempswa
  \DeclareFontFamily{OT1}{pzc}{}
  \DeclareFontShape{OT1}{pzc}{m}{it}{<-> s * [1.15] pzcmi7t}{}
  \DeclareMathAlphabet{\mathpzc}{OT1}{pzc}{m}{it}
\fi\makeatother
\usepackage[utf8]{inputenc}

\usepackage[sorting=nyt,firstinits=true,backend=bibtex,maxbibnames=4]{biblatex}
\makeatletter\@ifpackageloaded{biblatex}{%
  \bibliography{references}
  \renewbibmacro{in:}{%
    \ifentrytype{incollection}{\printtext{\bibstring{in}\intitlepunct}}{}}
  \renewbibmacro{publisher+location+date}{%
    \iflistundef{publisher}
      {\setunit*{\addcomma\space}}
      {\setunit*{\addcomma\space}}%
    \printlist{publisher}%
    \setunit*{\addcomma\space}%
    \printlist{location}%
    \setunit*{\addcomma\space}%
    \usebibmacro{date}%
    \newunit}
  \DeclareFieldFormat[article]{pages}{#1\isdot}
  \DeclareFieldFormat[article,incollection,inproceedings,unpublished]{title}{#1\isdot}
  \DeclareFieldFormat[article]{journaltitle}{\mkbibemph{#1\isdot}}
  
}{}\makeatother

\declaretheorem[numberwithin=section,refname={theorem,theorems},Refname={Theorem,Theorems}]{theorem}
\declaretheorem[sibling=theorem,style=definition]{definition}
\declaretheorem[sibling=theorem,name=Lemma]{lemma}
\declaretheorem[sibling=theorem,name=Proposition]{proposition}
\declaretheorem[sibling=theorem,name=Corollary]{corollary}

\declaretheorem[numberwithin=theorem,name=Claim,Refname={Claim,Claims}]{claim}

\declaretheorem[numbered=no,name=Question]{question}

\makeatletter\@ifpackageloaded{hyperref}{%
  \usepackage{xcolor}
  \definecolor{dark-red}{rgb}{0.4,0.15,0.15}
  \definecolor{dark-blue}{rgb}{0.15,0.15,0.4}
  \definecolor{medium-blue}{rgb}{0,0,0.5}
  \hypersetup{
    colorlinks,
    linkcolor={dark-red},
    citecolor={dark-blue},
    urlcolor={medium-blue}%
  }

}{}\makeatother

\newcommand{\address}[1]{\vspace{-2em}\begin{center}{\footnotesize #1}\end{center}}

\usepackage{sectsty}
\sectionfont{\large\bfseries}
\subsectionfont{\normalsize\bfseries}

\newcommand{\mirror}[1]{\widetilde{#1}}
\newcommand{\N}{\mathbb{N}}
\newcommand{\Z}{\mathbb{Z}}
\newcommand{\T}{\mathbb{T}}
\newcommand{\Lang}{\mathcal{L}}
\newcommand{\st}[1]{Stand(#1)}
\newcommand{\sst}[1]{Stand^+\!(#1)}
\newcommand{\rst}[1]{RStand(#1)}
\newcommand{\rsst}[1]{RStand^+\!(#1)}
\newcommand{\B}{\mathcal{B}}
\newcommand{\oa}{\mathfrak{a}}
\newcommand{\ob}{\mathfrak{b}}

\newcommand{\keywords}[1]{\par\noindent{\footnotesize{\em Keywords\/}: #1}}

\begin{document}
  \title{A square root map on Sturmian words}
  \author{Jarkko Peltomäki\textsuperscript{a} and Markus Whiteland\textsuperscript{b}\\
          \small \href{mailto:jspelt@utu.fi}{jspelt@utu.fi}, \quad \href{mailto:mawhit@utu.fi}{mawhit@utu.fi}}
  \date{}
  \maketitle
  \address{\textsuperscript{a}Turku Centre for Computer Science TUCS, 20520 Turku, Finland\\
  \textsuperscript{a,b}University of Turku, Department of Mathematics and Statistics, 20014 Turku, Finland}

  \noindent
  \hrulefill
  \begin{abstract}
    \vspace{-1em}
    \noindent

    \noindent
    We introduce a square root map on Sturmian words and study its properties. Given a Sturmian word of slope $\alpha$,
    there exists exactly six minimal squares in its language (a minimal square does not have a square as a proper
    prefix). A Sturmian word $s$ of slope $\alpha$ can be written as a product of these six minimal squares:
    $s = X_1^2 X_2^2 X_3^2 \cdots$. The square root of $s$ is defined to be the word $\sqrt{s} = X_1 X_2 X_3 \cdots$.
    The main result of this paper is that that $\sqrt{s}$ is also a Sturmian word of slope $\alpha$. Further, we
    characterize the Sturmian fixed points of the square root map, and we describe how to find the intercept of
    $\sqrt{s}$ and an occurrence of any prefix of $\sqrt{s}$ in $s$. Related to the square root map, we characterize
    the solutions of the word equation $X_1^2 X_2^2 \cdots X_n^2 = (X_1 X_2 \cdots X_n)^2$ in the language of Sturmian
    words of slope $\alpha$ where the words $X_i^2$ are minimal squares of slope $\alpha$.

    We also study the square root map in a more general setting. We explicitly construct an infinite set of
    non-Sturmian fixed points of the square root map. We show that the subshifts $\Omega$ generated by these words have
    a curious property: for all $w \in \Omega$ either $\sqrt{w} \in \Omega$ or $\sqrt{w}$ is periodic. In particular,
    the square root map can map an aperiodic word to a periodic word.
    \vspace{1em}
    \keywords{sturmian word, standard word, optimal squareful word, word equation, continued fraction}
    \vspace{-1em}
  \end{abstract}
  \hrulefill

  \section{Introduction}
  Kalle Saari studies in \cite{diss:kalle_saari,2010:everywhere_alpha-repetitive_sequences_and_sturmian_words} optimal
  squareful words which are aperiodic words containing the least number of minimal squares (that is, squares with no
  proper square prefixes) such that every position starts a square. Saari proves that an optimal squareful word always
  contains exactly six minimal squares, and he characterizes these squares; less than six minimal squares forces a word
  to be ultimately periodic. Moreover, he shows that Sturmian words are a proper subclass of optimal squareful words.
  
  We propose a square root map for Sturmian words. Let $s$ be a Sturmian word of slope $\alpha$, and write it as a
  product of the six minimal squares in its language $\Lang(\alpha)$: $s = X_1^2 X_2^2 X_3^2 \cdots$. The square root
  of $s$ is defined to be the word $\sqrt{s} = X_1 X_2 X_3 \cdots$. The main result of this paper is that the word
  $\sqrt{s}$ is also a Sturmian word of slope $\alpha$. More precisely, we prove that the square root of the Sturmian
  word $s_{x,\alpha}$ of intercept $x$ and slope $\alpha$ is $s_{\psi(x),\alpha}$ where
  $\psi(x) = \frac{1}{2}(x + 1 - \alpha)$. In addition to proving that the square root map preserves the language of a
  Sturmian word $s$, we show how to locate any prefix of $\sqrt{s}$ in $s$. We also characterize the Sturmian words of
  slope $\alpha$ which are fixed points of the square root map; they are the two Sturmian words $01c_\alpha$ and
  $10c_\alpha$ where $c_\alpha$ is the infinite standard Sturmian word of slope $\alpha$. The majority of the proofs of
  results on Sturmian words rely heavily on the interpretation of Sturmian words as rotation words. Continued fractions
  and results from Diophantine approximation theory play a key role in several proofs.

  Solutions of the word equation $X_1^2 \cdots X_n^2 = (X_1 \cdots X_n)^2$ where the words $X_i^2$ are among the six
  minimal squares in $\Lang(\alpha)$ for some fixed irrational $\alpha$ are closely linked to the square root map. The
  study of these solutions to this word equation arises naturally from the study of fixed points of the square root
  map. The Sturmian fixed points of the square root map are fixed because they have arbitrarily long prefixes
  $X_1^2 \cdots X_n^2$ which satisfy the word equation. We characterize these specific solutions, i.e., those primitive
  words $w$ such that $w^2 \in \Lang(\alpha)$ and $w^2$ can be written as a product of minimal squares
  $X_1^2 \cdots X_n^2$ satisfying the word equation. On the circle $[0,1)$, the interval $[w]$ of such a word $w$ can
  be seen to satisfy the square root condition $\psi([w^2]) \subseteq [w]$, so we instead study and characterize the
  primitive words satisfying this square root condition. The result is that the specific solutions to the word equation
  (or, equivalently, the primitive words satisfying the square root condition) are the reversals of standard and
  semistandard words of slope $\alpha$ (see \autoref{ssec:sturmian_words} for a definition) and the reversed standard
  words with the first two letters exchanged. In particular, all of these specific solutions are nonperiodic. It was
  known that the word equation $(X_1^2 \cdots X_n^2) = X_1^2 \cdots X_n^2$ has nonperiodic solutions
  \cite{2000:in_search_of_a_word_with_special_combinatorial_properties}, but according to our knowledge no large
  families of nonperiodic solutions have been identified until our result. Word equations of the type
  $X_1^k \cdots X_n^k = (X_1 \cdots X_n)^k$ have been considered by \v{S}t\v{e}p\'{a}n Holub
  \cite{1999:a_solution_of_the_equation,2000:in_search_of_a_word_with_special_combinatorial_properties,2001:local_and_global_cyclicity_in_free_semigroups}.
 
  The final central topic of this paper concerns the square root map in a more general setting. The square root map can
  be defined not only for Sturmian words but for any optimal squareful word. We construct an infinite family of
  non-Sturmian, linearly recurrent optimal squareful words $\Gamma$ with properties similar to Sturmian words. The
  words $\Gamma$ are fixed points of the square root map. They are constructed by finding non-Sturmian solutions of the
  word equation $X_1^2 \cdots X_n^2 = (X_1 \cdots X_n)^2$ and by building infinite words having arbitrarily long
  squares of such solutions as prefixes. The subshifts $\Omega$ generated by the words $\Gamma$ exhibit behavior
  similar to Sturmian subshifts. The square root map preserves the language of several but not every word in $\Omega$.
  Curiously, if the language of a word in $\Omega$ is not preserved under the square root map, then the image must be
  periodic. This result is very surprising since it is contrary to the plausible hypothesis that the square root of an
  aperiodic word is aperiodic.

  The paper is organized as follows. In \autoref{sec:square_root_map} we prove that the square root map preserves the
  language of a Sturmian word. As a corollary we obtain a description of those Sturmian words which are fixed points of
  the square root map. In \autoref{sec:square_root_map} we observe that the intervals of the minimal squares in
  $\Lang(\alpha)$ satisfy the square root condition. In \autoref{sec:square_root_condition} we characterize all words
  $w^2 \in \Lang(\alpha)$ satisfying the square root condition. The result is that $w^2$ with $w$ primitive satisfies
  the square root condition if and only if $w$ is a reversed standard or semistandard word or a reversed standard word
  with the first two letters exchanged. \autoref{sec:word_equation_characterization} contains a proof of the
  characterization of the specific solutions of the word equation $X_1^2 \cdots X_n^2 = (X_1 \cdots X_n)^2$ mentioned
  earlier. We show that a primitive word $w$ satisfies the square root condition if and only if $w^2$ can be written as
  a product of minimal squares satisfying the word equation. In \autoref{sec:combinatorial_version} we show how to
  locate prefixes of $\sqrt{s}$ in $s$. As an important step in proving this, we provide necessary and sufficient
  conditions for a Sturmian word to be a product of squares of reversed standard and semistandard words. We give a
  formula describing the square root of the Fibonacci word in \autoref{sec:fibonacci}. \autoref{sec:counter_example} is
  devoted to constructing the non-Sturmian fixed points $\Gamma$ mentioned above and to demonstrating that the
  languages of the words in their subshifts are preserved or they are mapped to periodic words. We conclude the paper
  by giving some remarks on possible generalizations in \autoref{sec:generalizations} and by discussing a few open
  problems in \autoref{sec:open_problems}.

  A short version of this paper was published as an extended abstract in the proceedings of WORDS 2015
  \cite{2015:a_square_root_map_on_sturmian_words}.

  \section{Notation and Preliminary Results}\label{sec:preliminaries}
  In this section we review notation and basic concepts and results of word combinatorics, optimal squareful words,
  continued fractions, and Sturmian words. Most of the definitions and results provided here about words can be found
  in Lothaire's book \cite{2002:algebraic_combinatorics_on_words}.

  \emph{An alphabet} $A$ is a finite non-empty set of \emph{letters}, or \emph{symbols}. A (finite) \emph{word} over
  $A$ is a finite sequence of letters of $A$ obtained by concatenation. The concatenation of two words
  $u = a_0 \cdots a_{n-1}$ and $v = b_0 \cdots b_{m-1}$ is the word
  $u \cdot v = uv = a_0 \cdots a_{n-1} b_0 \cdots b_{m-1}$. In this paper we consider only \emph{binary words}, that
  is, words over an alphabet of size two. Most of the time we take $A$ to be the set $\{0,1\}$. \emph{The set of
  nonempty words over $A$} is denoted by $A^+$. We denote \emph{the empty word} by $\varepsilon$ and set
  $A^* = A^+ \cup \{\varepsilon\}$. A nonempty subset of $A^*$ is called a \emph{language}. Let
  $w = a_0 a_1 \cdots a_{n-1}$ be a word of $n$ letters. We denote \emph{the length} $n$ of $w$ by $|w|$; by convention
  $|\varepsilon| = 0$. The set of proper powers of a word $w$ is denoted by $w^+$.
  
  An \emph{infinite word} $w$ over the alphabet $A$ is a function from the nonnegative integers to $A$. We write concisely
  $w = a_0 a_1 a_2 \cdots$ with $a_i \in A$. The set of infinite words over $A$ is denoted by $A^\omega$. An infinite
  word $w$ is said to be \emph{ultimately periodic} if we can write it in the form $w = uv^\omega = uvvv \cdots$ for
  some words $u,v \in A^*$. If $u = \varepsilon$, then $w$ is said to be \emph{periodic}, or \emph{purely periodic}. An
  infinite word which is not ultimately periodic is \emph{aperiodic}. The \emph{shift operator $T$} acts on infinite
  words as follows: $T(a_0 a_1 a_2 \ldots) = a_1 a_2 \cdots$.

  A finite word $u$ is a \emph{factor} of the finite or infinite word $w$ if we can write $w = vuz$ for some
  $v \in A^*$ and $z \in A^* \cup A^\omega$. If $v = \varepsilon$, then the factor $u$ is called a \emph{prefix} of
  $w$. If $z = \varepsilon$, then we say that $u$ is a \emph{suffix} of $w$. The set of factors of $w$,
  the \emph{language of $w$}, is denoted by $\Lang(w)$.
  If $w = a_0 a_1 \cdots a_{n-1}$, then we let $w[i,j] = a_i \cdots a_j$ whenever the choices of positions $i$ and $j$
  make sense. This notion is extended to infinite words in a natural way. An \emph{occurrence} of $u$ in $w$ is a
  position $i$ such that $w[i,i+|u|-1] = u$. If such a position exists, then we say that $u$ \emph{occurs} in $w$.

  A positive integer $p$ is a \emph{period} of $w = a_0 \cdots a_{n-1}$ if $a_i = a_{i+p}$ for $0 \leq i \leq n-p-1$.
  If the finite word $w$ has period $p$ and $|w|/p \geq \alpha$ for some real $\alpha$ such that $\alpha \geq 1$, then
  $w$ is called an \emph{$\alpha$-repetition}. An $\alpha$-repetition is \emph{minimal} if it does not have an
  $\alpha$-repetition as a proper prefix. If $w = u^2$, then $w$ is a \emph{square} with \emph{square root} $u$. A
  square is \emph{minimal} if it does not have a square as a proper prefix. A word $w$ is \emph{primitive} if it is of
  the form $z^n$ if and only if $n = 1$. Equivalently, a word $w$ is primitive if and only if $w$ occurs in $w^2$
  exactly twice. The \emph{primitive root of $w$} is the unique primitive word $u$ such that $w = u^n$ for some
  $n \geq 1$. Let $w = v^\omega$ be a periodic infinite word. The \emph{minimal period of $w$} is defined to be the
  primitive root of $v$.
  
  Let $w = a_0 a_1 \cdots a_{n-1}$ be a word. \emph{The reversal} $\mirror{w}$ of $w$ is the word $a_{n-1} \cdots a_1
  a_0$. If $w = \mirror{w}$, then we call $w$ \emph{a palindrome}. Let $C$ be the \emph{cyclic shift operator} defined
  by the formula $C(a_0 a_1 \cdots a_{n-1}) = a_1 \cdots a_{n-1} a_0$. The words $w,C(w),C^2(w),\ldots,C^{|w|-1}(w)$
  are the \emph{conjugates of $w$}. If $u$ is a conjugate of $w$, then we say that $u$ is conjugate to $w$.
  
  An infinite word $w$ is \emph{recurrent} if each of its factors occurs in it infinitely often. Let $(i_n)_{n\geq1}$
  be the sequence of consecutive occurrences of a factor $u$ in a recurrent word $w$. The \emph{return time} of $u$ is
  the quantity
  \begin{align*}
    \sup\{i_{j+1} - i_j\colon j \in \{1,2,\ldots\}\},
  \end{align*}
  which can be infinite. The factors $w[i_j,i_{j+1}-1]$, $j \geq 1$ are the \emph{returns to $u$} in $w$. If the return
  time of each factor of $w$ is finite, then the word $w$ is \emph{uniformly recurrent}. Equivalently, $w$ is uniformly
  recurrent if for each factor $u$ of $w$ there exists an integer $R$ such that every factor of $w$ of length $R$
  contains an occurrence of $u$. If there exists a global constant $K$ such that the return time of any factor $u$ of
  $w$ is at most $K|u|$, then we say that $w$ is \emph{linearly recurrent}. Clearly a linearly recurrent word is
  uniformly recurrent. The \emph{index} of a factor $u$ of an infinite word $w$ is defined to be
  \begin{align*}
    \sup\{n\colon u^n \in \Lang(w)\}.
  \end{align*}
  If $w$ is uniformly recurrent and aperiodic, then the index of every factor of $w$ is finite.

  A \emph{subshift} $\Omega$ is a subset of $A^\omega$ such that
  \begin{align*}
    \Omega = \{w \in A^\omega\colon \Lang(w) \subseteq \Lang\}
  \end{align*}
  for some language $\Lang$ such that $\Lang \subseteq A^*$. If we set above $\Lang = \Lang(w)$ where $w$ is an
  infinite word, then we say that the subshift $\Omega$ is generated by $w$. Subshifts are clearly shift-invariant. If
  every word in a subshift is aperiodic, then we call the subshift \emph{aperiodic}. A subshift is \emph{minimal} if it
  does not contain nonempty subshifts as proper subsets. A nonempty subshift is minimal if and only if it is generated
  by a uniformly recurrent word.

  \subsection{Optimal Squareful Words}
  In \cite{2010:everywhere_alpha-repetitive_sequences_and_sturmian_words} Kalle Saari considers $\alpha$-repetitive
  words. An infinite word is \emph{$\alpha$-repetitive} if every position in the word starts an $\alpha$-repetition and
  the number of distinct minimal $\alpha$-repetitions occurring in the word is finite. If $\alpha = 2$, then
  $\alpha$-repetitive words are called \emph{squareful words}. This means that every position of a squareful word
  begins with a minimal square. Saari proves that if the number of distinct minimal squares occurring in a squareful
  word is at most $5$, then the word must be ultimately periodic. On the other hand, if a squareful word contains at
  least $6$ distinct minimal squares, then aperiodicity is possible. Saari calls the aperiodic squareful words
  containing exactly $6$ minimal squares \emph{optimal squareful words}. Further, he shows that optimal squareful words
  are always binary and that the six minimal squares must take a very specific form:

  \begin{proposition}\label{prp:min_square_roots}
    Let $w$ be an optimal squareful word. If $10^i 1$ occurs in $w$ for some $i > 1$, then
    the roots of the six minimal squares in $w$ are
    \begin{alignat}{2}\label{eq:min_squares}
      &S_1 = 0,                 && S_4 = 10^\oa, \nonumber \\
      &S_2 = 010^{\oa-1}, \quad && S_5 = 10^{\oa+1}(10^\oa)^\ob, \\
      &S_3 = 010^\oa,           && S_6 = 10^{\oa+1}(10^\oa)^{\ob+1}, \nonumber
    \end{alignat}
    for some $\oa \geq 1$ and $\ob \geq 0$.
  \end{proposition}

  The optimal squareful words containing the minimal square roots of \eqref{eq:min_squares} are called \emph{optimal
  squareful words with parameters $\oa$ and $\ob$}. For the rest of this paper we reserve this meaning for the symbols
  $\oa$ and $\ob$. Furthermore, we agree that the symbols $S_i$ always refer to the minimal square roots
  \eqref{eq:min_squares}.

  Saari completely characterizes optimal squareful words
  \cite[Theorem~17]{2010:everywhere_alpha-repetitive_sequences_and_sturmian_words}.

  \begin{proposition}\label{prp:optimal_squareful_characterization}
    An aperiodic infinite word $w$ is optimal squareful if and only if (up to renaming of letters) there exists
    integers $\oa \geq 1$ and $\ob \geq 0$ such that $w$ is an element of the language
    \begin{align*}
      0^* (10^\oa)^* (10^{\oa+1}(10^\oa)^\ob + 10^{\oa+1}(10^\oa)^{\ob+1})^\omega = S_1^* S_4^* (S_5 + S_6)^\omega.
    \end{align*}
  \end{proposition}

  \subsection{Continued Fractions and Rational Approximations}\label{ssec:continued_fractions}
  In this section we review results on continued fractions and best rational approximations of irrational numbers
  needed in this paper. Good references on these subjects are the books of Khinchin \cite{1997:continued_fractions} and
  Cassels \cite{1957:an_introduction_to_diophantine_approximation}.

  Every irrational real number $\alpha$ has a unique infinite continued fraction expansion
  \begin{align}\label{eq:cf}
    \alpha = [a_0; a_1, a_2, a_3, \ldots] = a_0 + \dfrac{1}{a_1 + \dfrac{1}{a_2 + \dfrac{1}{a_3 + \cdots}}}
  \end{align}
  with $a_0 \in \Z$ and $a_k \in \N$ for all $k \geq 1$. The numbers $a_i$ are called \emph{the partial quotients} of
  $\alpha$. We focus here only on irrational numbers, but we note that with small tweaks much of what follows also
  holds for rational numbers, which have finite continued fraction expansions.
  
  The \emph{convergents} $c_k = \frac{p_k}{q_k}$ of $\alpha$ are defined by the recurrences
  \begin{alignat*}{4}
    p_0 = a_0, &\qquad p_1 = a_1 a_0 + 1, &\qquad p_k = a_k p_{k-1} + p_{k-2}, \qquad& k \geq 2, \\
    q_0 = 1,   &\qquad q_1 = a_1,         &\qquad q_k = a_k q_{k-1} + q_{k-2}, \qquad& k \geq 2.
  \end{alignat*}
  The sequence $(c_k)_{k \geq 0}$ converges to $\alpha$. Moreover, the even convergents are less than $\alpha$ and form
  an increasing sequence and, on the other hand, the odd convergents are greater than $\alpha$ and form a decreasing
  sequence.

  If $k \geq 2$ and $a_k > 1$, then between the convergents $c_{k-2}$ and $c_k$ there are \emph{semiconvergents}
  (called intermediate fractions in Khinchin's book \cite{1997:continued_fractions}) which are of the form
  \begin{align*}
    \frac{p_{k,\ell}}{q_{k,\ell}} = \frac{\ell p_{k-1} + p_{k-2}}{\ell q_{k-1} + q_{k-2}}
  \end{align*}
  with $1 \leq \ell < a_k$. When the semiconvergents (if any) between $c_{k-2}$ and $c_k$ are ordered by the size of
  their denominators, the sequence obtained is increasing if $k$ is even and decreasing if $k$ is odd.

  Note that we make a clear distinction between convergents and semiconvergents, i.e., convergents are not a specific
  subtype of semiconvergents.

  A rational number $\frac{a}{b}$ is a \emph{best approximation} of the real number $\alpha$ if for every fraction
  $\frac{c}{d}$ such that $\frac{c}{d} \neq \frac{a}{b}$ and $d \leq b$ it holds that
  \begin{align*}
    \left|b\alpha - a\right| < \left|d\alpha - c\right|.
  \end{align*}
  In other words, any other multiple of $\alpha$ with a coefficient at most $b$ is further away from the nearest
  integer than $b\alpha$ is. The next important proposition shows that the best approximations of an irrational number
  are connected to its convergents (for a proof see Theorems 16 and 17 of \cite{1997:continued_fractions}).

  \begin{proposition}\label{prp:best_approximation_convergent}
    The best rational approximations of an irrational number are exactly its convergents.
  \end{proposition}

  We identify the unit interval $[0,1)$ with the unit circle $\T$. Let $\alpha \in (0,1)$ be irrational. The map
  \begin{align*}
    R: [0, 1) \to [0, 1), \, x \mapsto \{x + \alpha\},
  \end{align*}
  where $\{x\}$ stands for the fractional part of the number $x$, defines a rotation on $\T$. The circle partitions
  into the intervals $(0,\frac{1}{2})$ and $(\frac{1}{2},1)$. Points in the same interval of the partition are said to
  be on the same side of $0$ and points in different intervals are said to be on the opposite sides of $0$. (We are not
  interested in the location of the point $\frac{1}{2}$.) The points $\{q_k \alpha\}$ and $\{q_{k-1}\alpha\}$ are
  always on the opposite sides of $0$. The points $\{q_{k,\ell}\alpha\}$ with $0 < \ell \leq a_k$ always lie between
  the points $\{q_{k-2}\alpha\}$ and $\{q_k \alpha\}$; see \eqref{eq:distance_difference}.

  We measure the shortest distance to $0$ on $\T$ by setting
  \begin{align*}
    \|x\| = \min\{\{x\}, 1 - \{x\}\}. 
  \end{align*}
  We have the following facts for $k \geq 2$ and for all $l$ such that $0 < l \leq a_k$:
  \begin{align}\label{eq:distance_formula}
    \|q_{k,\ell} \alpha\| &= (-1)^k(q_{k,\ell} \alpha - p_{k,\ell}),\\
    \|q_{k,\ell}\alpha\|  &= \|q_{k,\ell-1}\alpha\| - \|q_{k-1}\alpha\|. \label{eq:distance_difference}
  \end{align}
  We can now interpret \autoref{prp:best_approximation_convergent} as
  \begin{align}\label{eq:min_distance}
    \min_{0 < n < q_k} \|n\alpha\| = \|q_{k-1}\alpha\|, \quad \text{for } k \geq 1.
  \end{align}
  Note that rotating preserves distances; a fact we will often use without explicit mention. In particular, the
  distance between the points $\{n\alpha\}$ and $\{m\alpha\}$ is $\||n-m|\alpha\|$. Thus by \eqref{eq:min_distance} the
  minimum distance between the distinct points $\{n\alpha\}$ and $\{m\alpha\}$ with $0 \leq n,m < q_k$ is at least
  $\|q_{k-1}\alpha\|$. Formula \eqref{eq:min_distance} tells what is the point closest to $0$ among the points
  $\{n\alpha\}$ for $1 \leq n \leq q_k - 1$. We are also interested in knowing the point closest to $0$ on the side
  opposite to $\{q_{k-1}\alpha\}$. The next result is very important and concerns this; see
  \cite[Proposition~2.2.]{2015:characterization_of_repetitions_in_sturmian_words_a_new}.

  \begin{proposition}\label{prp:closest}
    Let $\alpha$ be an irrational number. Let $n$ be an integer such that $0 < n < q_{k,\ell}$ with $k \geq 2$ and
    $0 < \ell \leq a_k$. If $\|n\alpha\| < \|q_{k,\ell-1}\alpha\|$, then $n = mq_{k-1}$ for some integer $m$ such that
    $1 \leq m \leq \min\{\ell, a_k - \ell + 1\}$.
  \end{proposition}

  \subsection{Sturmian Words}\label{ssec:sturmian_words}
  Sturmian words are a well-known class of infinite, aperiodic binary words with minimal factor complexity. They are
  defined as the infinite words having $n+1$ factors of length $n$ for every $n \geq 0$. For our purposes it is more
  convenient to view Sturmian words as the infinite words obtained as codings of orbits of points in an irrational
  circle rotation with two intervals; see
  \cite{2002:substitutions_in_dynamics_arithmetics_and_combinatorics,2002:algebraic_combinatorics_on_words}.
  Let us make this more precise. The frequency $\alpha$ of letter $1$ (called the \emph{slope}) in a Sturmian words
  exists, and it is irrational. Divide the circle $\T$ into two intervals $I_0$ and $I_1$ defined by the points $0$ and
  $1-\alpha$, and define the coding function $\nu$ by setting $\nu(x) = 0$ if $x \in I_0$ and $\nu(x) = 1$ if
  $x \in I_1$. The coding of the orbit of a point $x$ is the infinite word $s_{x,\alpha}$ obtained by setting its
  $n^\text{th}, n \geq 0,$ letter to equal $\nu(R^n(x))$ where $R$ is the rotation by angle $\alpha$. This word is
  Sturmian with slope $\alpha$, and conversely every Sturmian word with slope $\alpha$ is obtained this way. To make the
  definition proper, we need to define how $\nu$ behaves in the endpoints $0$ and $1-\alpha$. We have two options:
  either take $I_0 = [0,1-\alpha)$ and $I_1 = [1-\alpha,1)$ or $I_0 = (0,1-\alpha]$ and $I_1 = (1-\alpha,1]$. The
  difference is seen in the codings of the orbits of the special points $\{-n\alpha\}$, and both options are needed to
  be able to obtain every Sturmian word of slope $\alpha$ as a coding of a rotation. However, in this paper we are not
  concerned about this choice. We make the convention that $I(x, y)$ with $x \neq y$ and $x, y \neq 0$ is either of the half-open
  intervals of $\T$ separated by the points $x$ and $y$ (taken modulo $1$ if necessary) not containing the point $0$ as
  an interior point. The interval $I(x,0) = I(0,x)$ is either of the half-open intervals separated by the points $0$
  and $x$ having smallest length (the case $x = \frac{1}{2}$ is not important in this paper). Since the sequence
  $(\{n\alpha\})_{n\geq 0}$ is dense in $[0,1)$---as is well-known---every Sturmian word of slope $\alpha$ has the same
  language (that is, the set of factors); this language is denoted by $\Lang(\alpha)$. Further, all Sturmian words are
  uniformly recurrent.

  For every factor $w = a_0 a_1 \cdots a_{n-1}$ of length $n$ there exists a unique subinterval $[w]$ of $\T$ such that
  $s_{x,\alpha}$ begins with $w$ if and only if $x \in [w]$. Clearly
  \begin{align*}
    [w] = I_{a_0} \cap R^{-1}(I_{a_1}) \cap \ldots \cap R^{-(n-1)}(I_{a_{n-1}}).
  \end{align*}
  We denote the length of the interval $[w]$ by $|[w]|$. The points
  $0, \{-\alpha\}, \{-2\alpha\}, \ldots, \{-n\alpha\}$ partition the circle into $n+1$ intervals, which have one-to-one
  correspondence with the words of $\Lang(\alpha)$ of length $n$. Among these intervals the interval containing the
  point $\{-(n+1)\alpha\}$ corresponds to the right special factor of length $n$. A factor $w$ is \emph{right special}
  if both $w0, w1 \in \Lang(\alpha)$. Similarly a factor is \emph{left special} if both $0w, 1w \in \Lang(\alpha)$. In
  a Sturmian word there exists a unique right special and a unique left special factor of length $n$ for all
  $n \geq 0$. The language $\Lang(\alpha)$ is mirror-invariant, that is, for every $w \in \Lang(\alpha)$ also
  $\mirror{w} \in \Lang(\alpha)$. It follows that the right special factor of length $n$ is the reversal of the left
  special factor of length $n$. Sturmian words are also \emph{balanced}; that is, the number of occurrences of the
  letter $1$ in any two factors of the same length differ at most by $1$.

  Given the continued fraction expansion of an irrational $\alpha \in (0,1)$ as in \eqref{eq:cf}, we define the
  corresponding \emph{standard sequence} $(s_k)_{k\geq 0}$ of words by
  \begin{alignat*}{5}
    s_{-1} = 1, \qquad& s_0 = 0, \qquad& s_1 = s_0^{a_1-1}s_{-1}, \qquad& s_k = s_{k-1}^{a_k}s_{k-2}, \qquad&k \geq 2.
  \end{alignat*}
  As $s_k$ is a prefix of $s_{k+1}$ for $k \geq 1$, the sequence $(s_k)$ converges to a unique infinite word $c_\alpha$
  called the infinite standard Sturmian word of slope $\alpha$, and it equals $s_{\alpha,\alpha}$. Inspired by the
  notion of semiconvergents, we define \emph{semistandard} words for $k \geq 2$ by
  \begin{align*}
    s_{k,\ell} = s_{k-1}^\ell s_{k-2}
  \end{align*}
  with $1 \leq \ell < a_k$. Clearly $|s_k| = q_k$ and $|s_{k,\ell}| = q_{k,\ell}$. Instead of writing ``standard or
  semistandard'', we often simply write ``(semi)standard''. The set of standard words of slope $\alpha$ is denoted by
  $\st{\alpha}$, and the set of standard \emph{and} semistandard words of slope $\alpha$ is denoted by $\sst{\alpha}$.
  (Semi)standard words are left special as prefixes of the word $c_\alpha$. Every (semi)standard word is primitive
  \cite[Proposition~2.2.3]{2002:algebraic_combinatorics_on_words}. An important property of standard words is that the
  words $s_k$ and $s_{k-1}$ almost commute; namely $s_k s_{k-1} = wxy$ and $s_{k-1} s_k = wyx$ for some word $w$ and
  distinct letters $x$ and $y$. For more on standard words see
  \cite{2002:algebraic_combinatorics_on_words,1999:on_the_index_of_sturmian_words}.
  
  The only difference between the words $c_\alpha$ and $c_{\overline{\alpha}}$ where
  $\alpha = [0; 1, a_2, a_3, \ldots]$ and $\overline{\alpha} = [0; a_2 + 1, a_3, \ldots]$ is that the roles of the
  letters $0$ and $1$ are reversed. We may thus assume without loss of generality that $a_1 \geq 2$. For the rest of
  this paper we make the convention that $\alpha$ stands for an irrational number in $(0,1)$ having the continued
  fraction expansion as in \eqref{eq:cf} with $a_1 \geq 2$, i.e., we assume that $0 < \alpha < \frac12$. The numbers
  $q_k$ and $q_{k,\ell}$ refer to the denominators of the convergents of $\alpha$, and the words $s_k$ and $s_{k,\ell}$
  refer to the standard or semistandard words of slope $\alpha$.

  \subsection{Powers in Sturmian Words}
  In this section we review some known results on powers in Sturmian words, and prove helpful results for the next
  section.

  If a square $w^2$ occurs in a Sturmian word of slope $\alpha$, then the length of the word $w$ must be a really
  specific number, namely a denominator of a convergent or a semiconvergent of $\alpha$. The proof can be found in
  \cite[Theorem~1]{2003:powers_in_sturmian_sequences} or
  \cite[Proposition~4.1]{2015:characterization_of_repetitions_in_sturmian_words_a_new}.

  \begin{proposition}\label{prp:square_length}
    If $w^2 \in \Lang(\alpha)$ with $w$ nonempty and primitive, then $|w| = q_0$, $|w| = q_1$ or $|w| = q_{k,\ell}$ for
    some $k \geq 2$ with $0 < \ell \leq a_k$.
  \end{proposition}

  Next we need to know when conjugates of (semi)standard words occur as squares in a Sturmian word.

  \begin{proposition}\label{prp:conjugate_square}
    The following holds:
    \begin{enumerate}[(i)]
      \item A factor $w \in \Lang(\alpha)$ is conjugate to $s_k$ for some $k \geq 0$ if and only if $|w| = |s_k|$ and $w^2 \in \Lang(\alpha)$.
      \item Let $w$ be a conjugate of $s_{k,\ell}$ with $k \geq 2$ and $0 < \ell < a_k$. Then $w^2 \in \Lang(\alpha)$ if
            and only if the intervals $[w]$ and $[s_{k,\ell}]$ have the same length.
      \item Let $n = q_0$, $n = q_1$, or $n = q_{k,\ell}$ with $k \geq 2$ and $0 < \ell \leq a_k$, and let $s$ be the
            (semi)standard word of length $n$. A factor $w \in \Lang(\alpha)$ of length $n$ is conjugate to $s$ if and
            only if $w$ and $s$ have equally many occurrences of the letter $0$.
    \end{enumerate}
  \end{proposition}
  \begin{proof}
    Claim \emph{(i)} is a direct consequence of \cite[Theorem~3]{2003:powers_in_sturmian_sequences} or alternatively
    \cite[Theorem~4.5]{2015:characterization_of_repetitions_in_sturmian_words_a_new}. Claim \emph{(ii)} can be inferred
    from Theorems 4.3 and 4.5 of \cite{2015:characterization_of_repetitions_in_sturmian_words_a_new}. Finally, claim
    \emph{(iii)} is evident from the proof of
    \cite[Theorem~4.3]{2015:characterization_of_repetitions_in_sturmian_words_a_new}, but a short proof can be given:
    the idea is that every factor of length $n$ except one exceptional factor $v$ is conjugate to $s$ since $s^2$
    occurs in $\Lang(\alpha)$ by \emph{(i)} and $\emph{(ii)}$. As not every factor of length $n$ may have the same
    number of letters $0$ (a right special factor always extends to two factors having different number of letters
    $0$), it must be that $v$ has a different number of letters $0$ than any conjugate of $s$.
  \end{proof}

  We also need to know the index of certain factors of Sturmian words. The following proposition follows directly from
  Theorems 3 and 4 of \cite{2003:powers_in_sturmian_sequences} or from
  \cite[Theorem~4.5]{2015:characterization_of_repetitions_in_sturmian_words_a_new}.

  \begin{proposition}\label{prp:standard_index}
    The index of the standard word $s_k$ in $\Lang(\alpha)$ is $a_{k+1} + 2$ for $k \geq 2$ and $a_2+1$ for $k = 1$.
    The index of the semistandard word $s_{k,\ell}$ in $\Lang(\alpha)$ with $k \geq 2$ and $0 < \ell < a_k$ is $2$.
  \end{proposition}

  \section{The Square Root Map}\label{sec:square_root_map}
  In \cite{2010:everywhere_alpha-repetitive_sequences_and_sturmian_words} Saari observed that every Sturmian word with
  slope $\alpha = [0;a_1,a_2,\ldots]$ is an optimal squareful word with parameters $\oa = a_1 - 1$ and $\ob = a_2 - 1$.
  The assumption $0 < \alpha < \frac12$ implies that $0^2 \in \Lang(\alpha)$, so the six minimal squares in
  $\Lang(\alpha)$ are the same as in \eqref{eq:min_squares}. In particular, Saari's result means that every Sturmian
  word can be (uniquely) written as a product of the six \emph{minimal squares of slope $\alpha$}
  \eqref{eq:min_squares}. Thus the square root map introduced next is well-defined.

  \begin{definition}
    Let $s$ be a Sturmian word with slope $\alpha$ and factorize it as a product of minimal squares
    $s = X_1^2 X_2^2 X_3^2 \cdots.$ The \emph{square root of $s$} is then defined to be the word
    $\sqrt{s} = X_1 X_2 X_3 \cdots$.
  \end{definition}

  Let us consider as an example the famous Fibonacci word $f$. The Fibonacci word is a Sturmian word of slope
  $[0;2,1,1,\ldots]$, so it has parameters $\oa = 1$ and $\ob = 0$. It is also the fixed point of the substitution
  $0\mapsto 01, 1\mapsto 0$. For more information, see for instance \cite{2002:algebraic_combinatorics_on_words}. We
  have that
  \begin{align*}
    f        &= (010)^2 (100)^2 (10)^2 (01)^2 0^2 (10010)^2 (01)^2 \cdots \, \text{ and } \\
    \sqrt{f} &= 010\cdot100\cdot10\cdot01\cdot0\cdot10010\cdot01\cdots.
  \end{align*}

  Note that a square root map can be defined for any optimal squareful word. However, now we only focus on Sturmian
  words; we study later the square root map for other optimal squareful words in \autoref{sec:counter_example}.

  We aim to prove the surprising fact that given a Sturmian word $s$ the word $\sqrt{s}$ is also a Sturmian word
  having the same slope as $s$. Moreover, knowing the intercept of $s$, we can compute the intercept of $\sqrt{s}$.

  In the proof we need a special function $\psi: \T \to \T$ defined as follows. For $x \in (0,1)$ we set
  \begin{align*}
    \psi(x) = \frac12 (x+1-\alpha),
  \end{align*}
  and we set
  \begin{align*}
    \psi(0) = \begin{cases}
                \frac12 (1-\alpha), &\text{ if } I_0 = [0,1-\alpha), \\
                1-\frac\alpha 2,    &\text{ if } I_0 = (0,1-\alpha].
              \end{cases}
  \end{align*}
  The mapping $\psi$ moves a point $x$ on the circle $\T$ towards the point $1-\alpha$ by halving the distance between
  the points $x$ and $1-\alpha$. The distance to $1-\alpha$ is measured in the interval $I_0$ or $I_1$ depending on
  which of these intervals the point $x$ belongs to.

  We can now state the result.

  \begin{theorem}\label{thm:square_root}
    Let $s_{x,\alpha}$ be a Sturmian word with slope $\alpha$ and intercept $x$. Then
    $\sqrt{s_{x,\alpha}} = s_{\psi(x),\alpha}$. In particular, $\sqrt{s_{x,\alpha}}$ is a Sturmian word with slope
    $\alpha$.
  \end{theorem}

  For a combinatorial version of the above theorem see \autoref{thm:combinatorial_version} in
  \autoref{sec:combinatorial_version}.

  The main idea of the proof is to demonstrate that the square root map is actually the symbolic counterpart of the
  function $\psi$. We begin with a definition.

  \begin{definition}
    A square $w^2 \in \Lang(\alpha)$ satisfies the \emph{square root condition} if $\psi([w^2]) \subseteq [w]$.
  \end{definition}

  Note that if the interval $[w]$ in the above definition has $1-\alpha$ as an endpoint, then $w$ automatically
  satisfies the square root condition. This is because $\psi$ moves points towards the point $1-\alpha$ but does not
  map them over this point. Actually, if $w$ satisfies the square root condition, then necessarily the interval $[w]$
  has $1-\alpha$ as an endpoint (see \autoref{cor:interval_endpoint}).

  We will only sketch the proof of the following lemma.

  \begin{figure}
  \centering
  \begin{tikzpicture}
    \draw (0,0) circle (1.5);
    \filldraw[red] (0:1.5) circle(0.8pt);
    \filldraw[black] (85:1.5) circle(0.8pt);
    \filldraw[black] (170:1.5) circle(0.8pt);
    \filldraw[red] (222.5:1.5) circle(0.8pt);
    \filldraw[black] (254.95:1.5) circle(0.8pt);
    \filldraw[black] (307.5:1.5) circle(0.8pt);
    \draw (85:1.3) arc (85:307.5:1.3);
    \filldraw[black] (85.5:1.3) circle(0.8pt);
    \filldraw[black] (222.5:1.3) circle(0.8pt);
    \filldraw[black] (307.5:1.3) circle(0.8pt);
    \node at (1.8,0) {\footnotesize{$0$}};
    \node at (-1.6,-1.1) {\footnotesize{$1-\alpha$}};
    \node at (1.5,1.2) {\footnotesize{$[S_1^2]$}};
    \node at (-1.4,1.3) {\footnotesize{$[S_3^2]$}};
    \node at (-1.8,-0.45) {\footnotesize{$[S_2^2]$}};
    \node at (-1.0,-1.6) {\footnotesize{$[S_5^2]$}};
    \node at (0.55,-1.7) {\footnotesize{$[S_6^2]$}};
    \node at (1.7,-0.8) {\footnotesize{$[S_4^2]$}};
    \node at (-0.65,0.7) {\footnotesize{$[S_2]$}};
    \node at (-0.1,-1.0) {\footnotesize{$[S_5]$}};
  \end{tikzpicture}
  \caption{The positions of the intervals on the circle in the proof sketch of
           \autoref{lem:min_square_root_formula}.}
  \label{fig:locations}
  \end{figure}
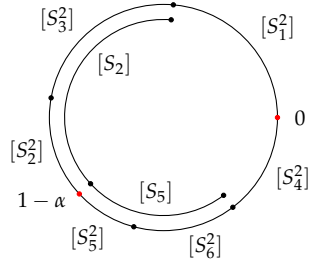


  \begin{lemma}\label{lem:min_square_root_formula}
    For every $i \in \{1,\ldots,6\}$ the minimal square root $S_i$ of slope $\alpha$ satisfies the square root
    condition and $\psi(\{x+2|S_i|\alpha\}) = \{\psi(x) + |S_i|\alpha\}$ for all $x \in [S_i^2]$.
  \end{lemma}
  \begin{proof}[Proof Sketch]
    It is straightforward to verify that
    \begin{alignat*}{2}
      &[S_1] = I(0,1-\alpha),               && [S_4] = I(1-\alpha,1), \\
      &[S_2] = I(-2\alpha,1-\alpha),        && [S_5] = I(1-\alpha,-q_{2,1}\alpha), \\
      &[S_3] = I(-2\alpha,1-\alpha), \quad  && [S_6] = I(1-\alpha,-q_{2,1}\alpha)
    \end{alignat*}
    and
    \begin{alignat*}{2}
      &[S_1^2] = I(0,-2\alpha),                        && [S_4^2] = I(-q_{2,1}\alpha,1), \\
      &[S_2^2] = I(-(q_{2,1}+1)\alpha,1-\alpha),       && [S_5^2] = I(1-\alpha,-(q_{3,1}+1)\alpha), \\
      &[S_3^2] = I(-2\alpha,-(q_{2,1}+1)\alpha), \quad && [S_6^2] = I(-(q_{3,1}+1)\alpha,-q_{2,1}\alpha),
    \end{alignat*}
    see \autoref{fig:locations}. Since $\psi$ does not map points over the point $1-\alpha$, it is evident that every
    minimal square root satisfies the square root condition.

    Consider then the latter claim. Let $i \in \{1,\ldots,6\}$. Suppose that $x \in [S_i^2]\setminus\{0\}$,
    $\{x+2|S_i|\alpha\} \neq 0$, and $\lfloor x+2|S_i|\alpha \rfloor = 2r$ for some $r \geq 0$. Then
    \begin{align}\label{eq:psi_formula}
      \psi(\{x+2|S_i|\alpha\}) = \frac12 (x+2|S_i|\alpha - 2r + 1 - \alpha) = \psi(x)+|S_i|\alpha-r = \{\psi(x)+|S_i|\alpha\}
    \end{align}
    since $\psi$ is a function from $\T$ to $\T$. We consider next the cases $i = 1$ and $i = 5$; the other cases are
    similar.

    Suppose that $S_i = S_1$. Now $x+2\alpha \geq 2\alpha > 0 = 2p_0$ and
    $x+2\alpha \leq 1-2\alpha + 2\alpha = 1 = 2p_0 + 1$, so $x+2\alpha \in (2p_0, 2p_0+1]$. The claim is thus clear
    as in \eqref{eq:psi_formula} if $x \neq 0$ and $x \neq 1-2\alpha$. If $x = 0$, then $I_0 = [0,1-\alpha)$
    and $\{\psi(x)+\alpha\} = \{\frac12 (1-\alpha) + \alpha\} = \frac12 (1+\alpha) = \psi(\{x+2\alpha\})$. If
    $x = 1-2\alpha$, then $I_0 = (0,1-\alpha]$ and $\psi(\{x+2\alpha\}) = 1-\frac\alpha 2 = \{\psi(x)+\alpha\}$.

    Assume then that $S_i = S_5$. Note that $|S_5| = q_2$. Using \eqref{eq:distance_difference} we obtain that
    \begin{align*}
      x+2q_2\alpha &\leq \|(q_{3,1}+1)\alpha\| + 2q_2\alpha \\
                   &= 1-\alpha+\|q_{3,1}\alpha\|+2p_2+2\|q_2\alpha\| \\
                   &= 1-\alpha+\|q_1\alpha\|-\|q_2\alpha\|+2p_2+2\|q_2\alpha\| \\
                   &= 1-\alpha+\|q_1\alpha\|+\|q_2\alpha\|+2p_2 \\
                   &\leq 2p_2+1,
    \end{align*}
    where equality holds only if $x = \|(q_{3,1}+1)\alpha\|$ and $a_2 = 1$. The length of the interval $[S_5^2]$ is
    $\|q_{3,1}\alpha\|$. Since $1-\alpha \geq \alpha + \|q_1\alpha\|$ and $\alpha > \|q_1\alpha\| > \|q_2\alpha\|$, it
    follows from the preceding inequalities that $x+2q_2\alpha > 2p_2$. Therefore $x+2q_2\alpha \in (2p_2, 2p_2+1]$. If
    $a_2 > 1$ or $x \neq \|(q_{3,1}+1)\alpha\|$, then the conclusion follows as in \eqref{eq:psi_formula}. Suppose
    finally that $a_2 = 1$ and $x = \|(q_{3,1}+1)\alpha\|$. Now $I_0 = (0,1-\alpha]$, so
    $\psi(\{x+2q_2\alpha\}) = \psi(0) = 1-\frac\alpha 2$. On the other hand,
    \begin{align*}
      \psi(x)+q_2\alpha &= \frac12 (1-\alpha+\|q_{3,1}\alpha\|+1-\alpha) + p_2 + \|q_2\alpha\| \\
                        &= \frac12 (1-\alpha+\|q_1\alpha\|-\|q_2\alpha\|+1-\alpha+2\|q_2\alpha\|)+p_2 \\
                        &= 1-\frac\alpha 2 + p_2,
    \end{align*}
    so the conclusion holds also in this case.
  \end{proof}

  \begin{proof}[Proof of \autoref{thm:square_root}]
    Write $s_{x,\alpha} = X_1^2 X_2^2 X_3^2 \cdots$ as a product of minimal squares. Since the minimal square $X_1^2$
    satisfies the square root condition by \autoref{lem:min_square_root_formula}, we have that $\psi(x) \in [X_1]$.
    Hence both $\sqrt{s_{x,\alpha}}$ and $s_{\psi(x),\alpha}$ begin with $X_1$. \autoref{lem:min_square_root_formula}
    implies that $\psi(\{x + 2|X_1|\alpha\}) = \{\psi(x) + |X_1|\alpha\}$ for all $x \in [X_1^2]$. Thus by shifting
    $s_{x,\alpha}$ the amount $2|X_1|$ and by applying the preceding reasoning, we conclude that $s_{\psi(x),\alpha}$
    shifted by the amount $|X_1|$ begins with $X_2$. Therefore the words $\sqrt{s_{x,\alpha}}$ and $s_{\psi(x),\alpha}$
    agree on their first $|X_1| + |X_2|$ letters. By repeating this procedure, we conclude that
    $\sqrt{s_{x,\alpha}} = s_{\psi(x),\alpha}$.
  \end{proof}

  \autoref{thm:square_root} allows us to effortlessly characterize the Sturmian words which are fixed points of the
  square root map.

  \begin{corollary}\label{cor:fixed_points}
    The only Sturmian words of slope $\alpha$ which are fixed by the square root map are the two words $01c_\alpha$ and
    $10c_\alpha$, both having intercept $1-\alpha$.
  \end{corollary}
  \begin{proof}
    The only fixed point of the map $\psi$ is the point $1-\alpha$. Having this point as an intercept, we obtain two
    Sturmian words: either $01c_\alpha$ or $10c_\alpha$, depending on which of the intervals $I_0$ and $I_1$ the point
    $1-\alpha$ belongs to.
  \end{proof}

  The set $\{01c_\alpha, 10c_\alpha\}$ is not only the set of fixed points but also the unique attractor of the square
  root map in the set of Sturmian words of slope $\alpha$. When iterating the square root map on a fixed Sturmian word
  $s_{x,\alpha}$, the obtained word has longer and longer prefixes in common with either of the words $01c_\alpha$ and
  $10c_\alpha$ because $\psi^n(x)$ tends to $1-\alpha$ as $n$ increases.

  \section{One Characterization of Words Satisfying the Square Root Condition}\label{sec:square_root_condition}
  In the previous section we saw that the minimal squares, which satisfy the square root condition, were crucial in
  proving that the square root of a Sturmian word is again Sturmian with the same slope. The minimal squares of slope
  $\alpha$ are not the only squares in $\Lang(\alpha)$ satisfying the square root condition; in this section we will
  characterize combinatorially such squares. To be able to state the characterization, we need to define
  \begin{align*}
    \rst{\alpha} = \{\mirror{w}\colon w \in \st{\alpha}\},
  \end{align*}
  the set of reversed standard words of slope $\alpha$. Similarly we set
  \begin{align*}
    \rsst{\alpha} = \{\mirror{w}\colon w \in \sst{\alpha}\}.
  \end{align*}
  We also need the operation $L$ which exchanges the first two letters of a word (we do not apply this operation to too
  short words).

  The main result of this section is the following.

  \begin{theorem}\label{thm:square_root_condition_iff_rsst}
    A square $w^2 \in \Lang(\alpha)$ satisfies the square root condition if and only if
    $w \in \rsst{\alpha} \cup L(\rst{\alpha})$.
  \end{theorem}

  As we remarked in \autoref{sec:square_root_map}, a square $w^2 \in \Lang(\alpha)$ trivially satisfies the square
  root condition if its interval $[w]$ has $1-\alpha$ as an endpoint. Our aim is to prove that the converse is also
  true. We begin with a technical lemma.

  \begin{lemma}\label{lem:maps_out}
    Let $n = q_1$ or $n = q_{k,l}$ for some $k \geq 2$ with $0 < l \leq a_k$, and let $i$ be an integer such that
    $1 < i \leq n$.
    \begin{enumerate}[(i)]
      \item If $\{-i\alpha\} \in I_0$ and $\{-(i+n)\alpha\} < \{-i\alpha\}$, then
            $\psi(-(i+n)\alpha) > \{-i\alpha\}$.
      \item If $\{-i\alpha\} \in I_1$ and $\{-(i+n)\alpha\} > \{-i\alpha\}$, then $\psi(-(i+n)\alpha) < \{-i\alpha\}$.
    \end{enumerate}
  \end{lemma}
  \begin{proof}
    We prove \emph{(i)}, the second assertion is symmetric. Suppose $\{-i\alpha\} \in I_0$
    and $\{-(i+n)\alpha\} <\allowbreak \{-i\alpha\}$. Note that the distance between the points $\{-i\alpha\}$ and
    $\{-(i+n)\alpha\}$ is less than $\alpha$. It follows that $\{-n\alpha\} \in I_1$. Assume on the contrary that 
    $\psi(-(i+n)\alpha) \leq \{-i\alpha\}$, that is,
    \begin{align*}
      \{-(i+n)\alpha\} + \frac{1}{2}(\{1-\alpha\} - \{-(i+n)\alpha\}) \leq \{-i\alpha\}.
    \end{align*}
    Since $0 < \{-(i+n)\alpha\} < \{-i\alpha\}$, the distance between $\{-(i+n)\alpha\}$ and $\{-i\alpha\}$ is the same
    as the distance between $1$ and $\{-n\alpha\}$. Thus by substituting
    $\{-(i+n)\alpha\} = \{-i\alpha\} - (1-\{-n\alpha\})$ to the above and rearranging, we have that
    \begin{align*}
      \{1-\alpha\} - \{-i\alpha\} \leq 1 - \{-n\alpha\}.
    \end{align*}
    Since $\{-n\alpha\} \in I_1$, we obtain that
    \begin{align}\label{eq:maps_over_distance}
      \|-(i-1)\alpha\| \leq \|-n\alpha\|.
    \end{align}
    Suppose now first that $n = q_{k,l}$ for some $k \geq 2$ and $0 < l \leq a_k$. Since $i - 1 < n$,
    \autoref{prp:closest} and \eqref{eq:maps_over_distance} imply that $i-1 = mq_{k-1}$ for some
    $1 \leq m \leq \min\{l, a_k - l + 1\}$. As $\{-n\alpha\} \in I_1$, the point $\{-q_{k-1}\alpha\}$ must lie on the
    opposite side of $0$ in the interval $I_0$. Therefore $\{-(i-1)\alpha\} \in I_0$. Then by
    \eqref{eq:maps_over_distance}, the point $\{-i\alpha\}$ must lie in $I_1$. This is a contradiction. Suppose then
    that $n = q_1$. It is easy to see that \eqref{eq:maps_over_distance} cannot hold for any $i$ greater than $1$.
    This concludes the proof.
  \end{proof}


  \begin{corollary}\label{cor:interval_endpoint}
    If $w^2 \in \Lang(\alpha)$ with $w$ primitive satisfies the square root condition, then the interval $[w]$ has
    $1-\alpha$ as an endpoint.
  \end{corollary}
  \begin{proof}
    Let $n = |w|$. \autoref{prp:square_length} implies that $n = q_0$, $n = q_1$, or $n = q_{k,l}$ for some $k \geq 2$
    with $0 < l \leq a_k$. Say $n = q_0 = 1$. As the only factor of length $1$ occurring as a square is $0$, the claim
    holds as $[0] = I_0 = I(0, 1-\alpha)$. Suppose then that $n = q_1$ or $n = q_{k,l}$.

    Let $[w] = I(-i\alpha, -j\alpha)$. Then either $[w^2] = I(-i\alpha, -(j + |w|)\alpha)$ or
    $[w^2] = I(-(i + |w|)\alpha, -j\alpha)$. Suppose first that $[w] \subseteq I_0$. By symmetry we may assume that
    $\{-j\alpha\} > \{-i\alpha\}$. Now $[w^2] = [-(i+|w|)\alpha, -j\alpha)$ if and only if $j = 1$. Namely, if
    $j \neq 1$, then it is clear that it is possible to find a point $x \in I(-i\alpha, -j\alpha)$ close to
    $\{-j\alpha\}$ such that $\psi(x) > \{-j\alpha\}$, so the condition $\psi([w^2]) \subseteq [w]$ cannot be
    satisfied. If $[w^2] = [-i\alpha, -(j+|w|)\alpha)$ and $j \neq 1$, then by \autoref{lem:maps_out}
    $\psi(-(j+|w|)\alpha) > \{-j\alpha\}$, so the condition $\psi([w^2]) \subseteq [w]$ cannot be satisfied.
    Thus also in this case necessarily $j = 1$. The case where $[w] \subseteq I_1$ is proven symmetrically
    using the latter symmetric assertion of \autoref{lem:maps_out}.
  \end{proof}

  Next we study in more detail the properties of squares $w^2 \in \Lang(\alpha)$ whose interval has $1-\alpha$ as an
  endpoint.

  \begin{proposition}\label{prp:endpoint_alpha_characterization}
    Consider the intervals of factors in $\Lang(\alpha)$ of length $n = q_1$ or $n = q_{k,l}$ with $k \geq 2$ and
    $0 < l \leq a_k$. Let $u$ and $v$ be the two distinct words of length $n$ having intervals with endpoint
    $1-\alpha$. Then the following holds.
    \begin{enumerate}[(i)]
      \item There exists a word $w$ such that $u = xyw$ and $v = yxw = L(u)$ for distinct letters $x$ and $y$.
      \item Either $u$ or $v$ is right special.
      \item If $\mu$ is the right special word among the words $u$ and $v$, then $\mu^2 \in \Lang(\alpha)$.
      \item If $\lambda$ is the word among the words $u$ and $v$ which is not right special, then
            $\lambda^2 \in \Lang(\alpha)$ if and only if $n = q_1$ or $l = a_k$.
    \end{enumerate}
  \end{proposition}
  \begin{proof}
    Suppose first that $n = q_1$. Then it is straightforward to see that the factors $u$ and $v$ of length $n$ having
    intervals with endpoint $1-\alpha$ are $010^{a_1 - 2} = S_2$ and $10^{a_1 - 1} = S_4$. Clearly $S_4$ is right
    special and $L(S_4) = S_2$. Moreover $S_2^2, S_4^2 \in \Lang(\alpha)$.

    Assume that $n = q_{k,l}$ for some $k \geq 2$ with $0 < l \leq a_k$. By \autoref{prp:closest} the point
    $\{-n\alpha\}$ is the point closest to $0$ on the side opposite to the point $\{-q_{k-1}\alpha\}$. Thus either
    $\{-(n+1)\alpha\} \in [u]$ or $\{-(n+1)\alpha\} \in [v]$. Assume by symmetry that $\{-(n+1)\alpha\} \in [u]$. This
    means that the word $u$ is right special, proving \emph{(ii)}. Further, the endpoint of $[u]$ which is not
    $1-\alpha$ must be after a rotation the next closest point to $0$ on the side opposite to the point
    $\{-q_{k-1}\alpha\}$. Thus by \autoref{prp:closest} $[u] = I(-(q_{k,l-1} + 1)\alpha, 1-\alpha)$ and consequently
    $[v] = I(1-\alpha, -(q_{k-1} + 1)\alpha)$.

    Since the points $x = \{(-(q_{k,l-1}+1)\alpha\}$ and $y = \{-(q_{k-1}+1)\alpha\}$ are on the opposite sides of the
    point $1-\alpha$ and the points $\{x + \alpha\}$ and $\{y + \alpha\}$ are on the opposite sides of the point $0$,
    it follows that $u$ begins with $cd$ and $v$ begins with $dc$ for distinct letters $c$ and $d$. Assume on the
    contrary that $u = cdzeu'$ and $v = dczfv'$ for distinct letters $e$ and $f$. In particular, $|z| \leq n - 3$. This
    means that the point $x' = \{x + (|z|+2)\alpha\}$ is in $[e]$ and the point $y' = \{y + (|z|+2)\alpha\}$ is in
    $[f]$. It must be that $e = c$ and $f = d$ as otherwise the point $x' - \alpha$ would be in $[c]$ and the point
    $y' - \alpha$ would be in $[d]$ contradicting the choice of $z$. Since $\alpha$ is irrational, either $x'$ is
    closer to $1-\alpha$ than $x$ or $y'$ is closer to $1-\alpha$ than $y$.
    
    Suppose that $x'$ is closer to $1-\alpha$ than $x$. Since $x'$ is on the same side of the point $1-\alpha$ as $x$,
    it follows that
    \begin{align*}
      \|x' + \alpha\| = \|(q_{k,l-1} - |z| - 2)\alpha\|< \|q_{k,l-1}\alpha\| = \|x + \alpha\|.
    \end{align*}
    Since $q_{k,l-1} - |z| - 2 < q_{k,l-1}$, by \autoref{prp:closest} it must be that $q_{k,l-1} - |z| - 2 \leq 0$.
    However, as $\|q_{k,l-1}\alpha\| = \|-q_{k,l-1}\alpha\|$, it follows by \autoref{prp:closest} that
    $|z| + 2 - q_{k,l-1} = mq_{k-1}$ for some $m \geq 1$. Thus $|z| + 2 \geq q_{k,l-1} + q_{k-1} = q_{k,l} = n$.
    This is, however, a contradiction as $|z| \leq n - 3$.

    Suppose then that $y'$ is closer to $1-\alpha$ than $y$. Similar to above, it follows that
    \begin{align*}
      \|y' + \alpha\| = \|(q_{k-1} - |z| - 2)\alpha\| < \|q_{k-1}\alpha\| = \|y + \alpha\|.
    \end{align*}
    Again, it must be that $q_{k-1} - |z| - 2 \leq 0$. Since $\|q_{k-1}\alpha\| = \|-q_{k-1}\alpha\|$, it follows from
    \eqref{eq:min_distance} that $|z| + 2 - q_{k-1} \geq q_k$. Therefore $|z| + 2 \geq q_k + q_{k-1} > n$. This is
    again a contradiction with the fact that $|z| \leq n - 3$.

    Thus we conclude that $u = cdw$ and $v = dcw$ for some word $w$ proving \emph{(i)}.
    As $n = q_{k,l}$, it must be that the right special word of length $n$
    equals $\mirror{s}_{k,l}$. Since $u$ and $v$ are conjugate by \autoref{prp:conjugate_square} (iii), \autoref{prp:conjugate_square} implies that if
    $l = a_k$, then $u^2, v^2 \in \Lang(\alpha)$. Suppose that $l \neq a_k$. By \autoref{prp:conjugate_square}, the
    word $s_{k,l}$ occurs as a square in $\Lang(\alpha)$. Since $\Lang(\alpha)$ is mirror-invariant, also
    $u^2 = \mirror{s}_{k,l}^{\,2} \in \Lang(\alpha)$. Therefore from \autoref{prp:conjugate_square} it follows that
    $|[u]| = \|q_{k,l-1}\alpha\| = |[s_{k,l}]|$. Now $[v] = I(1-\alpha, -(q_{k-1} + 1)\alpha)$, so
    $|[v]| = \|q_{k-1}\alpha\| \neq |[u]|$. Thus \autoref{prp:conjugate_square} implies that
    $v^2 \notin \Lang(\alpha)$. Hence \emph{(iii)} and \emph{(iv)} are proved.
  \end{proof}

  \begin{proof}[Proof of \autoref{thm:square_root_condition_iff_rsst}]
    If $|w| = 1$, then clearly $w = 0 = \mirror{s}_0$, so the claim holds. We may thus focus on the case that $|w| > 1$.

    Suppose that $w^2 \in \Lang(\alpha)$ satisfies the square root condition. By \autoref{cor:interval_endpoint} the
    interval $[w]$ has $1-\alpha$ as an endpoint. Moreover, \autoref{prp:square_length} implies that $|w| = q_1$ or
    $|w| = q_{k,l}$ for some $k \geq 2$ with $0 < l \leq a_k$. Thus from \autoref{prp:endpoint_alpha_characterization}
    it follows that $w = \mirror{s}$ or $w = L(\mirror{s}^{\,})$ where $s$ is the (semi)standard word of length $|w|$.
    By \autoref{prp:endpoint_alpha_characterization} we have that $\mirror{s}^{\,2} \in \Lang(\alpha)$. Moreover, by
    \autoref{prp:endpoint_alpha_characterization} we have that $L(\mirror{s}^{\,})^2 \in \Lang(\alpha)$ if and only if
    $|w| = q_k$ for some $k \geq 1$. Thus $w \in \rsst{\alpha} \cup L(\rst{\alpha})$.

    Suppose then that $w \in \rsst{\alpha} \cup L(\rst{\alpha})$. Note first that $L(w)$ has the same number of letters
    $0$ as $w$, so $w$ is conjugate to $L(w)$ by \autoref{prp:conjugate_square}. Thus it follows from
    \autoref{prp:conjugate_square} that $w^2 \in \Lang(\alpha)$. Let $u$ and $v$ be the factors of length $|w|$ having
    endpoint $1-\alpha$. By \autoref{prp:endpoint_alpha_characterization} the word $u$ must be right special and
    $v = L(u)$. Since the right special factor of length $|w|$ is unique, either $w = u$ or $L(w) = u$. Thus the
    interval $[w]$ has $1-\alpha$ as an endpoint. Then clearly $w^2$ satisfies the square root condition.
  \end{proof}

  \section{Characterization by a Word Equation}\label{sec:word_equation_characterization}
  It turns out that the squares of slope $\alpha$ satisfying the square root condition have also a different
  characterization in terms of specific solutions of the word equation
  \begin{align}\label{eq:square}
    X_1^2 X_2^2 \cdots X_n^2 = (X_1 X_2 \cdots X_n)^2
  \end{align}
  in the language $\Lang(\alpha)$. We are interested only in the solutions of \eqref{eq:square} where all words $X_i$
  are \emph{minimal square roots} \eqref{eq:min_squares}, i.e., primitive roots of minimal squares. Thus we give the
  following definition.

  \begin{definition}
    A nonempty word $w$ is a \emph{solution to \eqref{eq:square}} if $w$ can be written as a product of minimal square
    roots $w = X_1 X_2 \cdots X_n$ which satisfy the word equation \eqref{eq:square}. The solution is \emph{trivial} if
    $X_1 = X_2 = \ldots = X_n$ and \emph{primitive} if $w$ is primitive. The word $w$ is a \emph{solution to
    \eqref{eq:square} in $\Lang(\alpha)$} if $w$ is a solution to \eqref{eq:square} and $w^2 \in \Lang(\alpha)$.
  \end{definition}

  All minimal square roots of slope $\alpha$ are trivial solutions to \eqref{eq:square}. One example of a nontrivial
  solution is $w = S_2S_1S_4$ in the language of the Fibonacci word (i.e., in the language of slope
  $[0;2,1,1,\ldots]$) since $w^2 = (01010)^2 = (01)^2 \cdot 0^2 \cdot (10)^2 = S_2^2 S_1^2 S_4^2$. Note that in the
  language of any Sturmian word there are only finitely many trivial solutions as the index of every factor is finite.

  Note that the factorization of a word as product of minimal squares is unique. Indeed, if
  $X_1^2 \cdots X_n^2 = Y_1^2 \cdots Y_m^2$, where the squares $X_i^2$ and $Y_i^2$ are minimal, then either $X_1^2$ is
  a prefix of $Y_1^2$ or vice versa. Therefore by minimality $X_1^2 = Y_1^2$, that is, $X_1 = Y_1$. The uniqueness of
  the factorization follows.

  Our aim is to complete the characterization of \autoref{thm:square_root_condition_iff_rsst} as follows.

  \begin{theorem}\label{thm:complete_characterization}
    Let $w \in \Lang(\alpha)$. The following are equivalent:
    \begin{enumerate}[(i)]
      \item $w$ is a primitive solution to \eqref{eq:square} in $\Lang(\alpha)$,
      \item $w^2$ satisfies the square root condition,
      \item $w \in \rsst{\alpha} \cup L(\rst{\alpha})$.
    \end{enumerate}
  \end{theorem}

  For later use in \autoref{sec:counter_example} we define the language $\Lang(\oa,\ob)$.

  \begin{definition}
    The language $\Lang(\oa,\ob)$ consists of all factors of the infinite words in the language
    \begin{align*}
      (10^{\oa+1}(10^\oa)^\ob + 10^{\oa+1}(10^\oa)^{\ob+1})^\omega = (S_5 + S_6)^\omega.
    \end{align*}
  \end{definition}

  Observe that by \autoref{prp:optimal_squareful_characterization} every factor in $\Lang(\oa,\ob)$ is a factor of some
  optimal squareful word with parameters $\oa$ and $\ob$. Moreover, if $\alpha = [0; \oa+1, \ob+1, \ldots]$, then
  $\Lang(\alpha) \subseteq \Lang(\oa,\ob)$.

  \begin{definition}
    The language $\Pi(\oa,\ob)$ consists of all nonempty words in $\Lang(\oa,\ob)$ which can be written as products of
    the minimal squares \eqref{eq:min_squares}.
  \end{definition}

  Let $w \in \Pi(\oa,\ob)$, that is, $w = X_1^2 \cdots X_n^2$ for minimal square roots $X_i$. Then we can define the
  square root of $w$ by setting $\sqrt{w} = X_1 \cdots X_n$.

  We need two technical lemmas. Their proofs are straightforward case-by-case analysis. The statement of
  \autoref{lem:exchange_squares} has a technical condition for later use in \autoref{sec:counter_example}, which is
  perhaps better understood if the reader first reads the proof of \autoref{lem:product_of_squares} up to the point
  where \autoref{lem:exchange_squares} is invoked.
 
  \begin{lemma}\label{lem:exchange_squares}
    Let $u$ and $v$ be words such that
    \begin{itemize}
      \item $u$ is a nonempty suffix of $S_6$,
      \item $|v| \geq |S_5 S_6|$,
      \item $v$ begins with $xy$ for distinct letters $x$ and $y$,
      \item $uv \in \Lang(\oa,\ob)$ and $L(v) \in \Lang(\oa,\ob)$.
    \end{itemize}
    Suppose there exists a minimal square $X^2$ such that $|X^2| > |u|$ and $X^2$ is a prefix of $uv$ or $uL(v)$. Then
    there exist minimal squares $Y_1^2, \ldots, Y_n^2$ such that $X^2$ and $Y_1^2 \cdots Y_n^2$ are prefixes of $uv$
    and $uL(v)$ of the same length and $X = Y_1 \cdots Y_n$.
  \end{lemma}
  \begin{proof}
    Let $Z^2$ be a minimal square such that $|Z^2| > |u|$ and $Z^2$ is a prefix of $uv$ or $uL(v)$. It is not obvious
    at this point that $Z$ exists but its existence becomes evident as this proof progresses.
    By symmetry we may assume that $Z^2$ is a prefix of $uv$. To prove the claim we consider different cases depending
    on the word $Z$.

    \textbf{Case A.} $Z = S_1 = 0$. Since $u$ is a nonempty suffix of $S_6$ and $|Z^2| > |u|$, it must be that $u = 0$.
    As $v$ begins with $0$, we have that $v$ begins with $01$ by assumption. Since $v \in \Lang(\oa,\ob)$ and
    $|v| \geq |S_6|$, the word $v$ begins with either $010^\oa10^\oa$ or $010^{\oa+1}10^\oa$. In the latter case $L(v)$
    would begin with $10^{\oa+2}1$ contradicting the assumption $L(v) \in \Lang(\oa,\ob)$. Hence $v$ begins with
    $010^\oa 10^\oa$. It follows that $uv$ has $0010^\oa10^\oa$ as a prefix, that is, $uv$ begins with $S_1^2 S_4^2$.
    On the other hand, the word $uL(v)$ has the word $S_3^2 = 010^{\oa+1}10^\oa$ as a prefix. Since $S_3 = S_1 S_4$,
    the conclusion of the claim holds.

    \textbf{Case B.} $Z = S_2 = 010^{\oa-1}$. If $u = 0$, then $v$ has $10^\oa 10^\oa$ as a prefix and, consequently,
    $L(v)$ has $10^{\oa-1}10^\oa$ as a prefix contradicting the fact that $L(v) \in \Lang(\oa,\ob)$. Therefore by
    the assumptions that $u$ is a nonempty suffix of $S_6$ and $|Z^2| > |u|$, it follows that $u = 010^\oa$. Thus $v$
    has $10^\oa$ as a prefix. Using the fact that $L(v) \in \Lang(\oa,\ob)$, we see that $v$ begins with $10^{\oa+1}$
    and $L(v)$ begins with $010^\oa$. Hence $uv$ has $S_2^2 S_1^2$ as a prefix, and $uL(v)$ has $S_3^2$ as a prefix.
    Since $S_2 S_1 = S_3$, we conclude, as in the previous case, that the conclusion holds.

    \textbf{Case C.} $Z = S_3 = 010^\oa$. Using again the fact that $u$ is a suffix of $S_6$ and $|Z^2| > |u|$, we see
    that either $u = 0$ or $u = 010^\oa$. In the first case $v$ begins with $10^{\oa+1}10^\oa$ and $L(v)$ begins with
    $010^\oa 10^\oa$. Hence the word $uL(v)$ has $S_1^2 S_4^2$ as a prefix. As $S_1 S_4 = S_3$, the conclusion follows.
    Let us then consider the other case. Now $L(v)$ begins with $10^{\oa+1}$, so the word $uL(v)$ has $S_2^2 S_1^2$ as
    a prefix. Again, the conclusion follows since $S_2 S_1 = S_3$.

    \textbf{Case D.} $Z = S_4 = 10^\oa$. Now the only option is that $u = 10^\oa$. Using the fact that
    $v \in \Lang(\oa,\ob)$, we see that $v$ cannot begin with $10^\oa1$, so $v$ must have $10^{\oa+1}$ as a prefix.
    Further, since $|v| \geq |S_6|$, it must be that $S_6$ is a prefix of $v$. If $S_6 1$ would be a prefix of $v$,
    then the word $L(v)$ would have the word $(10^\oa)^{\ob+2}1$ as a factor contradicting the fact that
    $L(v) \in \Lang(\oa,\ob)$. Thus $S_6 0$ is a prefix of $v$. Since $v \in \Lang(\oa,\ob)$ and $|v| \geq |S_5S_6|$,
    we have that $S_6 0(10^\oa)^{\ob+1} = S_5^2 10^\oa$ is a prefix of $v$. Consequently, the word $L(v)$ begins with
    $0(10^\oa)^{\ob+1}10^{\oa+1}(10^\oa)^{\ob+1}$, so $uL(v)$ has $S_6^2$ as a prefix. Assume first that $\ob$ is odd.
    It is straightforward to see that in this case
    \begin{align*}
      0(10^\oa)^\ob 10^{\oa+1}(10^\oa)^{\ob+1} = (S_2^2)^{(\ob+1)/2} S_1^2 (S_4^2)^{(\ob+1)/2}.
    \end{align*}
    Thus for the prefix $10^\oa S_5 10^\oa$ of $uv$ we have that
    \begin{align*}
      10^\oa S_5^2 10^\oa = S_4^2 (S_2^2)^{(\ob+1)/2} S_1^2 (S_4^2)^{(\ob+1)/2}.
    \end{align*}
    As $S_6 = S_4 S_2^{(\ob+1)/2} S_1 S_4^{(\ob+1)/2}$, the conclusion follows as before. Assume then that $\ob$ is
    even. It is now easy to show that
    \begin{align*}
      0(10^\oa)^\ob 10^{\oa+1} (10^\oa)^{\ob+1} = (S_2^2)^{\ob/2} S_3^2 (S_4^2)^{\ob/2}.
    \end{align*}
    Therefore
    \begin{align*}
      10^\oa S_5^2 10^\oa = S_4^2 (S_2^2)^{\ob/2} S_3^2 (S_4^2)^{\ob/2}.
    \end{align*}
    Since $S_6 = S_4 S_2^{\ob/2} S_3 S_4^{\ob/2}$, the conclusion again follows.

    \textbf{Case E.} $Z = S_5 = 10^{\oa+1}(10^\oa)^\ob$. Now either $u = 10^\oa$ or $u = 10^{\oa+1}(10^\oa)^{\ob+1}$.
    In the first case $v$ must begin with $0(10^\oa)^\ob10^{\oa+1}(10^\oa)^\ob$. However, this implies that $L(v)$
    begins with $10^{\oa+1}(10^\oa)^{\ob-1}10^{\oa+1}(10^\oa)^\ob$ contradicting the fact that
    $L(v) \in \Lang(\oa,\ob)$. Consider then the latter case where $v$ begins with $0(10^\oa)^\ob$. As
    $L(v) \in \Lang(\oa,\ob)$ and $|v| \geq |S_6|$, it must be that $L(v)$ begins with $10^{\oa+1}(10^\oa)^{\ob+1}$.
    Hence the word $uL(v)$ has $S_6^2$ as a prefix. Since the word $v$ begins with $0(10^\oa)^{\ob+2}$, the word $uv$
    has $S_5^2 S_4^2$ as a prefix. The conclusion follows as $S_5 S_4 = S_6$.

    \textbf{Case F.} $Z = S_6 = 10^{\oa+1}(10^\oa)^{\ob+1}$. Now there are two possibilities: either $u = 10^\oa$ or
    $u = 10^{\oa+1}(10^\oa)^{\ob+1}$. In the first case $v$ begins with $0(10^\oa)^{\ob+1}10^{\oa+1}(10^\oa)^{\ob+1}$,
    so $L(v)$ begins with $10^{\oa+1}(10^\oa)^\ob 10^{\oa+1}(10^\oa)^{\ob+1}$. The word $uL(v)$ has
    $S_4^2 0(10^\oa)^\ob 10^{\oa+1}(10^\oa)^{\ob+1}$ as a prefix. Proceeding as in the Case D depending on the parity
    of $\ob$, we see that the conclusion holds. Consider then the latter case $u = 10^{\oa+1}(10^\oa)^{\ob+1}$. The
    word $v$ must begin with $u$, so $L(v)$ has $0(10^\oa)^{\ob+2}$ as a prefix. Clearly the word $uL(v)$ has
    $S_5^2 S_4^2$ as a prefix. As $S_6 = S_5 S_4$, the conclusion follows.
  \end{proof}

  A more intuitive way of stating \autoref{lem:exchange_squares} is that under the assumptions of the lemma swapping
  two adjacent and distinct letters which do not occur as a prefix of a minimal square affects a product of minimal
  square only locally and does not change its square root.

  \begin{lemma}\label{lem:product_of_squares}
    Let $w$ be a primitive solution to \eqref{eq:square} having the word $S_6 = 10^{\oa+1}(10^\oa)^{\ob+1}$ as a suffix
    such that $w^2, L(w) \in \Lang(\oa,\ob)$. Then $w L(w) \in \Pi(\oa,\ob)$ and $\sqrt{wL(w)} = w$.
  \end{lemma}
  \begin{proof}
    If $w = S_6$, then it is easy to see that $w L(w) = S_5^2 S_4^2$ and $w = S_5 S_4$, so the claim holds. We may thus
    suppose that $S_6$ is a proper suffix of $w$.

    Since $w$ is a solution to \eqref{eq:square}, we have that $w^2 = X_1^2 \cdots X_n^2$ and $w = X_1 \cdots X_n$ for
    some minimal square roots $X_i$. It must be that $n > 1$ as if $n = 1$ then $w = X_1$, and it is not possible for
    $S_6$ to be a proper suffix of $w$. Assume for a contradiction that $X_1 = S_1$. Since $X_1 X_2$ is a prefix of
    $w^2$, it follows that $X_2$ begins with the letter $0$. If $X_2 \neq S_1$, then $X_1 X_2$ begins with $001$ but
    $X_1^2 X_2^2$ begins with $000$, which is impossible. Hence $X_2 = S_1$, and by repeating the argument it follows
    that $X_k = S_1$ for all $k$ such that $1 \leq k \leq n$. Thus $w$ cannot have $S_6$ as a suffix, so we conclude
    that $X_1 \neq S_1$. Hence $w$ always begins with $01$ or $10$.
    
    We show that $|X_1^2| < |w|$. Assume on the contrary that $|X_1^2| \geq |w|$. Since $w$ has the word $S_6$ as a
    suffix, it follows that $S_6$ is a factor of $X_1^2$. It follows that $X_1$ is one of the words $S_5$, $S_6$ or
    $S_3$ (if $\ob = 0$). If $X_1 = S_5$, then $S_6$ occurs in $X_1^2 = 10^{a+1}(10^a)^b 10^{a+1}(10^a)^b$ only as a
    prefix. Thus $w = S_6$ contradicting the fact that $S_6$ is a proper suffix of $w$. If $X_1 = S_6$, then $S_6$
    occurs in $X_1^2 = 10^{\oa+1}(10^\oa)^{\ob+1} 10^{\oa+1}(10^\oa)^{\ob+1}$ as a prefix and as a suffix. Since
    $w \neq S_6$, it must be that $w = X_1^2$ contradicting the primitivity of $w$. Let finally $\ob = 0$ and
    $X_1 = S_3$. Then $S_6$ occurs in $X_1^2 = 010^{\oa+1}10^\oa$ as a suffix. Hence $w = X_1^2$ contradicting again
    the primitivity of $w$.

    Now there exists a maximal $r$ such that $1 \leq r < n$ and $X_1^2 \cdots X_r^2$ is a prefix of $w$. Actually
    $X_1^2 \cdots X_r^2$ is a proper prefix of $w$, as otherwise
    $w^2 = (X_1^2 \cdots X_r^2)^2 = (X_1 \cdots X_r X_1 \cdots X_r)^2$, so $w = (X_1 \cdots X_r)^2$ contradicting the
    primitivity of $w$. Thus when factorizing $w L(w)$ and $w^2$ as products of minimal squares, the first $r$ squares
    are equal. Let $u$ be the nonempty word such that $w = X_1^2 \cdots X_r^2 u$. By the definition of the number $r$,
    we have that $u$ is a proper prefix of $X_{r+1}^2$. Suppose for a contradiction that $|u| > |S_6|$. It follows that
    $u$ has $S_6$ as a proper suffix. This leaves only the possibilities that $X_{r+1}$ is either of the words $S_5$ or
    $S_6$. However, if $X_{r+1} = S_5$, then $S_6$ cannot be a proper suffix of $u$, and if $X_{r+1} = S_6$, then $r$
    is not maximal. We conclude that $|u| \leq |S_6|$.
    
    Next we show that $w$ must satisfy $|w| \geq |S_5S_6|$. Suppose first that $w$ begins with the letter $0$. Then as
    $S_6$ is a proper suffix of $w$ and $w^2 \in \Lang(\oa,\ob)$, it must be that $w$ begins with $0(10^\oa)^{\ob+1}$.
    Suppose that this prefix overlaps with the suffix $S_6$. Then clearly
    $w = 0(10^\oa)^\ob 10^{\oa+1}(10^\oa)^{\ob+1} = (0(10^\oa)^{\ob+1})^2$ contradicting the primitivity of $w$. If the
    prefix $0(10^\oa)^{\ob+1}$ does not overlap with the suffix $S_6$, then $|w| \geq |S_5S_6|$. Assume then that $w$
    begins with the letter $1$. Similar to above, the word $w$ must begin with $10^{\oa+1}(10^\oa)^{\ob+1}$. In this
    case necessarily $|w| \geq |S_5S_6|$.

    Finally, we can apply \autoref{lem:exchange_squares} to the words $u$ and $w$ with $X = X_{r+1}$. We obtain minimal
    squares $Y_1^2, \ldots, Y_m^2$ such that $Y_1^2 \cdots Y_m^2$ is a prefix of $uL(w)$ and
    and $Y_1 \cdots Y_m = X_{r+1} \cdots X_{r+t}$ for some $t \geq 1$. Thus
    \begin{align*}
      w L(w) &= X_1^2 \cdots X_r^2 Y_1^2 \cdots Y_m^2 X_{r+t+1}^2 \cdots X_n^2 \quad \text{and} \\
      w      &= X_1 \cdots X_n = X_1 \cdots X_r Y_1 \cdots Y_m X_{r+t+1} \cdots X_n.
    \end{align*}
    The claim is proved.
  \end{proof}

  \begin{proposition}\label{prp:standard_solutions}
    Let $w \in \rsst{\alpha} \cup L(\rst{\alpha})$. Then the word $w$ is a primitive solution to \eqref{eq:square} in
    $\Lang(\alpha)$.
  \end{proposition}
  \begin{proof}
    Note that \autoref{prp:conjugate_square} implies that $w^2 \in \Lang(\alpha)$. Suppose first that $|w| < |S_6|$
    where $S_6 = \mirror{s}_{3,1} = 10^{\oa+1}(10^\oa)^{\ob+1}$. Clearly the minimal square root $S_1, \ldots, S_5$ are
    solutions to \eqref{eq:square}, so we are left with the case where $w = \mirror{s}_{2,\ell} = 0(10^\oa)^\ell$ for
    some $\ell$ such that $1 < \ell \leq \ob+1$. It is straightforward to see that if $\ell$ is even, then
    \begin{align*}
      w^2 = (S_2^2)^{\ell/2} S_1^2 (S_4^2)^{\ell/2} \ \text{ and } \ w = S_2^{\ell/2} S_1 S_4^{\ell/2}.
    \end{align*}
    If $\ell$ is odd, then
    \begin{align*}
      w^2 = (S_2^2)^{(\ell+1)/2} S_3^2 (S_4^2)^{(\ell+1)/2} \ \text{ and } \ w = S_2^{(\ell+1)/2}S_3 S_4^{(\ell+1)/2}.
    \end{align*}
    Hence $w$ is a solution to \eqref{eq:square}.
    
    We may thus suppose that $|w| \geq |S_6|$, so $w$ has $S_6$ as a suffix. We proceed by induction. Now either
    $w = \mirror{s}_{k,\ell}$ for some $k \geq 3$ with $0 < \ell \leq a_k$ or $L(w) = \mirror{s}_k$ for some
    $k \geq 3$. We assume that the claim holds for every word satisfying the hypotheses which are shorter than $w$.
    Consider first the case $w = \mirror{s}_{k,\ell}$ for some $k \geq 3$ with $0 < \ell \leq a_k$. By the fact that
    $\mirror{s}_{k-1}\mirror{s}_{k-2} = L(\mirror{s}_{k-2})\mirror{s}_{k-1}$ we obtain that
    \begin{align*}
      w^2 = \mirror{s}_{k-2}\mirror{s}_{k-1}^{\,\ell} \mirror{s}_{k-2}\mirror{s}_{k-1}^{\,\ell}
      = \mirror{s}_{k-2}\mirror{s}_{k-1}^{\,\ell-1} L(\mirror{s}_{k-2}) \mirror{s}_{k-1}^{\,\ell-1} \cdot \mirror{s}_{k-1}^{\,2}
      = \mirror{s}_{k,\ell-1} L(\mirror{s}_{k,\ell-1}) \cdot \mirror{s}_{k-1}^{\,2}.
    \end{align*}
    Now if $k = 3$ and $\ell = 1$, then the conclusion holds as $\mirror{s}_{3,1} = S_6$ is a minimal square root.
    Hence we may assume that either $k > 3$ or $k = 3$ and $\ell > 1$. Since $\mirror{s}_{k-1}$ is a solution to
    \eqref{eq:square}, we have that $\mirror{s}_{k-1}^{\,2} = X_1^2 \cdots X_n^2$ and $\mirror{s}_{k-1} = X_1 \cdots
    X_n$ for some minimal square roots $X_i$. In other words,
    \begin{align*}
      \mirror{s}_{k-1}^{\,2} \in \Pi(\oa,\ob) \ \text{ and } \ \sqrt{\mirror{s}_{k-1}^{\,2}} = \mirror{s}_{k-1}.
    \end{align*}
    Since $|\mirror{s}_{k,\ell-1}| \geq |S_6|$, with an application of \autoref{lem:product_of_squares} we obtain that
    \begin{align*}
      \mirror{s}_{k,\ell-1} L(\mirror{s}_{k,\ell-1}) \in \Pi(\oa,\ob) \ \text{ and } \
      \sqrt{\mirror{s}_{k,\ell-1} L(\mirror{s}_{k,\ell-1})} = \mirror{s}_{k,\ell-1}.
    \end{align*}
    Thus $w^2 \in \Pi(\oa,\ob)$ and
    \begin{align*}
      \sqrt{w^2} = \sqrt{\mirror{s}_{k,\ell-1} L(\mirror{s}_{k,\ell-1})} \sqrt{\mirror{s}_{k-1}^{\,2}} =
      \mirror{s}_{k,\ell-1}\mirror{s}_{k-1} = w,
    \end{align*}
    so $w$ is a solution to \eqref{eq:square}.
    
    Consider next the case $w = L(\mirror{s}_k)$ for some $k \geq 3$. Similar to above,
    \begin{align*}
      w^2 & = L(\mirror{s}_{k-2})\mirror{s}_{k-1}^{\,a_k} L(\mirror{s}_{k-2})\mirror{s}_{k-1}^{\,a_k}
      = L(\mirror{s}_{k-2})\mirror{s}_{k-1}^{\,a_k+1} \mirror{s}_{k-2}\mirror{s}_{k-1}^{\,a_k-1} \\
      & = L(\mirror{s}_{k-2})\mirror{s}_{k-1} \mirror{s}_{k-3} \mirror{s}_{k-2}^{\,a_{k-1}} \mirror{s}_{k-1}^{\,a_k-1} \mirror{s}_{k-2}\mirror{s}_{k-1}^{\,a_k-1}
      = L(\mirror{s}_{k-2})\mirror{s}_{k-1} \mirror{s}_{k-3} \mirror{s}_{k-2}^{\,a_{k-1} - 1} \cdot \mirror{s}_{k,a_k-1}^{\,2} \\
      & = \mirror{s}_{k-1}\mirror{s}_{k-2} \mirror{s}_{k-3} \mirror{s}_{k-2}^{\,a_{k-1} - 1} \cdot \mirror{s}_{k,a_k-1}^{\,2}
      = \mirror{s}_{k-1} L(\mirror{s}_{k-1}) \cdot \mirror{s}_{k,a_k-1}^{\,2}.
    \end{align*}
    If $k > 3$, then the claim follows using the induction hypothesis and \autoref{lem:product_of_squares} as above. In
    the case $k = 3$ we have that
    \begin{align*}
      \mirror{s}_{k-1} L(\mirror{s}_{k-1}) \in \Pi(\oa,\ob) \ \text{ and } \
      \sqrt{\mirror{s}_{k-1} L(\mirror{s}_{k-1})} = \mirror{s}_{k-1}.
    \end{align*}
    Namely, it is not difficult to see that if $\ob$ is even, then
    \begin{align*}
      \mirror{s}_{k-1} L(\mirror{s}_{k-1}) = (S_2^2)^{1+\ob/2} S_1^2 (S_4^2)^{\ob/2} \ \text{ and } \
      \mirror{s}_{k-1} = S_2^{1+\ob/2} S_1 S_4^{\ob/2}.
    \end{align*}
    If $\ob$ is odd, then
    \begin{align*}
      \mirror{s}_{k-1} L(\mirror{s}_{k-1}) = (S_2^2)^{(\ob+1)/2}S_3^2 (S_4^2)^{(\ob-1)/2} \ \text{ and } \
      \mirror{s}_{k-1} = S_2^{(\ob+1)/2} S_3 S_4^{(\ob-1)/2}.
    \end{align*}
    Thus $w$ is a solution to \eqref{eq:square} also in the case $k = 3$.
  \end{proof}

  Note that a word $w$ in the set $L(\rsst{\alpha})\setminus L(\rst{\alpha})$ is a solution to \eqref{eq:square} but
  not in the language $\Lang(\alpha)$. Rather, $w$ is a solution to \eqref{eq:square} in $\Lang(\beta)$ where
  $\beta$ is a suitable irrational such that $L(w)$ is a reversed standard word of slope $\beta$.

  From \autoref{prp:standard_solutions} we conclude the following interesting fact:

  \begin{corollary}
    There exist arbitrarily long primitive solutions of \eqref{eq:square} in $\Lang(\alpha)$.
  \end{corollary}

  We can now prove \autoref{thm:complete_characterization}.

  \begin{proof}[Proof of \autoref{thm:complete_characterization}]
    By \autoref{prp:standard_solutions} and \autoref{thm:square_root_condition_iff_rsst} it is sufficient to prove that
    \emph{(i)} implies \emph{(ii)}.

    Suppose that $w$ is a solution to \eqref{eq:square} in $\Lang(\alpha)$. Write $w^2$ as a product of minimal
    squares: $w^2 = X_1^2 X_2^2 \cdots X_n^2$. Let $x \in [w^2]$. Then the word $s_{x,\alpha}$ begins with
    $X_1^2 X_2^2 \cdots X_n^2$, so by \autoref{thm:square_root} the word $\sqrt{s_{x,\alpha}} = s_{\psi(x),\alpha}$
    begins with $X_1 X_2 \cdots X_n$.  Therefore $\psi(x) \in [X_1 X_2 \cdots X_n] = [w]$. Thus $w^2$ satisfies the
    square root condition.
  \end{proof}

  \section{A More Detailed Combinatorial Description of the Square Root Map}\label{sec:combinatorial_version}
  Recall from \autoref{sec:square_root_map} that the square root $\sqrt{s}$ of a Sturmian word $s$ has the same factors
  as $s$. The proofs were dynamical; we used the special mapping $\psi$ on the circle. In this section we describe
  combinatorially why the language is preserved; we give a location for any prefix of $\sqrt{s}$ in $s$. As a side
  product, we are able to describe when a Sturmian word is uniquely factorizable as a product of squares of reversed
  (semi)standard words.

  Let us begin with an introductory example. Recall from \autoref{sec:square_root_map} the square root of the Fibonacci
  word $f$:
  \begin{align*}
    f        &= (010)^2 (100)^2 (10)^2 (01)^2 0^2 (10010)^2 (01)^2 \cdots, \\
    \sqrt{f} &= 010\cdot100\cdot10\cdot01\cdot0\cdot10010\cdot01\cdots.
  \end{align*}
  Obviously the square root $X_1 = 010$ of $(010)^2$ occurs as a prefix of $f$. Equally clearly the word
  $010\cdot100 = \sqrt{(010)^2(100)^2}$ occurs, not as a prefix, but after the prefix $X_1$ of $f$. Thus the position
  of the first occurrence of $010\cdot100$ shifted $|X_1| = 3$ positions from the position of the first occurrence of
  $X_1$. However, when comparing the position of the first occurrence of $\sqrt{(010)^2(100)^2(10)^2}$ with the first
  occurrence of $010\cdot100$, we see that there is no further shift. By further inspection, the word
  $\sqrt{(010)^2(100)^2(10)^2(01)^20^2(10010)^2}$ occurs for the first time at position $|X_1|$ of $f$. This is no
  longer true for the first seven minimal squares; the first occurrence of
  $X_1 X_2 = 010\cdot100\cdot10\cdot01\cdot0\cdot10010\cdot01$ is at position $|X_1 X_2| = 16$ of $f$. The amount of
  shift from the previous position $|X_1| = 3$ is $|X_2| = 13$; observe that both of these numbers are Fibonacci
  numbers. Thus the amount of shift was exactly the length of the square roots added after observing the previous
  shift. As an observant reader might have noticed, both of the words $X_1$ and $X_2$ are reversed standard words, or
  equivalently, primitive solutions to \eqref{eq:square}. Repeating similar inspections on other Sturmian words
  suggests that there is a certain pattern to these shifts and that knowing the pattern would make it possible to
  locate prefixes of $\sqrt{s}$ in the Sturmian word $s$. Thus it makes very much sense to ``accelerate'' the square
  root map by considering squares of solutions to \eqref{eq:square} instead of just minimal squares. Next we make these
  somewhat vague observations more precise.

  Every Sturmian word has a solution of \eqref{eq:square} as a square prefix. Next we aim to characterize Sturmian
  words having infinitely many solutions of \eqref{eq:square} as square prefixes. The next two lemmas are key
  results towards such a characterization.

  \begin{lemma}\label{lem:square_union}
    Consider the reversed (semi)standard word $\mirror{s}_{k,l}$ of slope $\alpha$ with $k \geq 2$ and
    $0 < \ell \leq a_k$. The set $[\mirror{s}_{k,\ell}]\setminus\{1-\alpha\}$ equals the disjoint union
    \begin{align*}
      \left( \bigcup_{i = 0}^\infty \bigcup_{j = 1}^{a_{k+2i}}[\mirror{s}_{k+2i,j}^{\,2}] \right) \setminus \bigcup_{i = 1}^{l-1} [\mirror{s}_{k,i}^{\,2}].
    \end{align*}
    Analogous representations exist for the sets $[\mirror{s}_0]\setminus\{1-\alpha\}$ and
    $[\mirror{s}_1]\setminus\{1-\alpha\}$.
  \end{lemma}

  To put it more simply: for each $x \neq 1-\alpha$ there exists a unique reversed (semi)standard word $w$ such that
  $x \in [w^2]$. To illustrate the proof, we begin by giving a proof sketch.

  \begin{proof}[Proof Sketch]
    Consider as an example the interval $[0] = I(0, 1-\alpha)$. It is easy to see that
    $[0^2] = I(0, -2\alpha) = I(0, -(q_0 + 1)\alpha)$, so $[0] = [0^2] \cup I(-(q_0 + 1)\alpha, 1-\alpha)$. The
    interval $I(-(q_0 + 1)\alpha, 1-\alpha)$ is the interval of the factor $\mirror{s}_{2,1}$. Therefore
    $[0] = [\mirror{s}_0^{\,2}] \cup [\mirror{s}_{2,1}]$. Since $\mirror{s}_{2,1}^{\,2} \in \Lang(\alpha)$, the
    interval $[\mirror{s}_{2,1}^2]$ splits into two parts: $[\mirror{s}_{2,1}] = [\mirror{s}_{2,1}^{\,2}] \cup J$. It
    is straightforward to show that $J = I(-(q_{2,1} + 1)\alpha, 1-\alpha)$. Again, the interval $J$ is the interval of
    the factor $w$ which equals either $\mirror{s}_{2,2}$ or $\mirror{s}_{4,1}$ depending on the number $a_2$. So
    $[0] = [\mirror{s}_0^{\,2}] \cup [\mirror{s}_{2,1}^{\,2}] \cup [w]$. This process can be repeated for the interval
    $[w]$ and indefinitely after that. The very same idea can be applied to any interval $[\mirror{s}_{k,\ell}]$.
  \end{proof}

  \begin{proof}[Proof of \autoref{lem:square_union}]
    Consider the lengths of the reversed (semi)standard words beginning with the same letter as $\mirror{s}_{k,\ell}$. Out
    of these lengths we can form the unique increasing sequence $(b_n)$ such that $b_1 = q_{k,\ell-1}$. If we set
    $s_1 = \mirror{s}_{k,\ell}$ and $J_1 = I(-(b_1+1)\alpha, 1-\alpha)$, then based on the observations in the proof of
    \autoref{prp:endpoint_alpha_characterization} we see that $J_1 = [s_1]$. The interval $J_1$ is split by the point
    $\{-(q_{k,\ell}+1)\alpha\} = \{-(b_2+1)\alpha\}$. It must be that $[s_1^2] = I(-(b_1+1)\alpha, -(b_2+1)\alpha)$.
    Otherwise $[s_1^2] = [s_1] \cap R^{-b_2}([s_1]) = I(-(b_2+1)\alpha, 1-\alpha)$, so the points
    $\{-(b_1+b_2)\alpha\}$ and $\{-b_1\alpha\}$ are on the opposite sides of $0$. Furthermore, $\|(b_1+b_2)\alpha\|$
    equals the distance between the points $\{-b_1\alpha\}$ and $\{-b_2\alpha\}$, so
    $\|(b_1+b_2)\alpha\| = \|q_{k-1}\alpha\|$. Since also the point $\{-q_{k-1}\alpha\}$ is on the side opposite to
    $\{-b_1\alpha\}$, it follows that $q_{k-1} = b_1 + b_2$ which is obviously false. Thus
    $J_2 = J_1 \setminus [s_1^2] = I(-(b_2+1)\alpha,1-\alpha)$ is the interval of $s_2$, the unique reversed
    (semi)standard word of length $b_3$ beginning with the same letter as $s_1$. By repeating this when $n > 1$, we see
    that the interval $J_n$ is split by the point $\{-(b_{n+1}+1)\alpha\}$ and that
    $[s_n^2] = I(-(b_n+1)\alpha, -(b_{n+1}+1)\alpha)$. Then there is a unique reversed (semi)standard word $s_{n+1}$
    such that $[s_{n+1}] = I(-(b_{n+1}+1)\alpha,1-\alpha) = J_n \setminus [s_n^2]$; we set $J_{n+1} = [s_{n+1}]$. By
    the definition of the sequence $(b_n)$, the words $s_{n+1}$ and $s_1$ begin with the same letter. This yields a
    well-defined sequence $(J_n)$ of nested subintervals of $J_1$. It is clear that $|J_n| \to 0$ as $n \to \infty$. It
    follows that
    \begin{align*}
      [\mirror{s}_{k,\ell}] \cup \{1-\alpha\} = J_1 \cup \{1-\alpha\} = \bigcup_{n=1}^\infty [s_n^2] \cup \{1-\alpha\}.
    \end{align*}
    The sets $[s_n^2]$ are by definition disjoint. The claim follows since the indexing in the claim is just another
    way to express the reversed (semi)standard words having lengths from the sequence $(b_n)$.

    The above proof works as it is for the cases $\mirror{s}_0$ and $\mirror{s}_1$; only minor adjustments in notation
    are needed.
  \end{proof}

  \begin{lemma}\label{lem:not_square_prefix}
    Let $u \in \rsst{\alpha}$ and $v \in \rsst{\alpha} \cup L(\rsst{\alpha})$. Then $u^2$ is never a proper prefix of
    $v^2$.
  \end{lemma}
  \begin{proof}
    If $v \in \rsst{\alpha}$ and $|u| \neq |v|$, then by \autoref{lem:square_union}, the intervals $[u^2]$ and $[v^2]$
    are disjoint. Hence $u^2$ can never be a proper prefix of $v^2$. Assume then that $v \in L(\rsst{\alpha})$. If
    $|v| \leq |\mirror{s}_1|$, then $v^2$ is a minimal square, so it is not possible for $u^2$ to be a proper prefix of
    $v^2$. Suppose that $|v| = |\mirror{s}_{k,\ell}|$ for some $k \geq 2$ with $0 < \ell \leq a_k$. As in the proof of
    \autoref{prp:endpoint_alpha_characterization}, we have that $[v] = I(-(q_{k-1}+1)\alpha, 1-\alpha)$. If $u$ begins
    with the same letter as $v$ and $|u| < |v|$, then $|u| \leq |\mirror{s}_{k-1}|$. It follows, as in the proof of
    \autoref{lem:square_union}, that the distance between $1-\alpha$ and either of the endpoints of the interval
    $[u^2]$ must be at least $\|q_{k-1}\alpha\|$. Hence the intervals $[v]$ and $[u^2]$ are disjoint, so $u^2$ is not a
    proper prefix of $v^2$.
  \end{proof}

  %

  Let $s$ be a fixed Sturmian word of slope $\alpha$. Since the index of a factor of a Sturmian word is finite,
  \autoref{lem:not_square_prefix} and \autoref{thm:complete_characterization} imply that if $s$ has infinitely many
  solutions of \eqref{eq:square} as square prefixes then no word in $\rsst{\alpha}$ is a square prefix of $s$. We have
  now the proper tools to prove the following:

  \begin{proposition}\label{prp:begins_with_rsst}
    Let $s_{x,\alpha}$ be a Sturmian word of slope $\alpha$ and intercept $x$. Then $s_{x,\alpha}$ begins with a square
    of a word in $\rsst{\alpha}$ if and only if $x \neq 1-\alpha$.
  \end{proposition}
  \begin{proof}
    If $x \neq 1-\alpha$, then $x \in I_0\setminus\{1-\alpha\} = [\mirror{s}_0]\setminus\{1-\alpha\}$ or
    $x \in I_1\setminus\{1-\alpha\} = [\mirror{s}_1]\setminus\{1-\alpha\}$. Thus by applying \autoref{lem:square_union}
    to $I_0\setminus\{1-\alpha\}$ or $I_1\setminus\{1-\alpha\}$, we see that the word $s_{x,\alpha}$ begins with a
    square of a word in $\rsst{\alpha}$.

    Suppose then that $x = 1-\alpha$. Then $s_{x,\alpha} \in \{01c_\alpha, 10c_\alpha\}$. It is a well-known fact that
    $s_{2k} = P_{2k} 10$ and $s_{2k+1} = Q_{2k+1} 01$ for some palindromes $P_{2k}$ and $Q_{2k+1}$ for every $k \geq 1$
    (see e.g. \cite[Lemma~2.2.8]{2002:algebraic_combinatorics_on_words}). As $c_\alpha = \lim_{k\to\infty} s_k$, it
    follows that $01c_\alpha = \lim_{k\to\infty} \mirror{s}_{2k}$ and
    $10c_\alpha = \lim_{k\to\infty} \mirror{s}_{2k+1}$. Hence by \autoref{lem:not_square_prefix}, the
    word $s_{x,\alpha}$ cannot have as a prefix a square of a word in $\rsst{\alpha}$.
  \end{proof}

  It follows that if $s$ has infinitely many solutions of \eqref{eq:square} as square prefixes, then
  $s \in \{01c_\alpha, 10c_\alpha\}$.

  Next we take one extra step and characterize when $s$ can be written as a product of squares of words in
  $\rsst{\alpha}$.

  \begin{theorem}
    A Sturmian word $s$ of slope $\alpha$ can be written as a product of squares of words in $\rsst{\alpha}$ if and
    only if $s$ is not of the form $X_1^2 X_2^2 \cdots X_n^2 c$ where $X_i \in \rsst{\alpha}$ and
    $c \in \{01c_\alpha, 10c_\alpha\}$. If $s$ is a product of squares in $\rsst{\alpha}$, then this product is unique.
  \end{theorem}
  \begin{proof}
    This is a direct consequence of \autoref{prp:begins_with_rsst} and \autoref{lem:not_square_prefix}.
  \end{proof}

  Suppose that $s \notin \{01c_\alpha, 10c_\alpha\}$. Then the word $s$ has only finitely many solutions of
  \eqref{eq:square} as square prefixes. We call the longest solution \emph{maximal}. Observe that the maximal solution is
  not necessarily primitive since any power of a solution to \eqref{eq:square} is also a solution. Sturmian words of
  slope $\alpha$ can be classified into two types.\vspace{\baselineskip}

  \textbf{Type A}. Sturmian words $s$ of slope $\alpha$ which can be written as products of maximal solutions to
  \eqref{eq:square}. In other words, it can be written that $s = X_1^2 X_2^2 \cdots$ where $X_i$ is the maximal
  solution occurring as a square prefix of the word $T^{h_i}(s)$ where $h_i = |X_1^2 X_2^2 \cdots X_{i-1}^2|$.

  \textbf{Type B}. Sturmian words $s$ of slope $\alpha$ which are of the form $s = X_1^2 X_2^2 \cdots X_n^2 c$ where
  $c \in \{01c_\alpha, 10c_\alpha\}$ and the words $X_i$ are maximal solutions as above.\vspace{\baselineskip}

  \autoref{prp:begins_with_rsst} and \autoref{lem:not_square_prefix} imply that the words $X_i$ in the above
  definitions are uniquely determined and that the primitive root of a maximal solution is in $\rsst{\alpha}$.
  Consequently, a maximal solution is always right special. When finding the factorization of a Sturmian word as a
  product of squares of maximal solutions, it is sufficient to detect at each position the shortest square of a word in
  $\rsst{\alpha}$ and take its largest even power occurring in that position.

  Keeping the Sturmian word $s$ of slope $\alpha$ fixed, we define two sequences $(\mu_k)$ and $(\lambda_k)$. We set
  $\mu_0 = \lambda_0 = \varepsilon$. Following the notation above, we define depending on the type of $s$ as follows.

  \textbf{(A)} If $s$ is of type A, then we set for all $k \geq 1$ that
  \begin{align*}
    \mu_k     &= X_1^2 X_2^2 \cdots X_k^2 \,\, \text{ and }\\
    \lambda_k &= X_1 X_2 \cdots X_k.
  \end{align*}

  \textbf{(B)} If $s$ is of type B, then we set for $1 \leq k \leq n$ that
  \begin{align*}
    \mu_k     &= X_1^2 X_2^2 \cdots X_k^2 \,\, \text{ and }\\
    \lambda_k &= X_1 X_2 \cdots X_k,
  \end{align*}
  and we let
  \begin{align*}
    \mu_{n+1}     &= X_1^2 X_2^2 \cdots X_n^2 c \,\, \text{ and }\\
    \lambda_{n+1} &= X_1 X_2 \cdots X_n c.
  \end{align*}

  Compare these definitions with the example in the beginning of this section; the words $X_1$ and $X_2$ are maximal
  solutions in the Fibonacci word (which is of type A).

  We are finally in a position to formulate precisely the observations made in the beginning of this section and state
  the main result of this section.

  \begin{theorem}\label{thm:combinatorial_version}
    Let $s$ be a Sturmian word with slope $\alpha$.

    \emph{(A)} If $s$ is of type A, then
    \begin{align*}
      \sqrt{s} = \lim_{k \to \infty} T^{|\lambda_k|}(s).
    \end{align*}
    Moreover, the first occurrence of the prefix $\lambda_{k+1}$ of $\sqrt{s}$ is at position $|\lambda_k|$ of $s$ for
    all $k \geq 0$.

    \emph{(B)} If $s$ is of type B, then
    \begin{align*}
      \sqrt{s} = T^{|\lambda_n|}(s).
    \end{align*}
    Moreover, the first occurrence of the prefix $\lambda_{k+1}$ with $0 \leq k \leq n-1$ is at position $|\lambda_k|$
    of $s$, and the first occurrence of any prefix of $\sqrt{s}$ having lenght greater than $|\lambda_n|$ is at
    position $|\lambda_n|$ of $s$.
    
    In particular $\sqrt{s}$ is a Sturmian word with slope $\alpha$.
  \end{theorem}

  The theorem only states where the prefixes $\lambda_k$ of $\sqrt{s}$ occur for the first time. For the first
  occurrence of other prefixes of $\sqrt{s}$ we do not have a guaranteed location.

  \begin{figure}
  \centering
  \begin{tikzpicture}
    \coordinate (A) at (-5.75,0.8);
    \coordinate (B) at ($ (A) + (0.87,0) $);
    \coordinate (X1) at ($ (A) + (0.5,0.4) $);
    \draw [decorate,decoration={brace}] (A) -- (B);
    \node at (X1) {\small{$X_1^2$}};
    \coordinate (C) at ($ (B) + (0.07,0) $);
    \coordinate (D) at ($ (C) + (1.53,0) $);
    \coordinate (X2) at ($ (C) + (0.8,0.4) $);
    \draw [decorate,decoration={brace}] (C) -- (D);
    \node at (X2) {\small{$X_2^2$}};
    \coordinate (E) at ($ (D) + (0.07,0) $);
    \coordinate (F) at ($ (E) + (2.45,0) $);
    \coordinate (X3) at ($ (D) + (1.35,0.4) $);
    \draw [decorate,decoration={brace}] (E) -- (F);
    \node at (X3) {\small{$X_3^2$}};
    \coordinate (G) at ($ (F) + (0.07,0) $);
    \coordinate (H) at ($ (G) + (4.05,0) $);
    \coordinate (X4) at ($ (G) + (2.10,0.4) $);
    \draw [decorate,decoration={brace}] (G) -- (H);
    \node at (X4) {\small{$X_4^2$}};
    \coordinate (I) at ($ (H) + (0.07,0) $);
    \coordinate (J) at ($ (I) + (4.05,0) $);
    \coordinate (X5) at ($ (I) + (2.10,0.4) $);
    \draw [decorate,decoration={brace}] (I) -- (J);
    \node at (X5) {\small{$X_5^2$}};

    \node (box) {%
    \begin{minipage}[t!]{0.5\textwidth}
      \setlength{\jot}{-0.2em}
      {\small
      \begin{align*}
        \tau      &:\, 01001010010100100101001001010010100100101001010010010100100101001010010010100100\cdots\\
        \lambda_1 &:\, 010\\
        \lambda_2 &:\, \phantom{010}01010010\\
        \lambda_3 &:\, \phantom{01001010}0101001001010010\\
        \lambda_4 &:\, \phantom{0100101001010010}01010010010100101001001010010\\
        \lambda_5 &:\, \phantom{01001010010100100101001001010}01010010010100101001001010010010100101001001010010
      \end{align*}
      }%
    \end{minipage}
    };

    \tikzstyle{muBar}=[line width=0.8pt,black,<->]
    \coordinate (B1) at ($ (B) + (0.06,-0.35) $);
    \draw[muBar] (B1) -- ($ (B1) + (0,-0.30) $);
    \coordinate (B2) at ($ (D) + (0.04,-0.35) $);
    \draw[muBar] (B2) -- ($ (B2) + (0,-0.60) $);
    \coordinate (B3) at ($ (F) + (0.05,-0.35) $);
    \draw[muBar] (B3) -- ($ (B3) + (0,-0.90) $);
    \coordinate (B4) at ($ (H) + (0.04,-0.35) $);
    \draw[muBar] (B4) -- ($ (B4) + (0,-1.25) $);

    \tikzstyle{lambdaBar}=[line width=0.8pt,gray,<->]
    \coordinate (D1) at ($ (B1) + 0.5*(B2) - 0.5*(B1) $);
    \draw[lambdaBar] (D1) -- ($ (D1) + (0,-0.30) $);
    \coordinate (D2) at ($ (B2) + 0.5*(B3) - 0.5*(B2) $);
    \draw[lambdaBar] (D2) -- ($ (D2) + (0,-0.60) $);
    \coordinate (D3) at ($ (B3) + 0.5*(B4) - 0.5*(B3) $);
    \draw[lambdaBar] (D3) -- ($ (D3) + (0,-0.90) $);
    \coordinate (D4) at ($ (B4) + (3.32,0) $);
    \draw[lambdaBar] (D4) -- ($ (D4) + (0,-1.25) $);

  \end{tikzpicture}
  \caption{The first occurrences of the words $\lambda_k$ in $\tau$. The eighth shift of the Fibonacci word was used
  since for the Fibonacci word the lengths $|\lambda_k|$ grow very rapidly.}
  \label{fig:tau}
  \end{figure}
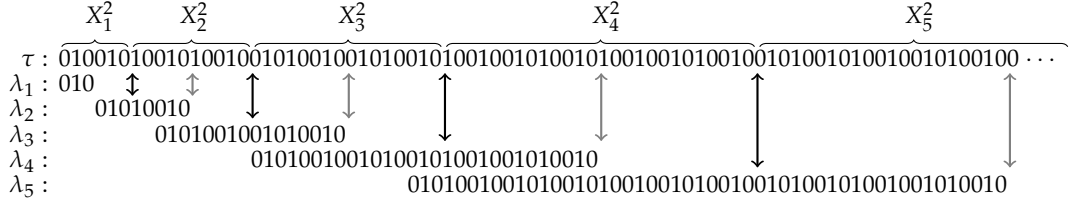

  To illustrate the theorem, consider next $\tau$, the eighth shift of the Fibonacci word. If we write under the word
  $\tau$ each of the corresponding words $\lambda_k$ at the position of their first occurrence we get the picture in
  \autoref{fig:tau}. \autoref{thm:combinatorial_version} shows that the nice pattern where the words $\lambda_k$
  overlap continues indefinitely and, moreover, that if we replace $\tau$ with any other Sturmian word (of type A) we
  obtain a similar picture. Most of the results of this paper were motivated by the discovery of this pattern.

  Before proving the theorem we need one more result.

  \begin{proposition}\label{prp:max_solution_suffix}
    Suppose that $s$ is a Sturmian word of type A. Then the word $\lambda_k$ is right special and a suffix of the word
    $\mu_k$ for all $k \geq 0$.
  \end{proposition}
  \begin{proof}
    This proof might be tricky to follow. We advise the reader to keep the picture of \autoref{fig:intervals} in mind
    while reading the proof. This picture depicts only the Case A below but is surely helpful.

    The assertion is evident when $k=0$. Suppose that $k > 0$ and assume that $\lambda_k$ is right special and that
    $\lambda_k$ is a suffix of the word $\mu_k$. It is equivalent to say that
    $\{-(|\lambda_k|+1)\alpha\} \in [\lambda_k]$ and $[\mu_k] \subseteq R^{-|\lambda_k|}([\lambda_k])$ (evidently
    $2|\lambda_k| = |\mu_k|$). We write simply $\lambda = \lambda_k$, $\mu = \mu_k$, and $X = X_{k+1}$. This proof
    utilizes only the facts that $\mu X^2 \in \Lang(\alpha)$ and that $\lambda$ is right special and a suffix of the
    word $\mu$, not the structure of the words $\lambda$ and $\mu$ implied by their definitions. Thus without loss of
    generality, we may assume that $X$ is primitive. Consequently, $X \in \rsst{\alpha}$. It follows that
    \begin{align}\label{eq:x_intervals}
      [X]   &= I(-(q+1)\alpha,1-\alpha) \text{ and }\\
      [X^2] &= [X] \cap R^{-|X|}([X]) = I(-(q+1)\alpha, -(|X| + 1)\alpha). \nonumber
    \end{align}
    for some nonnegative integer $q$. Let $x = \{-(|\mu|+1)\alpha\}$. It follows from the hypothesis
    $\{-(|\lambda|+1)\alpha\} \in [\lambda]$ that $x \in R^{-|\lambda|}([\lambda])$. By \eqref{eq:x_intervals} the point
    $x$ is an endpoint of the interval $R^{-|\mu|}([X])$.
    
    Let then $y = \{-(|\mu X|+1)\alpha\}$. By \eqref{eq:x_intervals} the point $y$ is an endpoint of the interval
    $R^{-|\mu X|}([X])$ and an interior point of the interval $R^{-|\mu|}([X])$. Suppose for a contradiction that
    $y \notin R^{-|\lambda|}([\lambda])$. As $x \notin R^{-|\mu X|}([X])$ (otherwise it would follow that
    $1-\alpha \in R^{-|X|}([X])$ which contradicts \eqref{eq:x_intervals}), it follows that
    $R^{-|\lambda|}([\lambda]) \cap R^{-|\mu X|}([X]) = \emptyset$. Since
    \begin{align*}
      [\mu_{k+1}] = [\mu] \cap R^{-|\mu|}([X]) \cap R^{-|\mu X|}([X]),
    \end{align*}
    we have that $[\mu_{k+1}] \subseteq R^{-|\mu X|}([X])$. By assumption
    $[\mu_{k+1}] \subseteq [\mu] \subseteq R^{-|\lambda|}([\lambda])$. Thus
    $[\mu_{k+1}] \subseteq R^{-|\lambda|}([\lambda]) \cap R^{-|\mu X|}([X])$, so by the above we are forced to conclude
    that $[\mu_{k+1}] = \emptyset$. This is a contradiction since $X$ is chosen in such a way that
    $[\mu_{k+1}] = [\mu X^2] \neq \emptyset$. We conclude that $y \in R^{-|\lambda|}([\lambda])$.

    Now $R^{-|\lambda|}([\lambda X]) = R^{-|\lambda|}([\lambda]) \cap R^{-|\mu|}([X])$. Since
    $y \in R^{-|\lambda|}([\lambda]), R^{-|\mu|}([X])$, it follows that
    $y = \{-(|\mu X|+1)\alpha\} \in R^{-|\lambda|}([\lambda X])$. Thus
    $R^{|\lambda|}(y) = \{-(|\lambda X| + 1)\alpha\} \in [\lambda X]$, so the word $\lambda X$ is right special. We
    have two cases depending on the length of the interval $R^{-|\mu|}([X])$ compared to the length of the interval
    $R^{-|\lambda|}([\lambda])$.

    \textbf{Case A.} $R^{-|\mu|}([X]) \nsubseteq R^{-|\lambda|}([\lambda])$. In this case
    $R^{-|\lambda|}([\lambda X]) = I(x, z)$ where $z$ an endpoint of $R^{-|\lambda|}([\lambda])$. Since $y$ is an
    interior point of $R^{-|\lambda|}([\lambda X])$, $R^{-|X|}(x) = y$, and $x \notin R^{-|\lambda X|}([\lambda X])$, we
    obtain that $I(y, z) \subseteq R^{-|\lambda X|}([\lambda X])$. Since $y$ is also an interior point of
    $R^{-|\lambda|}([\lambda])$, we obtain similarly that $R^{-|\lambda|}([\lambda]) \cap R^{-|\mu|}([X^2]) = I(y, z)$.
    Thus
    \begin{align*}
      [\mu_{k+1}] = [\mu] \cap R^{-|\mu|}([X^2]) \subseteq R^{-|\lambda|}([\lambda]) \cap R^{-|\mu|}([X^2]) = I(y,z) \subseteq R^{-|\lambda X|}([\lambda X]).
    \end{align*}
    This proves that $\lambda X = \lambda_{k+1}$ is a suffix of $\mu_{k+1}$.

    \textbf{Case B.} $R^{-|\mu|}([X]) \subseteq R^{-|\lambda|}([\lambda])$. It follows that
    $R^{-|\lambda|}([\lambda X]) = R^{-|\mu|}([X])$, so $R^{-|\lambda X|}([\lambda X]) = R^{-|\mu X|}([X])$. Since
    $R^{-|\mu|}([X^2]) \subseteq R^{-|\mu X|}([X])$, we get that
    \begin{align*}
      [\mu_{k+1}] = [\mu] \cap R^{-|\mu|}([X^2]) \subseteq [\mu] \cap R^{-|\mu X|}([X]) \subseteq [\mu] \cap R^{-|\lambda X|}([\lambda X])
    \end{align*}
    proving that also in this case $\lambda X = \lambda_{k+1}$ is a suffix of $\mu_{k+1}$.
  \end{proof}

  Note that even though $\lambda_k$ is right special and always a suffix of $\mu_k$, it is not necessary for $\mu_k$ to
  be right special.

  \begin{figure}
  \centering
  \begin{tikzpicture}
    \definecolor{pink}{RGB}{236,0,140}
    \definecolor{bluish}{RGB}{46,49,146}
    \definecolor{cyanish}{RGB}{0,174,239}
    \definecolor{greenish}{RGB}{0,148,69}
    \tikzstyle{own}=[line width=0.8pt,|-|]
    \tikzstyle{own2}=[line width=0.8pt,black]
    \tikzstyle{type1}=[line width=0.8pt,bluish,text=bluish,font=\bf]
    \tikzstyle{type2}=[line width=0.8pt,greenish,text=greenish,font=\bf]
    \tikzstyle{type3}=[line width=0.8pt,cyanish,text=cyanish,font=\bf]

    \filldraw[own2] (2,0) circle(0.9pt);
    \node at (2,0.3) {\scriptsize{$x$}};
    \filldraw[own2] (3,0) circle(0.9pt);
    \node at (3,0.3) {\scriptsize{$y$}};
    \node at (7,0.3) {\scriptsize{$z$}};

    \draw[own] (0,0) -- (7,0);
    \node at (10,0) {\scriptsize{$R^{-|\lambda|}([\lambda])$}};

    \draw[own] (2,-0.5) -- (8,-0.5);
    \node at (10,-0.5) {\scriptsize{$R^{-|\mu|}([X])$}};

    \draw[own] (3,-1) -- (9,-1);
    \node at (10,-1) {\scriptsize{$R^{-|\mu X|}([X])$}};

    \draw[dotted] (2,0) -- (2,-1.5);
    \draw[dotted] (7,0) -- (7,-2.5);

    \draw[dotted] (3,0) -- (3,-2.5);
    \draw[dotted] (6,0.5) -- (6,-1);

    \draw[own] (2,-1.5) -- (7,-1.5);
    \node at (10,-1.5) {\scriptsize{$R^{-|\lambda|}([\lambda X])$}};

    \draw[own] (3,-2) -- (8,-2);
    \node at (10,-2) {\scriptsize{$R^{-|\lambda X|}([\lambda X])$}};

    \draw[own] (3,-2.5) -- (7,-2.5);
    \node at (10,-2.5) {\scriptsize{$R^{-|\mu|}([X^2])$}};

    \draw[own] (1.5,0.5) -- (6,0.5);
    \node at (10,0.5) {\scriptsize{$[\mu]$}};
    \draw[own2,|-,cyanish] (2,-0.5) -- (6,-0.5);
    \draw[own2,|-,pink] (3,-1) -- (6,-1);
  \end{tikzpicture}
  \caption{A possible arrangement for the intervals in the Case A of the proof of \autoref{prp:max_solution_suffix}. The blue
  color marks the interval $[\mu X]$ and magenta marks the interval $[\mu X^2] = [\mu_{k+1}]$.}
  \label{fig:intervals}
  \end{figure}
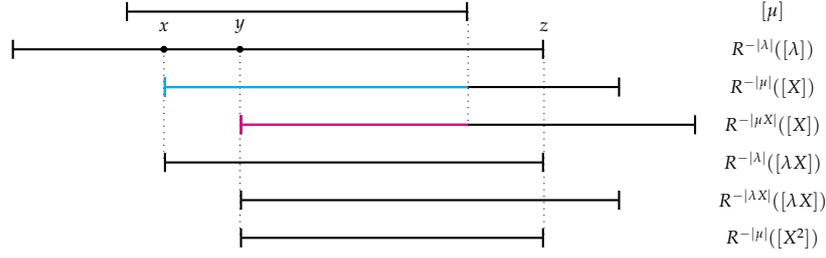

  \begin{figure}
  \centering
  \begin{tikzpicture}
    \tikzstyle{own}=[line width=0.8pt,|-|]
    \tikzstyle{own2}=[line width=0.8pt,black,-|]

    \coordinate (O) at (0,0);
    \coordinate (A) at (10,0);
    \coordinate (D) at (0,-0.7);
    \coordinate (U) at (0,0.3);
    \coordinate (M) at ($ 0.35*(A) - 0.35*(O) $);
    \coordinate (W) at ($ 0.07*(A) - 0.07*(O) $);
    \coordinate (X) at ($ 0.5*(A) - 0.5*(M) - 0.5*4*(W) $);

    \draw[own] (O) -- ($ (O) + (A) $);
    \draw[own2] (0,0) -- ($ (O) + (M) + (0.03,0) $);                      
    \draw[own2] ($ (O) + (M) $) -- ($ (O) + (M) + 2*(W) $);               
    \draw[own2] ($ (O) + (M) + 2*(W) $) -- ($ (O) + (M) + 4*(W) $);       
    \draw[own2] ($ (O) + (M) + 4*(W) $) -- ($ (O) + (M) + 4*(W) + (X) $); 
    \draw[own] ($ (O) + (D) + (M) $) -- ($ (O) + (D) + (M) + 4*(W) $);
    \foreach \i in {0,...,3} {
      \pgfmathsetmacro\r{\i + 1}
      \draw[own2] ($ (O) + (D) + (M) + \i*(W) $) -- ($ (O) + (D) + (M) + \r*(W) $);
    }
    \draw[own] ($ (O) + 2*(D) + 0.5*(M) $) -- ($ (O) + 2*(D) + (M) $);
    \draw[own] ($ (O) + 3*(D) + 0.5*(M) + (W) $) -- ($ (O) + 3*(D) + 0.5*(M) + (W) + 0.5*(M) + 2*(W) + (X) $);
    \draw[own] ($ (O) + 3*(D) + 0.5*(M) + (W) $) -- ($ (O) + 3*(D) + 0.5*(M) + (W) + 0.5*(M) $);
    \draw[own] ($ (O) + 3*(D) + 0.5*(M) + (W) $) -- ($ (O) + 3*(D) + 0.5*(M) + (W) + 0.5*(M) + 2*(W) $);
    \coordinate (E) at (-0.015,0);
    \draw[dotted] ($ (O) + (D) + (M) + (W) + (E) $) -- ($ (O) + 3*(D) + (M) + (W) + (E) $);
    \draw[dotted] ($ (O) + (D) + (M) + 3*(W) + (E) $) -- ($ (O) + 3*(D) + (M) + 3*(W) + (E) $);

    \node at ($ (O) + 0.5*(M) + (U) $) {$\mu_{k-1}$};
    \node at ($ (O) + (M) + 0.5*2*(W) + (U) $) {$X_k$};
    \node at ($ (O) + (M) + 1.5*2*(W) + (U) $) {$X_k$};
    \node at ($ (O) + (M) + 4*(W) + 0.5*(X) + (U) $) {$X_{k+1}$};
    \node at ($ (O) + (M) + 4*(W) + 1.5*(X) + (U) $) {$X_{k+1}$};
    \node at ($ (O) + (D) + (M) + 0.5*(W) + (U) $) {$w$};
    \node at ($ (O) + (D) + (M) + 1.5*(W) + (U) $) {$w$};
    \node at ($ (O) + (D) + (M) + 2.5*(W) + (U) $) {$w$};
    \node at ($ (O) + (D) + (M) + 3.5*(W) + (U) $) {$w$};
    \node at ($ (O) + 2*(D) + 0.75*(M) + (U) $) {$\lambda_{k-1}$};
    \node at ($ (O) + 3*(D) + 0.75*(M) + (W) + (U) $) {$\lambda_{k-1}$};
    \node at ($ (O) + 3*(D) + (M) + (W) + 0.5*2*(W) + (U) $) {$X_k$};
    \node at ($ (O) + 3*(D) + (M) + (W) + 2*(W) + 0.5*(X) + (U) $) {$X_{k+1}$};

  \end{tikzpicture}
  \caption{Possible locations for factors in the proof of \autoref{thm:combinatorial_version}.}
  \label{fig:factor_locations}
  \end{figure}
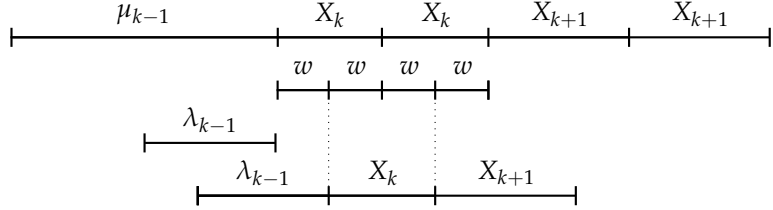

  \begin{proof}[Proof of \autoref{thm:combinatorial_version}]
    Since Sturmian words of type B differ from Sturmian words of type A essentially only by the fact that the sequence
    of maximal solutions is finite, it is in this proof enough to consider the case that $s$ is of type A.

    \autoref{prp:max_solution_suffix} says that $\lambda_k$ is always a suffix of $\mu_k$ for
    all $k \geq 0$. Since $|\mu_k| = 2|\lambda_k|$, it follows that the word $T^{|\lambda_k|}(s)$ has the word
    $\lambda_k$ as a prefix. Therefore $\sqrt{s} = \lim_{k\to\infty} T^{|\lambda_k|}(s)$.

    It remains to prove that the first occurrence of $\lambda_{k+1}$ in $s$ is at position $|\lambda_k|$ of $s$ for all
    $k \geq 0$. It is clear that the first occurrence of $\lambda_1 = X_1$ is at position $|\lambda_0| = 0$. Assume
    that $k > 0$, and suppose for a contradiction that $\lambda_{k+1}$ occurs before the position $|\lambda_k|$. Since
    $\lambda_k$ is a prefix of $\lambda_{k+1}$, by induction we see that $\lambda_{k+1}$ cannot occur before the
    position $|\lambda_{k-1}|$. This means that an occurrence of $X_k X_{k+1}$ begins in $s$ at position $\nu$ such
    that $|\mu_{k-1}| \leq \nu < |\mu_{k-1}X_k|$; see \autoref{fig:factor_locations}. Observe that $s$ has at position
    $|\mu_{k-1}|$ an occurrence of $X_k^2$. Write now $X_k = w^t$ with $w \in \rsst{\alpha}$. Since $w$ is primitive,
    we must have that $\nu = |\mu_{k-1}| + r|w|$ with $0 \leq r < t$. Thus $X_{k+1}$ occurs in $s$ at position
    $\nu + |X_k| = |\mu_{k-1}| + (r+t)|w|$. Since $r < t$, it follows that either $w$ is a prefix of $X_{k+1}$ or
    $X_{k+1}$ is a prefix of $w$.

    Suppose first that $w$ is a prefix of $X_{k+1}$. If $w = X_{k+1}$, then the prefix $\mu_{k-1} X_k^2$ of $s$ is
    followed by $w^2$. Now $w^{2t+2}$ is a solution to \eqref{eq:square} implying that $X_k$ is not a maximal solution
    to \eqref{eq:square}. Since this is contradictory, we infer that $|w| < |X_{k+1}|$. Since $X_{k+1}$ occurs at
    position $|\mu_{k-1}| + (r+t)|w| < |\mu_k|$ and $X_{k+1}$ has $w$ as a prefix, it must be that $X_{k+1}$ begins
    with $wa$ where $a$ is the first letter of $w$. Since $w$ is right special and $w^2 \in \Lang(\alpha)$, it follows
    that $X_{k+1}^2$ begins with $w^2$. Like above, this implies that $X_k$ is not maximal. This is a contradiction.

    Suppose then that $X_{k+1}$ is a proper prefix of $w$. First of all, $X_{k+1}$ must be primitive, as otherwise
    $X_{k+1}$ and consequently $w$ would have as a prefix a square of some word in $\rsst{\alpha}$ contradicting
    \autoref{lem:not_square_prefix}. The assumption that $X_{k+1}$ is a prefix of $w$ implies that $X_{k+1}$ and $w$
    begin with the same letter. Like above, since $w$ is right special and $w^2 \in \Lang(\alpha)$, it must be that $w$
    occurs after the prefix $\mu_k$ of $s$. Since also $X_{k+1}^2$ occurs after the prefix $\mu_k$, by
    \autoref{lem:not_square_prefix} we conclude that the word $w$ must be a proper prefix of $X_{k+1}^2$. Observe now
    that the assumption that $X_{k+1}$ is a proper prefix of $w$ excludes the possibilities that $w = \mirror{s}_0 = 0$
    or $w = \mirror{s}_1 = 10^a$. Therefore $w = \mirror{s}_{h,\ell}$ for some $h \geq 2$ with $0 < \ell \leq a_h$. Because
    $|w| < 2|X_{k+1}|$, we must have that $|X_{k+1}| > |\mirror{s}_{h-2}|$. On the other hand, since $|X_{k+1}| < |w|$
    and $X_{k+1}$ and $w$ begin with the same letter, the only option is that $X_{k+1} = \mirror{s}_{h,\ell'}$ with
    $0 < \ell' < \ell$. Now
    \begin{align*}
      X_{k+1}^2 = (\mirror{s}_{h-2} \mirror{s}_{h-1}^{\,\ell'})^2
      = \mirror{s}_{h-2}\mirror{s}_{h-1}^{\,\ell'} L(\mirror{s}_{h-1}) \mirror{s}_{h-2} \mirror{s}_{h-1}^{\,\ell'-1},
    \end{align*}
    so as $w$ is a prefix of $X_{k+1}^2$, it must be that $\mirror{s}_{h-1} = L(\mirror{s}_{h-1})$. This is a
    contradiction. This final contradiction ends the proof.
  \end{proof}

  As a conclusion of this section, we study the lengths of the maximal solutions of \eqref{eq:square}. Namely, let
  $s = X_1^2 X_2^2 \cdots$ be a Sturmian word of type A factorized as a product of maximal solutions $X_i$. Computer
  experiments suggest that typically the sequence $(|X_i|)$ is strictly increasing. However, there are examples where
  $|X_i| > |X_{i+1}|$ for some $i \geq 1$. It is natural to ask if the lengths can decrease significantly or if
  oscillation is possible. It turns out that neither is possible. In \autoref{cor:solution_length_increasing}
  we prove that $\liminf_{i\to\infty} |X_i| = \infty$.

  First we need a result on certain periods of (semi)standard words.

  \begin{lemma}\label{lem:standard_periods}
    Let $u,v \in \sst{\alpha}$ and $|u| > |v|$. If $u$ is a prefix of some word in $v^+$, then $u = s_{k,\ell}$ and
    $v = s_{k-1}$ for some $k \geq 2$ with $0 < \ell \leq a_k$.
  \end{lemma}
  \begin{proof}
    Suppose that $u$ is a prefix of some word in $v^+$. If $u = s_1 = 0^\oa 1$, then necessarily $v = s_0 = 0$. Then
    obviously $u$ is not a prefix of any word in $v^+$. Therefore $u = s_{k,\ell}$ for some $k \geq 2$ with
    $0 < \ell \leq a_k$. Suppose that $k = 2$. Then $u = (0^\oa 1)^\ell 0$. It is straightforward to show that $v$ must
    equal to $s_1 = 0^\oa 1$; $u$ cannot be a prefix of a word in $v^+$ if $v = s_0 = 0$ or $v = s_{2,\ell'}$ for some
    $\ell'$ such that $0 < \ell' < \ell$. Thus we may assume that $k > 2$.

    Suppose first that $|v| > |s_{k-1}|$. Then by the assumption $|u| > |v|$, it must be that $v = s_{k,\ell'}$ for
    some $\ell'$ such that $\ell' < \ell$. Since $u$ is a prefix of some word in $v^+$, it follows that the word
    $w = s_{k-1}^{\ell-\ell'}s_{k-2}$ is a prefix of some word in $s_{k-2}v^+$. Since the word $w$ begins with
    $s_{k-1}s_{k-2}$, we obtain that $s_{k-2}v$ begins with $s_{k-1}s_{k-2}$, so $s_{k-1}s_{k-2} = s_{k-2}s_{k-1}$.
    This is a contradiction.

    Assume then that $|v| < |s_{k-1}|$. Now the prefix $s_{k-1}$ of $u$ is a prefix of some word in $v^+$, so by
    induction $v = s_{k-2}$. Now $u = (s_{k-2}^{a_{k-1}}s_{k-3})^\ell s_{k-2}$, so as $u$ is a prefix of some word in
    $v^+$, it follows that $z = s_{k-3}s_{k-2}$ is a prefix of some word in $v^+$. This means that $z$ ends with a
    prefix of $s_{k-2}$ of length $|s_{k-3}|$. As the prefix of $s_{k-2}$ of length $|s_{k-3}|$ is $s_{k-3}$, the word
    $z$ ends with $s_{k-3}$. Consequently $s_{k-3}s_{k-2} = s_{k-2}s_{k-3}$; a contradiction.

    The only remaining option is that $v = s_{k-1}$. This is certainly possible.
  \end{proof}

  The next proposition describes precisely under which conditions it is possible that $|X_i| > |X_{i+1}|$. Moreover, it
  rules out the possibility that the lengths decrease significantly or oscillate. 

  \begin{proposition}\label{prp:next_shorter}
    Let $s = X_1^2 X_2^2 X_3^2 \cdots$ be a Sturmian word of type $A$ with slope $\alpha$ factorized as a product of
    maximal solutions $X_i$. If $|X_1| > |X_2|$, then $X_1 = \mirror{s}_{k,\ell}$ for some $k \geq 2$ with
    $0 < \ell \leq a_k - 1$, the primitive root of $X_2$ is $\mirror{s}_{k-1}$, and $|X_3| > |X_1|$.
  \end{proposition}
  \begin{proof}
    Assume that $|X_1| > |X_2|$. Let us first make the additional assumption that $X_1$ is primitive. In particular,
    $X_1 \in \rsst{\alpha}$. Let $u$ be the primitive root of $X_2$. Then $u \in \rsst{\alpha}$ and, moreover, by the
    assumption $|X_1| > |X_2|$ it holds that $|u| < |X_1|$. By \autoref{prp:max_solution_suffix} the word
    $\lambda_2 = X_1 X_2$ is a suffix of the word $\mu_2 = X_1^2 X_2^2$. Therefore $X_1$ is a proper suffix of $X_1
    X_2$, so $X_1 X_2 = Z X_1$ for some nonempty word $Z$. A standard argument shows that $X_1$ is a suffix of some
    word in $X_2^+$ (see e.g., \cite[Proposition~1.3.4]{1983:combinatorics_on_words}). Consequently, $\mirror{X}_1$ is
    a prefix of a word in $\mirror{u}^+$. As $|u| < |X_1|$, \autoref{lem:standard_periods} implies that that
    $X_1 = \mirror{s}_{k,\ell}$ and $u = \mirror{s}_{k-1}$ for some $k \geq 2$ with $0 < \ell \leq a_k$.

    Suppose now that $\ell = a_k$. Then the word $X_1^2 X_2^2$ contains
    $\mirror{s}_{k-1} \mirror{s}_{k-2} \mirror{s}_{k-1}^{\, a_k+2}$ as a factor. Thus
    $s_{k-1}^{a_k+2} s_{k-2} s_{k-1} \in \Lang(\alpha)$. As $s_{k-1}$ is a prefix of $s_{k-2}s_{k-1}$, it follows that
    $s_{k-1}^{a_k+3} \in \Lang(\alpha)$ contradicting \autoref{prp:standard_index}. Therefore $\ell \leq a_k - 1$.

    Let us then relax the assumption that $X_1$ is primitive. Let $v$ be the primitive root of $X_1$, so that
    $X_1 = v^j$ for some $j \geq 1$. Consider now the Sturmian word $T^{(2j-2)|v|}(s) = v^2 X_2^2 \cdots$. By the above
    arguments $v = \mirror{s}_{k,\ell}$ for some $k \geq 2$ with $0 < \ell \leq a_k - 1$ and the primitive root of
    $X_2$ is $\mirror{s}_{k-1}$. Further, as $\ell \neq a_k$, it follows from \autoref{prp:standard_index} that $v^3
    \notin \Lang(\alpha)$. Thus $j=1$, that is, $X_1 = \mirror{s}_{k,\ell}$.
    
    It remains to show that $|X_3| > |X_1|$. Assume for a contradiction that $|X_3| \leq |X_1|$. It is not possible
    that $|X_3| < |X_2|$ as the preceding arguments show that then $X_2$ must be reversed semistandard word; however,
    $X_2$ is a power of the reversed standard word $\mirror{s}_{k-1}$. Hence by the maximality of $X_2$ we have that
    $|X_3| > |X_2|$. Let $X_3 = w^t$ with $w \in \rsst{\alpha}$ and $t \geq 1$. As $|X_2| < |X_3| \leq |X_1|$, we have
    that $|s_{k-1}| < t|w| \leq |s_{k,\ell}|$.
    
    Assume for a contradiction that $|w| < |s_{k-1}|$. If $w$ is semistandard, then \autoref{prp:standard_index}
    implies that $t = 1$, so $t|w| > |s_{k-1}|$ cannot hold. Thus $w$ is standard. If $w = \mirror{s}_0 = 0$, then
    clearly $t|w| > |s_{k-1}| \geq |s_1|$ cannot hold as the index of the factor $0$ in $\Lang(\alpha)$ is $a+1$. Thus
    $w \neq \mirror{s}_0$. Suppose first that $w = \mirror{s}_{k-2}$. Now
    \begin{align*}
      t|w| > |s_{k-1}| = a_{k-1}|s_{k-2}| + |s_{k-3}|,
    \end{align*}
    so $t > a_{k-1}$. Since $X_3^2 \in \Lang(\alpha)$, \autoref{prp:standard_index} implies that $2t \leq a_{k-1} + 2$.
    Therefore
    \begin{align*}
      a_{k-1}+2 \geq 2t > 2a_{k-1}
    \end{align*}
    implying that $a_{k-1} = 1$. However, if $a_{k-1} = 1$, then $a_{k-1} + 2$ is odd, so actually $2t < a_{k-1} + 2$.
    Then $a_{k-1} + 2 > 2t > 2a_{k-1}$, so $a_{k-1} < 1$; a contradiction. Suppose then that $w = \mirror{s}_{k-3}$.
    Now
    \begin{align*}
      t|w| > |s_{k-1}| \geq |s_{k-2}s_{k-3}| = |s_{k-3}^{a_{k-2}}s_{k-4}s_{k-3}| > (a_{k-2}+1)|s_{k-3}|,
    \end{align*}
    so $t > a_{k-2} + 1$. Like previously, as $X_3^2 \in \Lang(\alpha)$, \autoref{prp:standard_index} implies that
    $2t \leq a_{k-2} + 2$. Like above, we obtain that $a_{k-2} < 0$; a contradiction. Similar to above
    \begin{align*}
      |s_{k-1}| \geq (a_{k-2}+1)|s_{k-3}| + |s_{k-4}| \geq 2|s_{k-3}| + |s_{k-4}| > (2a_{k-3} + 1)|s_{k-4}|.
    \end{align*}
    As $2a_{k-3} + 1 \geq a_{k-3} + 2$, we conclude that $|s_{k-4}^{a_{k-3}+2}| < |s_{k-1}|$. Therefore by
    \autoref{prp:standard_index} it is not possible that $|w| \leq |s_{k-4}|$. In conclusion, it is not possible that
    $t|w| > |s_{k-1}|$. This is a contradiction.
    
    Now $|w| > |\mirror{s}_{k-1}|$ (by the maximality of $X_2$ it must be that $w \neq \mirror{s}_{k-1}$). Because
    $|w| \leq |\mirror{s}_{k,\ell}|$, we have that $w = \mirror{s}_{k,\ell'}$ for some $\ell'$ such that
    $0 < \ell' \leq l$. Since $\ell \neq a_k$, the word $w$ is semistandard so by \autoref{prp:standard_index} we have
    that $t = 1$. By \autoref{prp:max_solution_suffix} the word $\lambda_3 = X_1 X_2 X_3$ is a suffix of the word
    $\mu_3 = X_1^2 X_2^2 X_3^2$. It follows that
    $\mirror{s}_{k-2}\mirror{s}_{k-1}^{\,\ell+r} = \mirror{s}_{k-1}^{\,\ell+r-\ell'}\mirror{s}_{k-2}\mirror{s}_{k-1}^{\,\ell'}$
    where $r$ is such that $X_2 = \mirror{s}_{k-1}^{\,r}$. Therefore the words $\mirror{s}_{k-2}$ and
    $\mirror{s}_{k-1}$ commute; a contradiction. This final contradiction proves that $|X_3| > |X_1|$.
  \end{proof}

  \begin{corollary}\label{cor:solution_length_increasing}
    Let $s = X_1^2 X_2^2 \cdots$ be a Sturmian word of type $A$ with slope $\alpha$ factorized as a product of maximal
    solutions $X_i$. Then $\liminf_{i \to \infty} |X_i| = \infty$.
  \end{corollary}
  \begin{proof}
    This follows from \autoref{prp:next_shorter}: if $|X_{i+1}| < |X_i|$ for some $i \geq 1$, then $|X_{i+2}| > |X_i|$.
  \end{proof}

  \section{The Square Root of the Fibonacci Word}\label{sec:fibonacci}
  In this section we prove a formula for the square root of the Fibonacci word. To do this we factorize the Fibonacci
  word as a product of maximal solutions to \eqref{eq:square}.

  We denote by $\Phi$ the slope of the Fibonacci word, that is, $\Phi = [0;2,1,1,\ldots]$. Further, we set
  \begin{align*}
    t_k = 
    \begin{cases}
      01, & \text{ if } k \text{ is even }, \\
      10, & \text{ if } k \text{ is odd }.
    \end{cases}
  \end{align*}
  We need two lemmas specific to the slope $\Phi$.

  \begin{lemma}\label{lem:fibonacci_formula1}
    For the standard words of slope $\Phi$ it holds that
    $t_k s_k s_{k+1} s_{k+2} = \mirror{s}_{k+2}^{\, 2} t_{k+1}$ for all $k \geq 0$.
  \end{lemma}
  \begin{proof}
    The case $k = 0$ is verified directly:
    $t_0 s_0 s_1 s_2 = 01\cdot0\cdot01\cdot010 = (010)^2 \cdot 10 = \mirror{s}_2^{\, 2} t_1$.
    Let then $k \geq 1$. There exists a palindrome $P_k$ such that $s_k = P_k \mirror{t}_k$ for all $k \geq 1$ (see
    e.g. \cite[Lemma~2.2.8]{2002:algebraic_combinatorics_on_words}). Now
    \begin{align*}
      t_k s_k s_{k+1} s_{k+2} &= t_k P_k \mirror{t}_k s_{k+1} s_{k+2} = \mirror{s}_k \mirror{t}_k P_{k+1}
      \mirror{t}_{k+1} s_{k+2} = \mirror{s}_k t_{k+1} P_{k+1} \mirror{t}_{k+1} s_{k+2} \\
      &= \mirror{s}_k \mirror{s}_{k+1} t_{k+1} P_{k+2} \mirror{t}_{k+2} = \mirror{s}_k \mirror{s}_{k+1} \mirror{s}_{k+2}
      \mirror{t}_{k+2} = \mirror{s}_{k+2}^{\, 2} t_{k+1},
    \end{align*}
    which proves the claim.
  \end{proof}

  \begin{lemma}\label{lem:fibonacci_formula2}
    For the standard words of slope $\Phi$ it holds that
    $s_{3k + 4} = \prod_{i = 0}^k \mirror{s}_{3i + 2}^{\, 2} \cdot t_{k+1}$ for all $k \geq 0$.
  \end{lemma}
  \begin{proof}
    If $k = 0$, then $s_4 = 01001010 = \mirror{s}_2^{\, 2} t_1$. Let then $k \geq 1$. Now
    \begin{align*}
      s_{3k+4} &= s_{3k+3}s_{3k+2} = s_{3k+2}s_{3k+1}s_{3k+2} = s_{3k+1}s_{3k}s_{3k+1}s_{3k+2} \\
      &= s_{3(k-1)+4}s_{3k}s_{3k+1}s_{3k+2}
      = \prod_{i=0}^{k-1} \mirror{s}_{3i+2}^{\, 2} \cdot t_k s_{3k}s_{3k+1}s_{3k+2}
    \end{align*}
    where the last equality follows by induction. By applying \autoref{lem:fibonacci_formula1} we obtain that
    \begin{align*}
      s_{3k+4} = \prod_{i=0}^{k-1} \mirror{s}_{3i+2}^{\, 2} \cdot \mirror{s}_{3k+2}^{\, 2} t_{k+1} = \prod_{i = 0}^k \mirror{s}_{3i + 2}^{\, 2} \cdot t_{k+1},
    \end{align*}
    which proves the claim.
  \end{proof}

  As an immediate corollary to \autoref{lem:fibonacci_formula2} we obtain a formula for the square root of the
  Fibonacci word.

  \begin{theorem}\label{thm:fibonacci_square_root}
    For slope $\Phi$ we have that
    \begin{align*}
      c_\Phi = \prod_{i=0}^\infty \mirror{s}_{3i+2}^{\, 2} \quad \text{ and } \quad
      \sqrt{c_\Phi} = s_{\frac{1}{2},\Phi} = \prod_{i=0}^\infty \mirror{s}_{3i+2}.
    \end{align*}
  \end{theorem}

  The preceding arguments are very specific to the Fibonacci word. The reader might wonder if formulas for the square
  roots of other standard Sturmian words exist. Surely, for some specific words such formulas can be derived, but we
  believe no general factorization for the square roots of standard Sturmian words can be given. Let us give some
  arguments supporting our belief.

  Let $s = X_1^2 X_2^2 \cdots$ be a standard Sturmian word of slope $\alpha$ factorized as a product of maximal solutions
  to \eqref{eq:square}. The word $s$ begins with the word $0^\oa 1$. Therefore if $\oa > 1$, then
  $X_1 = 0^{\lfloor \oa/2 \rfloor}$. Thus if $\oa > 1$, then $X_2$ begins with $0$ if and only if $\oa$ is odd. Because
  of the asymmetry of the letters $0$ and $1$ in the minimal squares of slope $\alpha$ \eqref{eq:min_squares}, the
  parity of the parameter $\oa$ greatly influences the remaining words $X_i$. Moreover, it is not just the partial
  quotient $a_1$ which influences the factorization. Suppose for instance that $a_1 = 2$ and $a_2 = 1$.
  \autoref{tbl:changes_1} shows how the values of the partial quotients $a_3$ and $a_4$ affect the words $X_i$. The
  cell of the table tells to which squares of reversed standard words the words $X_1, X_2$ and $X_3$ correspond to. For
  example if $a_3 = 2$ and $a_4 = 1$, then the standard Sturmian word of slope $[0;2,1,2,1,\ldots]$ begins with
  $\mirror{s}_2^{\, 2} \mirror{s}_4^{\, 2} \mirror{s}_7^{\, 2}$. \autoref{tbl:changes_2} tells the first letter of the
  corresponding word $X_4^2$. As can be observed from \autoref{tbl:changes_2}, the first letter of $X_4^2$ varies when
  $a_3$ and $a_4$ vary. Because of the asymmetry, it is thus expected that slight variation in partial quotients
  drastically changes the factorization as a product of maximal solutions to \eqref{eq:square}. Since similar behavior
  is expected from the rest of the partial quotients, it seems to us that no nice formula (like e.g., the formula of
  \autoref{thm:fibonacci_square_root}) can be given for the square root of a standard Sturmian word in terms of
  reversed standard words.

  \begin{table}
    \centering
    \begin{tabular}{|l|l|l|l|}
      \hline
      \diagbox{$a_4$}{$a_3$} & $1$ & $2$ & $3$ \\
      \hline
      $1$ & $2,5,8$ & $2,4,7$ & $2,5,8$ \\
      \hline
      $2$ & $2,3,6$ & $2,4,7$ & $2,3,6$ \\
      \hline
      $3$ & $2,3,5$ & $2,4,7$ & $2,3,5$ \\
      \hline
    \end{tabular}
    \caption{How $X_1$, $X_2$, and $X_3$ are affected when $a_3$ and $a_4$ vary in the case that $a_1 = 2$ and $a_2 = 1$.}
    \label{tbl:changes_1}
  \end{table}
  \begin{table}
    \centering
    \begin{tabular}{|l|l|l|l|}
      \hline
      \diagbox{$a_4$}{$a_3$} & $1$ & $2$ & $3$ \\
      \hline
      $1$ & $1$ & $0$ & $1$ \\
      \hline
      $2$ & $1$ & $0$ & $1$ \\
      \hline
      $3$ & $0$ & $0$ & $0$ \\
      \hline
    \end{tabular}
    \caption{How the first letter of $X_4$ varies when $a_3$ and $a_4$ vary in the case that $a_1 = 2$ and $a_2 = 1$.}
    \label{tbl:changes_2}
  \end{table}

  De Luca and Fici proved a nice formula for a certain shift of a standard Sturmian word
  \cite[Theorem~18]{2013:open_and_closed_prefixes_of_sturmian_words}.

  \begin{proposition}
    Let $c_\alpha$ be the standard Sturmian word of slope $\alpha = [0;\oa+1,\ob+1,\ldots]$. Then
    \begin{align*}
      c_\alpha = 0^\oa10^{\oa-1} \prod_{k=1}^\infty \mirror{s}_k^{\, 2}.
    \end{align*}
  \end{proposition}

  As a corollary of this theorem we obtain that the word $\sqrt{T^{2\oa}(c_\alpha)} = \prod_{k=1}^\infty \mirror{s}_k$
  is a Sturmian word of slope $\alpha$ with intercept $\psi(\{(2\oa+1)\alpha\}) = \oa\alpha$. We have thus shown that
  \begin{align*}
    c_\alpha = 0^{\oa-1} \prod_{k=1}^\infty \mirror{s}_k.
  \end{align*}
  In particular, we obtain the well-known result that the Fibonacci infinite word is a product of the reversed
  Fibonacci words.

  \section{A Curious Family of Subshifts}\label{sec:counter_example}
  In this section we construct a family of linearly recurrent and optimal squareful words which are not Sturmian but
  are fixed points of the (more general) square root map. Moreover, we show that any subshift $\Omega$ generated by
  such a word has a curious property: for every $w \in \Omega$ either $\sqrt{w} \in \Omega$ or $\sqrt{w}$ is periodic.

  It is evident from \autoref{prp:optimal_squareful_characterization} that Sturmian words are a proper subclass of
  optimal squareful words. As Sturmian words have the exceptional property that their language is preserved under the
  square root map, it is natural to ask if other optimal squareful words can have this property. We show that, indeed,
  such words exist by an explicit construction. The idea behind the construction is to mimic the structure of the
  Sturmian words $01c_\alpha$ and $10c_\alpha$. The simple reason why these words are fixed points of the square root
  map (thus preserving the language) is that they have arbitrarily long squares of solutions to \eqref{eq:square} as
  prefixes. Thus to obtain a fixed point of the square root map, it is sufficient to find a sequence $(u_k)$ of
  solutions to \eqref{eq:square} with the property that $u_k^2$ is a proper prefix of $u_{k+1}^2$ for all $k \geq 1$.
  Let us show how such a sequence can be obtained.

  Let $S$ be a fixed primitive solution to \eqref{eq:square} in the language of some Sturmian word with slope
  $[0;\oa+1,\ob+1, \ldots]$ such that $|S| > |S_6|$. In particular, $S$ has the word $S_6 = 10^{\oa+1}(10^\oa)^{\ob+1}$
  as a proper suffix. Recall from the proof of \autoref{lem:product_of_squares} that $|S| \geq |S_5S_6|$. We denote the
  word $L(S)$ simply by $L$. Using the word $S$ as a seed solution, we produce a sequence $(\gamma_k)$ of primitive
  solutions to \eqref{eq:square} defined by the recurrence
  \begin{align}\label{eq:gamma_recurrence}
    \gamma_1 = S, \quad \gamma_{k+1} = L(\gamma_k) \gamma_k^2 \quad \text{ for } k \geq 2.
  \end{align}
  We need to prove that the sequence $(\gamma_k)$ really is a sequence of primitive solutions to \eqref{eq:square}.
  Before showing this, let us define
  \begin{align}\label{eq:gamma}
    \Gamma_1 = \lim_{k \to \infty} \gamma_{2k} \quad \text{ and } \quad \Gamma_2 = \lim_{k \to \infty} \gamma_{2k + 1}.
  \end{align}
  The limits exist as $\gamma_k^2$ is always a prefix of $\gamma_{k+2}$. Hence both $\Gamma_1$ and $\Gamma_2$ have
  arbitrarily long squares of words in the sequence $(\gamma_k)$ as prefixes. Observe also that
  $\Lang(\Gamma_1) = \Lang(\Gamma_2)$. As there is not much difference between $\Gamma_1$ and $\Gamma_2$ in terms of
  structure, we set $\Gamma$ to be either of these words.
 
  Taking for granted that the sequence $(\gamma_k)$ is a sequence of solutions to \eqref{eq:square}, we see that
  $\sqrt{\Gamma} = \Gamma$. Note that we also need to ensure that the word $\Gamma$ is optimal squareful for the square
  root map to make sense.

  Next we aim to prove the following.

  \begin{proposition}\label{prp:gamma_solutions}
    The word $\gamma_k$ is a primitive solution to \eqref{eq:square} in $\Lang(\oa,\ob)$ for all $k \geq 1$.
  \end{proposition}

  Recall from \autoref{sec:word_equation_characterization} that the language $\Lang(\oa,\ob)$ consists of all factors
  of the infinite words in the language
  \begin{align*}
    (10^{\oa+1}(10^\oa)^\ob + 10^{\oa+1}(10^\oa)^{\ob+1})^\omega = (S_5 + S_6)^\omega.
  \end{align*}

  Before we can prove \autoref{prp:gamma_solutions}, we need to know that the words $\gamma_k$ are primitive and that
  they are factors of some optimal squareful word with parameters $\oa$ and $\ob$.

  \begin{lemma}\label{lem:gamma_primitive}
    The word $\gamma_k$ is primitive for all $k \geq 1$.
  \end{lemma}
  \begin{proof}
    We proceed by induction. By definition $\gamma_1$ is primitive. Let $k \geq 1$, and suppose for a contradiction that
    $\gamma_{k+1}$ is not primitive; that is, $\gamma_{k+1} = L(\gamma_k) \gamma_k^2 = z^n$ for some primitive word $z$
    and $n > 1$. If $n = 2$, then obviously $|\gamma_k|$ must be even, and the suffix of $\gamma_k$ of length
    $|\gamma_k|/2$ must be a prefix of $\gamma_k$. This contradicts the primitivity of $\gamma_k$. The case
    $n = 3$ would clearly imply that $\gamma_k = L(\gamma_k)$, which is not possible. Hence $n > 3$, and
    further $|z| < |\gamma_k|$. As $\gamma_k^2$ is a suffix of some word in $z^+$, it follows that $z = uv$
    where $vu$ is a suffix of $\gamma_k$. On the other hand, $z$ is a suffix of $\gamma_k$, so $uv = vu$. Since
    $z$ is primitive, the only option is that $u$ is empty. Therefore $\gamma_k \in z^+$; a contradiction with the
    primitivity of $\gamma_k$.
  \end{proof}

  \begin{lemma}\label{lem:gamma_optimal_squareful}
    We have that $\gamma_k, L(\gamma_k) \in \Lang(\oa,\ob)$ for all $k \geq 1$.
  \end{lemma}
  \begin{proof}
    For a suitable slope $\alpha = [0;\oa+1,\ob+1, \ldots]$, either of the words $S$ and $L$ is a reversed standard
    word of slope $\alpha$. Thus by \autoref{thm:complete_characterization} both $S^2$ and $L^2$ are in
    $\Lang(\alpha)$, so $S^2, L^2 \in \Lang(\oa,\ob)$.

    We clearly have that $\gamma_1 \in \Lang(\oa,\ob)$. Note that by the assumption $|S| > |S_6|$ both of the words $S$
    and $L$ have the word $s = S_6 = 10^{\oa+1}(10^\oa)^{\ob+1}$ as a proper suffix. Write $S = us$. Since $s$ begins
    with $10^{\oa+1}$ and $S^2$ has $sus$ as a suffix, it follows that $us \in (S_5+S_6)^+$. Using the fact that
    $L \in \Lang(\oa,\ob)$, we see that $\gamma_2 = LSS = L(u)s(us)^2 \in \Lang(\oa,\ob)$. Clearly
    $L(\gamma_2) = S^3 = (us)^3 \in \Lang(\oa,\ob)$. Proceeding by induction we may assume that $k \geq 2$ and
    $\gamma_k, L(\gamma_k) \in \Lang(\oa,\ob)$. Since $\gamma_k$ has either $S$ or $L$ as a prefix, it can be written
    that $\gamma_k = vszs$ with $|vs| = |S|$. It follows that $sz \in (S_5+S_6)^+$. Since $svs$ is a suffix of either
    $S^2$ or $L^2$, we have that $sv \in (S_5+S_6)^+$. Therefore $svsz \in (S_5+S_6)^+$. As
    $L(\gamma_k) = L(vsz)s \in \Lang(\oa,\ob)$, we have that $L(vsz)$ is a suffix of some word in $(S_5+S_6)^+$.
    Overall, the word $\gamma_{k+1} = L(vsz)(svsz)^2 s$ is in $\Lang(\oa,\ob)$. Clearly then must the word
    $L(\gamma_{k+1} = (vszs)^3 = vsz(svsz)^2 s$ also be in $\Lang(\oa,\ob)$.
  \end{proof}

  Note that without the assumption $|S| > |S_6|$ the conclusion of the above lemma fails to hold. If
  $S = S_6 = 10^{\oa+1}(10^\oa)^{\ob+1}$, then $L = 0(10^\oa)^{\ob+2}$ and
  $LS = 0(10^\oa)^{\ob+2}10^{\oa+1}(10^\oa)^{\ob+1}$. Therefore $LS~\notin~\Lang(\oa,\ob)$, and consequently
  $\gamma_2 = LS^2 \notin \Lang(\oa,\ob)$.

  \begin{proof}[Proof of \autoref{prp:gamma_solutions}]
    We proceed by induction. By \autoref{lem:gamma_primitive} the word $\gamma_k$ is primitive for all $k \geq 1$.
    \autoref{lem:gamma_optimal_squareful} tells that both of the words $\gamma_k$ and $L(\gamma_k)$ are in
    $\Lang(\oa,\ob)$ for all $k \geq 1$. By definition both $\gamma_1$ and $L(\gamma_1)$ are solutions to
    \eqref{eq:square}. We may thus assume that $k \geq 1$ and both $\gamma_k$ and $L(\gamma_k)$ are solutions to
    \eqref{eq:square}. It follows from \autoref{lem:product_of_squares} that
    \begin{align*}
      \gamma_k L(\gamma_k) \in \Pi(\oa,\ob) \ \text{ and } \ \sqrt{\gamma_k L(\gamma_k)} = \gamma_k.
    \end{align*}
    Since $L(\gamma_k)$ is a solution to \eqref{eq:square}, \autoref{lem:product_of_squares} also implies that
    \begin{align*}
      L(\gamma_k) \gamma_k \in \Pi(\oa,\ob) \ \text{ and } \ \sqrt{L(\gamma_k) \gamma_k} = L(\gamma_k).
    \end{align*}
    Because
    \begin{align*}
      \gamma_{k+1}^2 = L(\gamma_k) \gamma_k \cdot \gamma_k L(\gamma_k) \cdot \gamma_k^2,
    \end{align*}
    we obtain that
    \begin{align*}
      \gamma_{k+1}^2 \in \Pi(\oa,\ob) \ \text{ and } \
      \sqrt{\gamma_{k+1}^2} = \sqrt{L(\gamma_k) \gamma_k} \sqrt{\gamma_k L(\gamma_k)} \sqrt{\gamma_k^2} = L(\gamma_k)\gamma_k\gamma_k = \gamma_{k+1}.
    \end{align*}
    This proves that $\gamma_{k+1}$ is a solution to \eqref{eq:square}. Consider next the word
    $L(\gamma_{k+1}) = \gamma_k^3$. Because $(L(\gamma_{k+1}))^2 = (\gamma_k^2)^3$, it is evident that
    \begin{align*}
      (L(\gamma_{k+1}))^2 \in \Pi(\oa,\ob) \ \text{ and } \
      \sqrt{(L(\gamma_{k+1}))^2} = \gamma_k^3 = L(\gamma_{k+1}).
    \end{align*}
    Therefore also $L(\gamma_{k+1})$ is a solution to \eqref{eq:square}. The conclusion follows.
  \end{proof}

  As we remarked earlier, we have now proved that $\Gamma$ is a fixed point of the square root map. Next we show that
  the word $\Gamma$ is aperiodic, linearly recurrent, and not Sturmian.

  \begin{lemma}\label{lem:gamma_not_sturmian}
    The word $\gamma_2^2$ is not a factor of any Sturmian word.
  \end{lemma}
  \begin{proof}
    By definition $\gamma_2 = LS^2$. Write $S = xyw$ and $L = yxw$ for some word $w$ and distinct letters $x$ and
    $y$. Now $\gamma_2^2 = xyxw(xyw)^2yxw(xyw)^2$, so the word $\gamma_2^2$ has factors $xwx$ and $ywy$. Hence
    $\gamma_2^2$ is not balanced, and it cannot be a factor of any Sturmian word.
  \end{proof}

  \begin{lemma}\label{lem:gamma_linearly_recurrent}
    The word $\Gamma$ is aperiodic and linearly recurrent.
  \end{lemma} 
  \begin{proof}
    The recurrence \eqref{eq:gamma_recurrence} and the definition \eqref{eq:gamma} of $\Gamma$ show that for all
    $k \geq 1$ the word $\Gamma$ is a product of the words $\gamma_{k+1} = L(\gamma_k)\gamma_k^2$ and
    $L(\gamma_{k+1}) = \gamma_k^3$ such that between two occurrences of $L(\gamma_{k+1})$ there is always $\gamma_k^2$
    or $\gamma_k^5$. From this it follows that the return time of a factor of $\Gamma$ of length $\gamma_k$ is at most
    the return time of the factor $L(\gamma_k)$, which is at most $6|\gamma_k|$. Let then $w$ be a factor of $\Gamma$
    such that $|\gamma_k| < |w| \leq |\gamma_{k+1}|$. Since $w$ is a factor of some factor of $\Gamma$ of length
    $|\gamma_{k+1}|$, it follows that the return time of $w$ is at most $6|\gamma_{k+1}|$. Now
    $6|\gamma_{k+1}| = 18|\gamma_k| < 18|w|$ proving that $\Gamma$ is linearly recurrent.

    The preceding shows that $\gamma_k$ is followed in $\Lang(\Gamma)$ by both $\gamma_k$ and $L(\gamma_k)$. As the
    first letters of $\gamma_k$ and $L(\gamma_k)$ are distinct, the factor $\gamma_k$ is right special. Thus
    $\Lang(\Gamma)$ contains arbitrarily long right special factors, so $\Gamma$ must be aperiodic.
  \end{proof}

  Since linearly recurrent words have linear factor complexity
  \cite[Theorem~24]{1999:substitution_dynamical_systems_bratteli_diagrams}, it follows from
  \autoref{lem:gamma_linearly_recurrent} that $\Gamma$ has linear factor complexity.

  We observed in the previous proof that the word $\Gamma$ is a product of the words $S$ and $L$ such that between two
  occurrences of $L$ in this product there is always $S^2$ or $S^5$. Since $S$ and $L$ are primitive, any word
  $w \in \Lang(\Gamma)$ which is a product of the words $S$ and $L$ such that $|w| \geq 6|S|$ must synchronize to the
  factorization of $\Gamma$ as a product of the words $S$ and $L$. That is, for any factorization
  $\Gamma = uw\Gamma'$ we must have that $|u|$ is a multiple of $|S|$.

  \begin{theorem}
    The word $\Gamma$ is a non-Sturmian, linearly recurrent optimal squareful word which is a fixed point of the square
    root map.
  \end{theorem}
  \begin{proof}
    The fact that $\Gamma$ is optimal squareful and linearly recurrent follows from Lemmas
    \ref{lem:gamma_optimal_squareful} and \ref{lem:gamma_linearly_recurrent}. The argument outlined at the beginning of
    this section shows that $\Gamma$ is a fixed point of the square root map as by \autoref{prp:gamma_solutions} the
    words $\gamma_k$ which occur as square prefixes in $\Gamma$ are solutions to \eqref{eq:square}. Finally, $\Gamma$
    contains the factor $\gamma_2^2$, so $\Gamma$ is not Sturmian by \autoref{lem:gamma_not_sturmian}.
  \end{proof}

  Denote by $\Omega$ the subshift consisting of the infinite words having language $\Lang(\Gamma)$. As $\Gamma$ is
  linearly recurrent, it is uniformly recurrent, so the subshift $\Omega$ is minimal. The rest of this section is
  devoted to proving the result mentioned in the beginning of this section.

  \begin{theorem}\label{thm:preserved_periodic}
    For all $w \in \Omega$ either $\sqrt{w} \in \Omega$ or $\sqrt{w}$ is (purely) periodic with minimal period
    conjugate to $S$. Moreover, there exists words $u, v \in \Omega$ such that $\sqrt{u} \in \Omega$ and $\sqrt{v}$ is
    periodic.
  \end{theorem}

  This result is very surprising since it is contrary to the plausible hypothesis that an aperiodic word must map to an
  aperiodic word under the square root map.

  It is not difficult to prove \autoref{thm:preserved_periodic} for words in $\Omega$ which are products of the words
  $S$ and $L$. We prove this special case next in \autoref{lem:product_sl_preserved}. However, difficulties arise since
  a word in $\Omega$ can start in an arbitrary position of an infinite product of $S$ and $L$. There are certain
  well-behaved positions in $S$ and $L$ which are easier to handle. \autoref{thm:preserved_periodic} is proved for
  these special positions in \autoref{lem:periodic_image}. The rest of the effort is in demonstrating that all the
  other cases can be reduced to these well-behaved cases. We begin by proving the easier cases, and we conclude with
  the reductions.

  \begin{lemma}\label{lem:product_sl_preserved}
    If a word $w \in \Omega$ can be written as a product of the words $S$ and $L$, then $\sqrt{w} \in \Omega$.
  \end{lemma}
  \begin{proof}
    Any word $u$ which is a product of the words $S$ and $L$ can be naturally written as a binary word $\overline{u}$
    over the alphabet $\{S, L\}$. If such a word $\overline{u}$ has even length, then it is a word over the alphabet
    $A = \{SS, SL, LS, LL\}$. Using the fact that $\sqrt{SS} = S$, $\sqrt{SL} = S$, $\sqrt{LS} = L$, and
    $\sqrt{LL} = L$ (see \autoref{lem:product_of_squares}), we can define a square root for a word over $A$.

    The word $\gamma_k^2$ is a prefix of $\Gamma$ for all $k \geq 1$. Thus $\gamma_k$ has occurrences at positions $0$
    and $|\gamma_k|$ of $\Gamma$. Clearly $|\gamma_k| = 3^{k-1}|S|$, so the word $\overline{\gamma}_k$ occurs in
    $\overline{\Gamma}$ in an even and in an odd position.

    Let $v$ be a prefix of $w$ of length $|v| = 2n|S|$ for some $n \geq 1$, so $\overline{v}$ is a word over $A$. Since
    $v$ is a prefix of $w$, the word $v$ is a factor of some $\gamma_k$. Since $\overline{\gamma}_k$ occurs in
    $\overline{\Gamma}$ in an even and in an odd position, the word $\overline{v}$ occurs in an even position in
    $\overline{\Gamma}$. Hence $\overline{\Gamma}$ can be factored as $\overline{\Gamma} = z\overline{v}t$ where $z$
    and $t$ are finite or infinite words over $A$. Since $\Gamma$ is a fixed point of the square root map, we have that
    $\overline{\Gamma} = \sqrt{z\vphantom{\overline{v}}}\sqrt{\overline{v}}\sqrt{t\vphantom{\overline{v}}}$. Hence
    $\sqrt{v} \in \Lang(\Gamma)$. It follows that $\Lang(\sqrt{w}) \subseteq \Lang(\Gamma)$, so $\sqrt{w} \in \Omega$.
  \end{proof}

  \begin{definition}
    Let $w$ be a word and $\ell$ be an integer such that $0 < \ell < |w|$. If the factor of $w^3$ of length $|w^2|$
    starting at position $\ell$ can be written as a product of minimal squares $X_1^2, \ldots, X_n^2$, then we say that
    the position $\ell$ of $w$ is \emph{repetitive}. If in addition
    $|X_1^2 \cdots X_m^2| \neq |w| - \ell, |w^2| - \ell$ for all $m$ such that $1\leq m\leq n$, then we say that the
    position $\ell$ is \emph{nicely repetitive}.
  \end{definition}

  For example if $\oa=1$, $\ob=0$, and $S = 1001001010010$, then the position $1$ of $S$ is repetitive as the factor
  $00100101001010010010100101$ of $S^3$ of length $|S^2| = 26$ starting at position $1$ is in $\Pi(\oa,\ob)$. This
  position is not nicely repetitive as $|0^2 \cdot (10010)^2| = 12 = |S| - 1$. The position $2$ of $S$, however, can be
  checked to be nicely repetitive. The position $4$ of $S$ is not repetitive as the factor $00101001010010010100101001$
  of length $26$ starting at position $4$ is not in $\Pi(\oa,\ob)$.

  In the upcoming proof of \autoref{thm:preserved_periodic} we will show that if $w \in \Omega$ is a product of the
  words $S$ and $L$ and $\ell$ is a nicely repetitive position of $S$, then the word $\sqrt{T^\ell(w)}$ is always
  periodic. On the other hand, we show that if $\ell$ is not a nicely repetitive position then $\sqrt{T^\ell(w)}$ is
  always in $\Omega$.

  Next we identify some good positions in the suffix $S_6$ of $S$. As we observed in the proof of
  \autoref{lem:product_of_squares}, the suffix $S_6$ of $S$ restricts locally how a
  factorization of a word as a product of minimal squares continues after an occurrence of $S_6$. Consider a product
  $X_1^2 \cdots X_n^2$ of minimal squares which has an occurrence of $S_6$ at position $\ell$. Then for some $m \in
  \{1,\ldots,n\}$ the minimal square $X_m^2$ must begin at some of the positions $\ell, \ell+1, \ldots, \ell+|S_6|-1$.
  Otherwise some minimal square would have $S_6$ as an interior factor; yet no such minimal square exists. Among the
  positions $\ell, \ell+1, \ldots, \ell+|S_6|-1$ we are interested in the largest position where a minimal square may
  begin. Let
  \begin{align*}
    \B = \{\ell \in \{0,\ldots,|S_6|-1\}\colon \text{no square of length at most } |S_6| - \ell \text{ begins at position } \ell \text{ of } S_6\}.
  \end{align*}
  It is straightforward to see that
  \begin{align*}
    \B = \{|S_6|-|S_6|,|S_6|-|S_4|,|S_6|-|S_3|,|S_6|-|S_1|\}.
  \end{align*}
  We are interested in those positions of the suffix $S_6$ of $S$ where no minimal square begins. Hence we define
  \begin{align*}
    \B_S = \{\ell: \ell-|S|+|S_6| \in \B\} = \{|S|-|S_6|, |S|-|S_4|, |S|-|S_3|, |S|-|S_1|\}.
  \end{align*}
  A consequence of the definitions is that if $\ell$ is a position of $S$ such that $\ell \notin \B_S$, then there
  exists $\ell' \in \B_S \cup \{|S|\}$ such that $S[\ell,\ell'-1] \in \Pi(\oa,\ob)$. This fact is used later several
  times.

  \begin{lemma}\label{lem:periodic_image}
    Suppose that $w \in \Omega$ can be written as a product of the words $S$ and $L$. Assume that the position
    $\ell \in \B_S$ is nicely repetitive. Let the prefix of $T^\ell(w)$ of length $|S^2|$ be factorized as a product of
    minimal squares $X_1^2 \cdots X_n^2$. Then the word $\sqrt{T^\ell(w)}$ is periodic with minimal period
    $X_1 \cdots X_n$. Moreover, $X_1 \cdots X_n$ is conjugate to $S$.
  \end{lemma}
  \begin{proof}[Proof Sketch]
    As $\ell$ is repetitive, the factor $u$ of length $|S^2|$ of $S^3$ starting at position $\ell$ is in
    $\Pi(\oa,\ob)$. If we substitute the middle $S$ in $S^3$ with $L$, then an application of
    \autoref{lem:exchange_squares} shows that the factor of length $|S^2|$ of $SLS$ starting at position $\ell$ is
    still in $\Pi(\oa,\ob)$ and that the square root of this factor coincides with the square root of $u$ (here we need
    that $\ell \in \B_S$). Further analysis shows that if we substitute the words $S$ in $S^3$ in any way, then the
    square root of the factor of length $|S^2|$ beginning at position $\ell$ is unaffected. Since $\ell$ is repetitive,
    the prefix of $T^{\ell+|S^2|}(w)$ of length $|S^2|$ is again in $\Pi(\oa,\ob)$ and has the same square root, and so
    on. Thus $\sqrt{T^\ell(w)}$ is periodic. Since both the square of the period and $S^2$ occur in a suitable Sturmian
    word; having equals lengths, they must be conjugate by \autoref{prp:conjugate_square}.
  \end{proof}
  \begin{proof}
    We have that $|S| \geq |S_5 S_6|$, so $\ell > 1$. Let $u$ be the suffix of $S$ of length $|S|-\ell$. Since $\ell$
    is repetitive, the factor $v$ of $S^3$ of length $|S^2|$ starting at position $\ell$ can be factorized as a product
    of minimal squares $Y_1^2 \cdots Y_m^2$. We have that $|Y_1^2| > |u|$ because $\ell \in \B_S$.
    
    Next we consider how the situation changes if any of the words $S$ in $S^3$ is substituted with $L$. Substituting
    the first $S$ with $L$ does not affect the product as $\ell > 1$. Suppose then that the second word $S$ is substituted
    with $L$. By applying \autoref{lem:exchange_squares} to the words $u$ and $S$ with $X = Y_1$, we see that the
    factor of length $|S^2|$ of $SLS$ starting at position $\ell$ can still be factorized as a product of minimal
    squares and that the square root of this factor coincides with the square root of $v$. Consider next what happens
    when the third word $S$ is substituted with $L$. Let
    \begin{align*}
      r = \max\{i \in \{1,\ldots,m\}\colon |Y_1^2 \cdots Y_i^2| \leq |S^2| - \ell\}.
    \end{align*}
    Set $\ell' = \ell + |Y_1^2 \cdots Y_r^2| - |S|$. Since $\ell$ is nicely repetitive, we have that $\ell' < |S|$. By
    the maximality of $r$ and the definition of the set $\B_S$, we thus have that $\ell' \in \B_S$. Applying
    \autoref{lem:exchange_squares} to the suffix of $S$ of length $|S|-\ell'$ and $S$ with $X = Y_{r+1}$ we obtain,
    like above, that the product of minimal squares is affected but the square root is not. Substituting the second and
    third words $S$ with $L$ gives the same result: first proceed as above and substitute the second word $S$ and then
    make the second substitution like above but apply \autoref{lem:exchange_squares} for the word $L$ instead of $S$.

    We have concluded that however we substitute the words $S$ in $S^3$, the square root of the factor of length
    $|S^2|$ beginning at position $\ell$ never changes. The word $w$ is obtained from the word $S^\omega$ by
    substituting some of the words $S$ with $L$. By the preceding, the prefix of $T^\ell(w)$ of length $|S^2|$ can be
    factorized as a product of minimal squares $X_1^2 \cdots X_n^2$. Since $\ell$ is repetitive, the prefix of
    $T^{\ell+|S^2|}(w)$ of length $|S^2|$ can also be factorized as a product of some minimal squares (perhaps
    different) but the square root still equals $X_1 \cdots X_n$. By repeating this observation we see that
    \begin{align*}
      \sqrt{T^\ell(w)} = (X_1 \cdots X_n)^\omega.
    \end{align*}

    By our choice of $S$ we have that $S \in \{\mirror{s}_k, L(\mirror{s}_k)\}$ where $\mirror{s}_k$ is a reversed
    standard word of some slope $\alpha = [0;\oa+1,\ob+1, \ldots]$. Let $\beta = [0; b_1, b_2, \ldots]$ be a number such
    that $a_i = b_i$ for $1 \leq i \leq k$ and $b_{k+1} \geq 5$. Then by the definition of standard words
    $S^5 \in \Lang(\beta)$. By the preceding, the prefix of $T^\ell(S^5)$ of length $|S^4|$ can be written as a product
    of minimal squares, and the square root of these minimal squares equals $(X_1 \cdots X_n)^2$. Since the square root
    of a Sturmian word of slope $\beta$ is a Sturmian word of slope $\beta$, we have that
    $(X_1 \cdots X_n)^2 \in \Lang(\beta)$. As $|X_1 \cdots X_n| = |S|$, it follows by \autoref{prp:conjugate_square}
    that $X_1 \cdots X_n$ is conjugate to $S$. Since $S$ is primitive, so is $X_1 \cdots X_n$, and hence the period
    $X_1 \cdots X_n$ is minimal.
  \end{proof}

  \begin{lemma}\label{lem:one_nicely_repetitive}
    Every seed solution $S$ has at least one nicely repetitive position $\ell$ such that $\ell \in \B_S$.
  \end{lemma}
  \begin{proof}
    Suppose that $S = \mirror{s}_{k,i}$ for some $k \geq 3$ and $0 < i \leq a_k$. It is sufficient to show that
    $r = |\mirror{s}_{k,i-1}|$ is a nicely repetitive position of $S$. If $r \notin \B_S$, then there exists
    $r' \in \B_S$ such that $S[r,r'-1] \in \Pi(\oa,\ob)$. Since the position $r$ is nicely repetitive, so must $r'$ be.
    If $S = L(\mirror{s}_{k,i})$, then as $r > 1$, an application of \autoref{lem:exchange_squares} shows that the
    conclusion holds also in this case.
  
    Observe that the word $\mirror{s}_{k,i-1}$ is both a prefix and a suffix of $S$. Using the fact that
    $\mirror{s}_{k-2}\mirror{s}_{k-3} = L(\mirror{s}_{k-3}\mirror{s}_{k-2})$ we obtain that
    \begin{align*}
      S^3 &= \mirror{s}_{k,i-1} \mirror{s}_{k-1} \mirror{s}_{k-2} \mirror{s}_{k-1}^{\,i} \mirror{s}_{k,i}
      = \mirror{s}_{k,i-1} \cdot \mirror{s}_{k-1} \mirror{s}_{k-2} \mirror{s}_{k-3} \mirror{s}_{k-2}^{\,a_{k-1} - 1} \cdot \mirror{s}_{k,i-1} \mirror{s}_{k,i} \\
          &= \mirror{s}_{k,i-1} \cdot \mirror{s}_{k-1} L(\mirror{s}_{k-1}) \cdot \mirror{s}_{k,i-1}^{\,2} \mirror{s}_{k-1}.
    \end{align*}
    By \autoref{lem:product_of_squares} the word $\mirror{s}_{k-1} L(\mirror{s}_{k-1})$ is in $\Pi(\oa,\ob)$. Since
    $\mirror{s}_{k,i-1}$ is a solution to \eqref{eq:square}, we have that $\mirror{s}_{k,i-1}^{\,2} \in \Pi(\oa,\ob)$.
    Overall, the factor $\mirror{s}_{k-1} L(\mirror{s}_{k-1}) \mirror{s}_{k,i-1}^{\,2}$ of $S^3$ of length $|S^2|$
    starting at position $r$ is in $\Pi(\oa,\ob)$. Thus the position $r$ of $S$ is repetitive.
    
    Suppose for a contradiction that the suffix of $S$ of length $|S| - r$ is in $\Pi(\oa,\ob)$, that is,
    $S = \mirror{s}_{k,i-1} X_1^2 \cdots X_n^2$ for some minimal square roots $X_j$. It follows that
    $s_{k-1} = X_1^2 \cdots X_n^2$. Since $s_{k-1}$ is a solution to \eqref{eq:square}, it follows that
    $s_{k-1} = (X_1 \cdots X_n)^2$. This contradicts the primitivity of $s_{k-1}$. Similarly if the suffix of $S^2$ of
    length $|S^2| - r$ is in $\Pi(\oa,\ob)$, then $\mirror{s}_{k,i-1} \in \Pi(\oa,\ob)$ contradicting the primitivity of
    $\mirror{s}_{k,i-1}$. We conclude that the position $r$ is nicely repetitive.
  \end{proof}

  \autoref{lem:periodic_image} and \autoref{lem:one_nicely_repetitive} now imply the following:

  \begin{corollary}
    There exist uncountably many linearly recurrent optimal squareful words having (purely) periodic square root.
  \end{corollary}
  \begin{proof}
    We only need to show that there are uncountably many such words. Consider the words in $\Omega$ which can be
    written as a product of the words $S$ and $L$. Viewed over the binary alphabet $\{S, L\}$, these words form an
    infinite subshift $\overline{\Omega}$. Let us show that $\overline{\Omega}$ is minimal. Then the conclusion follows
    by well-known arguments from topology: a minimal subshift is always finite or uncountable and an aperiodic subshift
    cannot be finite (use the fact that a perfect set is always uncountable).
    
    Let $\overline{w} \in \overline{\Omega}$ (we use the notation of the proof of \autoref{lem:product_sl_preserved}).
    Let $\overline{u} \in \Lang(\overline{w})$ be a factor such that $|\overline{u}| \geq 6$. As $|u| \geq 6|S|$, every
    occurrence of $u$ in $\Gamma$ must synchronize to the factorization of $\Gamma$ as a product of $S$ and $L$. It
    follows that every return to $u$ in $\Gamma$ is a product of $S$ and $L$. Since the return time of $u$ is finite in
    $\Gamma$, the return time of the word $\overline{u}$ in $\overline{w}$ is also finite. Hence $\overline{\Omega}$ is
    minimal.
  \end{proof}

  We also prove the following weaker result, which we need later.

  \begin{lemma}\label{lem:nicely_repetitive_position}
    The position $|S| - |S_6|$ of $S$ is repetitive.
  \end{lemma}
  \begin{proof}
    We prove first by induction that the prefix of the word $S_6 \mirror{s}_{k,\ell}^{\,2}$ of length
    $2|\mirror{s}_{k,\ell}| - |S_6|$ is a product of minimal squares for $k \geq 2$ and $\ell$ such that
    $0 < \ell \leq a_k$. Let us first establish the base cases.

    Recall that $\mirror{s}_2 = 0(10^\oa)^{\ob+1}$ and $\mirror{s}_{3,1} = S_6$. We have that
    \begin{align*}
      S_6 \mirror{s}_2^{\,2} = 10^{\oa+1}(10^\oa)^{\ob+1}(0(10^\oa)^{\ob+1})^2 = S_5^2 10^{\oa+1}(10^\oa)^{\ob+1} = S_5^2 S_6.
    \end{align*}
    In addition, for $0 < \ell \leq a_3$, we have that
    \begin{align*}
      S_6 \mirror{s}_{3,\ell}^{\,2} = S_6 \mirror{s}_{3,1} \mirror{s}_2^{\,\ell-1} \mirror{s}_{3,\ell} = S_6^2 \mirror{s}_2^{\,\ell-1} \mirror{s}_{3,\ell}
      = S_6^2 \mirror{s}_2^{\,\ell-1} \mirror{s}_1 \mirror{s}_2^{\,\ell}.
    \end{align*}
    The case $\ell = 1$ is clear. So let us assume that $\ell > 1$. We have that
    \begin{align*}
      S_6 \mirror{s}_{3,\ell}^{\,2} = S_6^2 \mirror{s}_2^{\,\ell-1} \mirror{s}_1 \mirror{s}_2^{\,\ell-2} \mirror{s}_0 \mirror{s}_1^{\,\ob} S_6,
    \end{align*}
    so it is sufficient to show that the word
    $\mirror{s}_2^{\,\ell-1} \mirror{s}_1 \mirror{s}_2^{\,\ell-2} \mirror{s}_0 \mirror{s}_1^{\,\ob}$ is in $\Pi(\oa,\ob)$.
    
    Suppose first that $\ell-1$ is even. Then as $\mirror{s}_2$ is a solution to \eqref{eq:square}, it is enough to
    show that $\mirror{s}_1 \mirror{s}_2^{\,\ell-2} \mirror{s}_0 \mirror{s}_1^{\,\ob} \in \Pi(\oa,\ob)$. Since
    $\mirror{s}_1 \mirror{s}_2 = L(\mirror{s}_2) \mirror{s}_1$, we have that
    \begin{align*}
      \mirror{s}_1 \mirror{s}_2^{\,\ell-2} \mirror{s}_0 \mirror{s}_1^{\,\ob} = L(\mirror{s}_2)^{\ell-2} \mirror{s}_1 \mirror{s}_0 \mirror{s}_1^{\,\ob}.
    \end{align*}
    Now $\mirror{s}_1 \mirror{s}_0 \mirror{s}_1^{\,\ob} = L(\mirror{s}_2)$. The word $L(\mirror{s}_2)$ is a solution to
    \eqref{eq:square}, so the conclusion follows as $\ell-1$ is even.

    Suppose next that $\ell-1$ is odd. We need to show that
    $\mirror{s}_2 \mirror{s}_1 \mirror{s}_2^{\,\ell-2} \mirror{s}_0 \mirror{s}_1^{\,\ob} \in \Pi(\oa,\ob)$. Using the
    facts $\mirror{s}_1 \mirror{s}_2 = L(\mirror{s}_2) \mirror{s}_1$ and
    $\mirror{s}_1 \mirror{s}_0 \mirror{s}_1^{\,\ob} = L(\mirror{s}_2)$ we obtain that
    \begin{align*}
      \mirror{s}_2 \mirror{s}_1 \mirror{s}_2^{\,\ell-2} \mirror{s}_0 \mirror{s}_1^{\,\ob} = \mirror{s}_2 L(\mirror{s}_2)^{\ell-1}.
    \end{align*}
    By \autoref{lem:product_of_squares} the word $\mirror{s}_2 L(\mirror{s}_2)$ is a product of minimal squares. Since
    $\ell-1$ is odd and $L(\mirror{s}_2)$ is a solution to \eqref{eq:square}, the conclusion follows.

    We have established the base cases. Now for $k \geq 4$ and $0 < \ell \leq a_k$, we have that
    \begin{align*}
      S_6 \mirror{s}_{k,\ell}^{\,2} = S_6(\mirror{s}_{k-2} \mirror{s}_{k-1}^{\,\ell})^2.
    \end{align*}
    By induction $S_6 \mirror{s}_{k-2} = X_1^2 \cdots X_n^2 S_6$ and $S_6 \mirror{s}_{k-1} = Y_1^2 \cdots Y_m^2 S_6$
    for some minimal square roots $X_1, \ldots, X_n, Y_1, \ldots, Y_m$. Therefore
    \begin{align*}
      S_6 \mirror{s}_{k,\ell}^{\,2} = (X_1^2 \cdots X_n^2(Y_1^2 \cdots Y_m^2)^\ell)^2 S_6.
    \end{align*}
    We have thus proved that the prefix of the word $S_6 \mirror{s}_{k,\ell}^{\,2}$ of length
    $2|\mirror{s}_{k,\ell}| - |S_6|$ is a product of minimal squares for $k \geq 2$ and $\ell$ such that
    $0 < \ell \leq a_k$.

    Now if $S = \mirror{s}_{k,\ell}$ for some $k \geq 2$ and $\ell$ such that $0 < \ell \leq a_k$, then the claim is
    clear by the above. Suppose that $S = L(\mirror{s}_{k,\ell})$. Now if
    $S_6 \mirror{s}_{k,\ell} \notin \Pi(\oa,\ob)$, then two applications of \autoref{lem:exchange_squares} (first with
    $u = S_6$, $v = S^2$ and then with $u = S_6L$, $v = S$) show that the claim holds. Assume that $S_6
    \mirror{s}_{k,\ell} \in \Pi(\oa,\ob)$. Since the prefix of $S_6 \mirror{s}_{k,\ell}$ of length
    $2|\mirror{s}_{k,\ell}|-|S_6|$ is in $\Pi(\oa,\ob)$, this means that the prefix of $\mirror{s}_{k,\ell}$ of length
    $|\mirror{s}_{k,\ell}|-|S_6|$ is in $\Pi(\oa,\ob)$. It is sufficient to show that the prefixes of
    $\mirror{s}_{k,\ell}$ and $L(\mirror{s}_{k,\ell})$ of length $2|\mirror{s}_2|$ are in $\Pi(\oa,\ob)$. Since
    $\mirror{s}_1 \mirror{s}_2 = L(\mirror{s}_2) \mirror{s}_1$, the word $\mirror{s}_{4,1} = \mirror{s}_2 \mirror{s}_3$
    has $\mirror{s}_2 L(\mirror{s}_2)$ as a prefix. If $a_3 > 1$, then the word
    $\mirror{s}_3 = \mirror{s}_1 \mirror{s}_2^{\,a_3}$ has $L(\mirror{s}_2) \mirror{s}_2$ as a prefix. Finally if
    $a_3 = 1$, then the word $\mirror{s}_{5,1} = \mirror{s}_3 \mirror{s}_4 = \mirror{s}_1 \mirror{s}_2 \mirror{s}_4$
    has $L(\mirror{s}_2)^2$ as a prefix. \autoref{lem:product_of_squares} shows that $\mirror{s}_2 L(\mirror{s}_2)$,
    $L(\mirror{s}_2) \mirror{s}_2$, and $L(\mirror{s}_2)^2$ are all in $\Pi(a,b)$. The conclusion follows.
  \end{proof}

  There is no clear pattern for other positions in $\B_S$; it depends on the word $S$ if a position in $\B_S$ is
  repetitive or not. The position $|S| - |S_6|$ is not always nicely repetitive. Suppose that $\oa = 1$, $\ob = 0$, and
  $S = \mirror{s}_{3,3} = 10(010)^3$. Then the factor beginning at position $|S| - |S_6| = 6$ of $S^3$ of length
  $|S^2|$ is a product of minimal squares: $(10010)^2 \cdot (010)^2 \cdot (100)^2$. As
  $|(10010)^2 \cdot (010)^2| = 16 = |S^2| - 6$, the position $6$ is not nicely repetitive.
  
  Since none of the minimal squares can be a proper prefix of another minimal square, it is easy to factorize words as
  products of minimal squares from left to right. Next we consider what happens if we start to backtrack from a given
  position to the left.

  \begin{lemma}[Backtracking Lemma]\label{lem:backtracking}
    Let $X, Y_1, \cdots Y_n$ be minimal square roots. Let $w$ be a word having both of the words $X^2$ and
    $Y_1^2 \cdots Y_n^2$ as suffixes. If $|X| > |Y_n|$, then $|X| > |Y_1 \cdots Y_n|$ and the word $Y_1 \cdots Y_n$ is
    a suffix of $X$.
  \end{lemma}
  \begin{proof}
    Suppose that $|X| > |Y_n|$. We may assume that $n$ is as large as possible. We prove the lemma by considering
    different options for the word $X$.

    Clearly we cannot have that $X = S_1$. Let $X = S_4$. Now $X^2$ can have a proper minimal square suffix only if
    $\oa>1$. If $\oa$ is even, then we must have that
    \begin{align*}
      X^2 = 10^\oa 1(S_1^2)^{\oa/2} \ \text{ and } \ Y_{n-\oa/2+1} = \ldots = Y_n = S_1.
    \end{align*}
    The suffix $(S_1)^{\oa/2}$ of $w$ cannot be preceded by $S_2^2$ as otherwise
    $w$ would have $S_2 S_1^\oa = 010^{2\oa-1}$ as a suffix; this is not possible as $2\oa-1 > \oa$. Therefore there is
    no choice for $Y_{n-\oa/2}$. Thus $|Y_1^2 \cdots Y_n^2| < |X^2|$ and $Y_1 \cdots Y_n$ is a suffix of $X$. If $\oa$
    is odd, then similarly
    \begin{align*}
      X^2 = 10^\oa10(S_1^2)^{(\oa-1)/2} \ \text{ and } \ Y_{n-(\oa-1)/2+1} = \ldots = Y_n = S_1.
    \end{align*}
    Again there is no choice for $Y_{n-(\oa-1)/2}$, and the conclusion holds. Similar considerations show that the
    conclusion holds if $X \in \{S_2, S_3\}$.

    Let then $X = S_5$. It is obvious that now $Y_n \in \{S_1, S_3, S_4\}$. If $Y_n = S_1$ or $\ob=0$, then like above
    $Y_1 = \ldots = Y_n = S_1$ and $Y_1 \cdots Y_n$ is a suffix of $X$. We may thus suppose that $\ob > 0$. Say
    $Y_n = S_3$. Then we must have $\ob = 1$ and $X^2 = 10^{\oa+1}10^{\oa-1} Y_n^2$. Like above, the remaining minimal
    square roots $Y_i$ with $i < n$ must equal to $S_1$ and there must be $\lfloor (\oa-1)/2 \rfloor$ of them. Since
    there is no further choice, the conclusion holds as clearly $Y_1 \cdots Y_n$ is a suffix of $X$. Suppose then that
    $\ob > 1$. The next case is $Y_n = S_4$. Assume first that $\ob$ is even. Then it is straightforward to see that
    necessarily
    \begin{align*}
      Y_{n-\ob/2+1} = \ldots = Y_n = S_4 \ \text{ and } X^2 = 10^{\oa+1}(10^\oa)^\ob10^{\oa+1}(S_4^2)^{\ob/2}.
    \end{align*}
    Thus $Y_{n-\ob/2} = S_1$ and, further, it must be that
    \begin{align*}
      Y_{n-\ob} = \ldots = Y_{n-\ob/2-1} = S_2 \ \text{ and } \ X^2 = 10^{\oa+1}10^{\oa-1} (S_2^2)^{\ob/2} S_1^2 (S_4^2)^{\ob/2}.
    \end{align*}
    Like before, the remaining minimal squares $Y_i$ with $i < n-\ob$ must equal to $S_1$ and there must be
    $\lfloor (\oa-1)/2 \rfloor$ of them. Therefore
    \begin{align*}
      Y_1 \cdots Y_n = S_1^{\lfloor (\oa-1)/2 \rfloor} S_2^{\ob/2} S_1 S_4^{\ob/2} = 0^{\lfloor (\oa-1)/2 \rfloor + 1}(10^\oa)^\ob
    \end{align*}
    is a suffix of $X$, so the conclusion holds. If $\ob$ is odd, then in a similar fashion
    \begin{align*}
      X^2 = 10^{\oa+1}10^{\oa-1} (S_2^2)^{(\ob-1)/2} S_3^2 (S_4)^{(\ob-1)/2},
    \end{align*}
    so $Y_{n-(\ob-1)/2} = S_3$ and
    \begin{align*}
      Y_{n-(\ob-1)/2+1} = \ldots = Y_n = S_4 \ \text{ and } \ Y_{n-\ob+1} = \ldots = Y_{n-(\ob-1)/2-1} = S_2.
    \end{align*}
    Again, the final $\lfloor (\oa-1)/2 \rfloor$ minimal square roots must equal $S_1$. Since 
    \begin{align*}
      Y_1 \cdots Y_n = S_1^{\lfloor (\oa-1)/2 \rfloor} S_2^{(\ob-1)/2} S_3 S_4^{(\ob-1)/2} = 0^{\lfloor (\oa-1)/2 \rfloor + 1}(10^\oa)^\ob
    \end{align*}
    is a suffix of $X$, the conclusion holds.

    If $X = S_6$, then it is clear that $Y_n \neq S_5$. The conclusion follows as in the case $X = S_5$.
  \end{proof}

  The next lemma is useful in the proof of \autoref{thm:preserved_periodic}.

  \begin{lemma}\label{lem:can_not_continue}
    Let $w$ be an infinite product of the words $S$ and $L$ and $\ell_1$, $\ell_2$, $\ell_3$ be positions of $w$ such
    that $\ell_1 < \ell_2 < \ell_3$. Let $r$ be the largest integer such that $\ell_1 \geq r|S|$. If
    \begin{itemize}
      \item $w[\ell_1,\ell_3-1], w[\ell_2,\ell_3-1] \in \Pi(\oa,\ob)$,
      \item $\ell_1 - r|S| \in \B_S$, and
      \item $\ell_2 \leq (r+1)|S|$,
    \end{itemize}
    then for all $u \in \Pi(\oa,\ob)$ such that $uw[\ell_2,\ell_3-1]$ is a suffix of $w[0,l_3-1]$ we have that
    $|uw[\ell_2,\ell_3-1]| < |w[\ell_1,\ell_3-1]|$.
  \end{lemma}
  \begin{proof}
    Let $v = w[\ell_1,\ell_3-1]$ and $u = w[\ell_2,\ell_3-1]$. Since $v,u \in \Pi(\oa,\ob)$, we may write
    $v = X_1^2 \cdots X_n^2$ and $u = Y_1^2 \cdots Y_m^2$ for some minimal square roots $X_i$ and $Y_i$. If $n \geq m$
    and $X_{n-m+i} = Y_i$ for all $i \in \{1,\ldots,m\}$, then as $|v| > |u|$, we must have that $n > m$. This means
    that the prefix $X_1^2$ of $v$ ends before the position $\ell_2$, that is, $\ell_1+|X_1^2| < \ell_2 \leq (r+1)|S|$.
    This contradicts the fact that $\ell_1-r|S| \in \B_S$. Therefore as $|v| > |u|$, we we conclude that
    there exists maximal $j \in \{1,\ldots,m\}$ such that $X_{n-m+j} \neq Y_j$. If $|Y_j| > |X_{n-m+j}|$, then by the
    \nameref{lem:backtracking} we have that $|X_1^2 \cdots X_{n-m+j}^2| < |Y_j^2|$. This is not possible as
    $|v| > |u|$. Therefore $|Y_j| < |X_{n-m+j}|$. Let $z \in \Pi(\oa,\ob)$ be such that $zu$ is a suffix of $w[0,l_3-1]$. Write
    $z = Z_1^2 \cdots Z_t^2$ for minimal square roots $Z_i$. Applying the \nameref{lem:backtracking} to the words
    $X_{n-m+j}^2$ and $Z_1^2 \cdots Z_t^2 Y_1^2 \cdots Y_j^2$
    yields that $|Z_1^2 \cdots Z_t^2 Y_1^2 \cdots Y_j^2| < |X_{n-m+j}^2|$. It follows that $|zu| < |v|$.
  \end{proof}

  Finally we can give a proof of \autoref{thm:preserved_periodic}.

  \begin{proof}[Proof of \autoref{thm:preserved_periodic}]
    Let $w \in \Omega$. Since $\Gamma$ is uniformly recurrent and a product of the words $S$ and $L$, there exists a
    word $w' \in \Omega$ such that $w'$ is a product of $S$ and $L$ and $w = T^\ell(w')$ for some $\ell$ such that
    $0 \leq \ell < |S|$ (recall that a product of $S$ and $L$ occurring in $\Gamma$ having length at least $6|S|$ must
    synchronize to the factorization of $\Gamma$ as a product of $S$ and $L$). If $\ell = 0$, then the conclusion holds
    by \autoref{lem:product_sl_preserved}, so we can assume that $\ell > 0$. Write $w$ as a product of minimal squares:
    $w = X_1^2 X_2^2 \cdots$. Let
    \begin{align*}
      r_1 = \max\{\{0\} \cup \{i \in \{1,2,\ldots\}\colon |X_1^2 \cdots X_i^2| \leq |S| - \ell\}\}.
    \end{align*}
    If $r_1 > 0$, then set $\ell_1 = \ell + |X_1^2 \cdots X_{r_1}^2|$. If $r_1 = 0$, then we set $\ell_1 = \ell$. By
    the maximality of $r_1$ and by the definition of the set $\B_S$, it follows that $\ell_1 \in \B_S \cup \{|S|\}$
    (indeed, the word $L$ also has $S_6$ as a suffix). See \autoref{fig:l_positions}.

    To aid comprehension we have separated different parts of the proof as distinct claims with their own proofs. Any
    new definitions and assumptions given in one of the subproofs are valid only up to the end of the subproof.

    \begin{claim}
      If $\ell_1 = |S|$, then $\sqrt{w} \in \Omega$.
    \end{claim}
    \begin{proof}
      Suppose that $\ell_1 = |S|$. By the definition of the number $r_1$, we have that $r_1 > 0$ and the word
      $T^{|S|-\ell}(w) = T^{|S|}(w') = X_{r_1+1}^2 X_{r_2+2}^2 \cdots$ is a product of the words $S$ and $L$. Now
      $zw' \in \Omega$ where $z \in \{S, L\}$. Since $zw'$ is a product of $S$ and $L$, by
      \autoref{lem:product_sl_preserved} $\sqrt{zw'} \in \Omega$. By the choice of $S$ as a solution to
      \eqref{eq:square} and by \autoref{lem:product_of_squares}, the first $|S^2|$ letters of $zw'$ can be written as a
      product of minimal squares. Hence $zw' = Y_1^2 \cdots Y_n^2 X_{r_1+1}^2 X_{r_1+2}^2 \cdots$ for some minimal
      square roots $Y_1, \ldots, Y_n$. By the \nameref{lem:backtracking}, we have that $X_1 \cdots X_{r_1}$ is a suffix
      of $Y_1 \cdots Y_n$. Thus the word $\sqrt{w} = X_1 \cdots X_{r_1} X_{r_1+1} \cdots$ is a suffix of the word
      $\sqrt{zw'} = Y_1 \cdots Y_n X_{r_1} X_{r_1+1} \cdots$. Therefore
      $\Lang(\sqrt{w}) \subseteq \Lang(\sqrt{zw'}) = \Lang(\Gamma)$, so $\sqrt{w} \in \Omega$.
    \end{proof}

    We assume that $\ell_1 \in \B_S$. Now either the position $\ell_1$ of $S$ is nicely repetitive or it is not.

    \begin{claim}\label{cl:nicely_repetitive}
      If $\ell_1$ is a nicely repetitive position of $S$, then $\sqrt{w}$ is periodic with minimal period conjugate to
      $S$.
    \end{claim}
    \begin{proof}
      By \autoref{lem:periodic_image} the word $\sqrt{T^{\ell_1}(w')}$ is periodic with minimal period $z$ conjugate to
      $S$. If $\ell_1 = \ell$, then there is nothing more to prove, so assume that $\ell_1 \neq \ell$. There exists
      $u, v \in \{S,L\}$ such that $uvw' \in \Omega$. Since $\ell_1$ is a nicely repetitive position of $S$, the prefix
      of $T^{\ell_1}(uvw')$ of length $|S^2|$ is a product of minimal squares and its square root equals $z$ by
      \autoref{lem:periodic_image}. Since the factor $w'[\ell,\ell_1-1]$ is also a product of minimal squares, the
      \nameref{lem:backtracking} implies that $\sqrt{w'[\ell,\ell_1-1]}$ is a suffix of $z$. Now
      $\sqrt{w} = \sqrt{w'[\ell,\ell_1-1]}\sqrt{T^{\ell_1}(w')}$, so $\sqrt{w}$ is periodic with minimal period
      conjugate to $S$.
    \end{proof}

    If the position $\ell_1$ of $S$ is not nicely repetitive, then either it is not repetitive or it is repetitive but
    not nicely repetitive.

    \begin{claim}\label{cl:repetitive_not_nicely}
      If $\ell_1$ is repetitive but not nicely repetitive position of $S$, then $\sqrt{w} \in \Omega$.
    \end{claim}
    \begin{proof}
      Suppose that $\ell_1$ is a repetitive but not a nicely repetitive position of $S$. This means that either
      $S^3[\ell_1,|S|-1] \in \Pi(\oa,\ob)$ or $S^3[\ell_1,|S^2|-1] \in \Pi(\oa,\ob)$ (they both cannot be in
      $\Pi(\oa,\ob)$ as this would imply that $S$ is not primitive). Thus either $w'[\ell_1,|S|-1] \in \Pi(\oa,\ob)$ or
      $w'[\ell_1,|S^2|-1] \in \Pi(\oa,\ob)$ (in the latter case \autoref{lem:exchange_squares} ensures that
      $w'[\ell_1,|S^2|-1] \in \Pi(\oa,\ob)$). The former case is, however, not possible as it would contradict the
      maximality of $r_1$. Thus only the latter option is possible. Since $w'$ is a product of the words $S$ and $L$,
      the prefix $w'[0,|S^2|-1]$ of $w'$ is a product of minimal squares. Since
      $w'[\ell,\ell_1-1], w'[\ell_1,|S^2|-1] \in \Pi(\oa,\ob)$, the \nameref{lem:backtracking} implies that
      $\sqrt{w'[\ell,|S^2|-1]}$ is a suffix of $\sqrt{w'[0,|S^2|-1]}$. Thus $\sqrt{w}$ is a suffix of $\sqrt{w'}$. As
      $\sqrt{w'} \in \Omega$ by \autoref{lem:product_sl_preserved}, we conclude that $\sqrt{w} \in \Omega$.
    \end{proof}

    Now we may suppose that $\ell_1$ is not a repetitive position of $S$. We let
    \begin{align*}
      r_2 &= \max\{i \in \{r_1+1,r_1+2,\ldots\}\colon |X_1^2 \cdots X_i^2| \leq |S^2| - \ell\}, \\
      r_3 &= \max\{i \in \{r_2+1,r_2+2,\ldots\}\colon |X_1^2 \cdots X_i^2| \leq |S^3| - \ell\}, \text{ and } \\
      r_4 &= \max\{i \in \{r_3+1,r_3+2,\ldots\}\colon |X_1^2 \cdots X_i^2| \leq |S^4| - \ell\}.
    \end{align*}
    The numbers $r_2$, $r_3$, and $r_4$ are well-defined as the words $S$ and $L$ are not minimal squares. We set
    \begin{align*}
      \ell_2 &= \ell_1 + |X_{r_1+1}^2 \cdots X_{r_2}^2|, \\
      \ell_3 &= \ell_2 + |X_{r_2+1}^2 \cdots X_{r_3}^2|, \text{ and } \\
      \ell_4 &= \ell_3 + |X_{r_3+1}^2 \cdots X_{r_4}^2|.
    \end{align*}
    Intuitively, the positions $\ell_1$, $\ell_2$, $\ell_3$, and $\ell_4$ are the successive positions of $w$ which are
    closest from the left to the boundaries of the words $S$ and $L$ in the factorization of $w'$ as a product of the
    words $S$ and $L$ such that the prefix up to the position is a product of minimal squares; see
    \autoref{fig:l_positions}. Let $g_1 = \ell_1$, $g_2 = \ell_2 - |S|$, $g_3 = \ell_3 - |S^2|$, and
    $g_4 = \ell_4 - |S^3|$. It is clear by the definitions that $g_i \in B \cup \{|S|\}$ for all $i \in \{1,2,3,4\}$.
    
    \begin{claim}
      We have that $g_1, g_3 \neq |S|$. If $g_2$ or $g_4$ equals $|S|$, then $\sqrt{w} \in \Omega$.
    \end{claim}
    \begin{proof}
      By our assumption that $\ell_1 \in \B_S$, we have that $g_1 \neq |S|$. If $g_2 = |S|$, then the factor
      $w'[\ell_1,|S^2|-1]$ would be a product of minimal squares. This case was already considered in
      \autoref{cl:repetitive_not_nicely} where we concluded that $\sqrt{w} \in \Omega$.
      
      Suppose that $g_3 = |S|$. Consider the positions $\ell_1$ and $|S|$ of $w'$. Both of the factors
      $u = w'[\ell_1,\ell_3-1]$ and $v = w'[|S|,\ell_3-1] = w'[\ell_3-|S^2|,\ell_3-1]$ are in $\Pi(\oa,\ob)$. Now
      $zw' \in \Omega$ for some $z \in \{S,L\}$. Since $SS, SL, LS, LL \in \Pi(\oa,\ob)$, the prefix of $zw'$ of length
      $|S^2|$ is in $\Pi(\oa,\ob)$. \autoref{lem:can_not_continue} applied to the word $(zw')[0,|S|+l_3-1]$ implies that
      $|(zw')[0,|S|+l_3-1]| < |u| < |S^3|$ which is nonsense. Therefore $g_3 \neq |S|$.

      Assume then that $g_4 = |S|$. Suppose for a contradiction that $g_2 \neq g_4$. Both of the factors
      $u' = w'[\ell_2,\ell_4-1]$ and $v' = w'[|S|^2,\ell_4-1] = w'[\ell_4-|S^2|,\ell_4-1]$ are in $\Pi(\oa,\ob)$. Since
      $g_2 \neq g_4$, also $\ell_2 \neq |S^2|$. Thus by the definition of $\ell_2$, we have that $\ell_2 < |S_2|$.
      \autoref{lem:can_not_continue} applied to the word $w'[0,\ell_4-1]$ shows that $|w'[0,l_4-1]| < |u'| < |S^3|$
      which is absurd. This contradiction shows that $g_2 = g_4 = |S|$, so $\sqrt{w} \in \Omega$.
    \end{proof}

    We may now assume that $g_i \in \B_S$ for all $i \in \{1,2,3,4\}$.
    
    \begin{claim}
      The position $g_2$ of $S$ is nicely repetitive.
    \end{claim}
    \begin{proof}
      Assume on the contrary that neither of the positions $g_2$ and $g_3$ is a repetitive position of $S$. First note
      that as $g_1$ is not repetitive, we have that $g_3 \neq g_1$. Similarly $g_2 \neq g_4$. If $g_1 = g_2$, then it
      follows from \autoref{lem:exchange_squares} and the definitions of the positions $l_2$ and $l_3$ that
      $g_2 = g_3$; a contradiction. Hence $g_1 \neq g_2$. Similarly $g_2 \neq g_3$ as otherwise the position $g_2$
      would be repetitive. Finally, $g_3 \neq g_4$ because $g_3$ is not repetitive. We have two cases: either
      $g_1 = g_4$ or $g_1 \neq g_4$.

      Assume that $g_4 \neq g_1$. By \autoref{lem:nicely_repetitive_position} the position $|S| - |S_6|$ of $S$ is
      repetitive, so $g_1, g_2, g_3 \in \B_S \setminus \{|S|-|S_6|\} = \{|S|-|S_1|,|S|-|S_3|,|S|-|S_4|\}$. Since all of
      the positions $g_1$, $g_2$, and $g_3$ are distinct, the only option is that $g_4 = |S|-|S_6|$. Since the position
      $|S|-|S_6|$ is repetitive, by \autoref{lem:exchange_squares} the factor $u = w'[\ell_4-|S^2|,\ell_4-1]$ is in
      $\Pi(\oa,\ob)$. By the definition of the positions $\ell_2$, $\ell_3$, and $\ell_4$ also
      $v = w'[\ell_2,\ell_4-1] \in \Pi(\oa,\ob)$. Since $g_2 \neq g_4$, also $\ell_2 \neq \ell_4-|S^2|$. Since
      $|S|-|S_6|$ is the smallest element of the set $\B_S$, we have that $\ell_2 > \ell_4-|S^2|$. As
      $w[l_1,l_2-1] \in \Pi(\oa,\ob)$, we obtain by \autoref{lem:can_not_continue} that $|w[l_1,l_4-1]| < |u| = |S^2|$.
      This is a contradiction.

      Hence we have that $g_1 = g_4$. Since the factor $w[\ell_1,\ell_2-1]$ is a product of minimal squares, the number
      $c_1 = \ell_2-\ell_1$ is even. Similarly the numbers $c_2 = \ell_3-\ell_2$ and $c_3 = \ell_4-\ell_3$ are even.
      Thus the number $c_1 + c_2 + c_3 = 3|S|$ is even, so $|S|$ is even. It follows that the numbers
      $d_1 = g_2 - g_1$, $d_2 = g_3 - g_2$, and $d_3 = g_4 - g_3 = g_1 - g_3$ are even. However, exactly two of the
      numbers $|S_1|$, $|S_3|$, and $|S_4|$ have odd length. Hence exactly two of the numbers $g_1$, $g_2$, and $g_3$
      are odd. Thus it is not possible that all of the numbers $d_1$, $d_2$, and $d_3$ are even. This is a
      contradiction.

      The previous contradiction shows that either of the positions $g_2$ and $g_3$ is a repetitive position of $S$.
      Suppose for a contradiction that $g_3$ is repetitive. We have that $w'[\ell_1,\ell_3-1] \in \Pi(\oa,\ob)$ and
      $w'[\ell_3-|S^2|,\ell_3-1] \in \Pi(\oa,\ob)$. Similar to the second paragraph of this subproof, using
      \autoref{lem:can_not_continue} we obtain a contradiction unless $g_1 = g_3$. Even this conclusion is
      contradictory as $g_1$ is not repetitive. Therefore $g_3$ can not be repetitive, so $g_2$ is a repetitive
      position of $S$. Now if $g_2$ would not be nicely repetitive, we would have by the maximality of $r_2$ that
      $w'[\ell_2,|S^3|-1] \in \Pi(\oa,\ob)$, that is, $g_3 = |S|$. However, since $g_3 \in \B_S$, we have that $g_2$ is a
      nicely repetitive position of $S$.
    \end{proof}

    We are now in the final stage of the proof. We will show that $\sqrt{w}$ is periodic with minimal period conjugate
    to $|S|$.

    We can now argue as in the proof of \autoref{cl:nicely_repetitive}. Since $g_2$ is a nicely repetitive position of
    $S$, by \autoref{lem:periodic_image} the word $\sqrt{T^{\ell_2}(w')}$ is periodic with minimal period $z$ conjugate
    to $S$. We have that $uw' \in \Omega$ for some $u \in \{S,L\}$. Since $g_2$ is a nicely repetitive position of $S$,
    the prefix of $T^{g_2}(uw')$ of length $|S^2|$ is a product of minimal squares and its square root equals $z$ by
    \autoref{lem:periodic_image}. Since $w'[\ell,\ell_2-1] \in \Pi(\oa,\ob)$, the \nameref{lem:backtracking} implies that
    $\sqrt{w'[\ell,\ell_2-1]}$ is a suffix of $z$. Now $\sqrt{w} = \sqrt{w'[\ell,\ell_2-1]}\sqrt{T^{\ell_2}(w')}$, so
    $\sqrt{w}$ is periodic with minimal period conjugate to $S$.

    By \autoref{lem:one_nicely_repetitive} the word $S$ always has at least one nicely repetitive position. It
    therefore follows that there exists a word in $\Omega$ having periodic square root.
  \end{proof}

  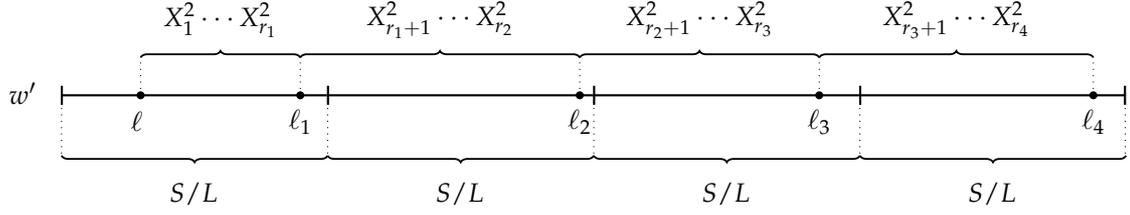
\begin{figure}
  \centering
  \begin{tikzpicture}
    \coordinate (O) at (0,0);   
    \coordinate (W) at (3.5,0); 
    \node (Z) at ($ (O) + (-0.5,0.03) $) {$w'$};
    \foreach \i in {0,...,3} {
      \pgfmathsetmacro\r{\i + 1}
      \draw[line width=0.8pt,black,|-] ($ (O) + \i*(W) $) -- ($ (O) + \r*(W) $);
    }
    \draw[line width=0.8pt,black,-|] ($ (O) + 4*(W) $) -- ($ (O) + 4*(W) + (0.01,0) $);

    \tikzstyle{tick}=[circle,fill,inner sep=0,minimum size=3pt]
    \node (X1) at ($ (O) + 0.3*(W) $) [tick] {};
    \foreach \i [count=\j from 2] in {0.9,0.95,0.85,0.88} {
      \pgfmathsetmacro\r{\j - 2}
      \node (X\j) at ($ (O) + \r*(W) + \i*(W) $) [tick] {};
    }
    \coordinate (D1) at (0.0,-0.35); 
    \node (L1) at ($ (X1) + (-0.05,-0.35) $) {$\ell$};
    \foreach \i in {2,...,5} {
      \pgfmathtruncatemacro\r{\i - 1} 
      \node (L\i) at ($ (X\i) + (D1) $) {$\ell_\r$};
    }
    \coordinate (U1) at (0,0.5); 
    \foreach \i in {1,...,4} {
      \pgfmathtruncatemacro\r{\i + 1}
      \draw[decorate,decoration={brace},line width=0.7pt] ($ (X\i) + (U1) $) -- ($ (X\r) + (U1) $);
    }
    \foreach \i in {1,...,5} {
      \draw[dotted] ($ (X\i) $) -- ($ (X\i) + (U1) $);
    }
    \coordinate (U2) at (0,1);
    \node (M1) at ($ 0.5*(X1) + 0.5*(X2) + (U2) $) {$X_1^2 \cdots X_{r_1}^2$};
    \node (M2) at ($ 0.5*(X2) + 0.5*(X3) + (U2) $) {$X_{r_1+1}^2 \cdots X_{r_2}^2$};
    \node (M3) at ($ 0.5*(X3) + 0.5*(X4) + (U2) $) {$X_{r_2+1}^2 \cdots X_{r_3}^2$};
    \node (M4) at ($ 0.5*(X4) + 0.5*(X5) + (U2) $) {$X_{r_3+1}^2 \cdots X_{r_4}^2$};
    \coordinate (D2) at (0,-0.8); 
    \coordinate (E) at (0.01,0);  
    \foreach \i in {0,...,3} {
      \pgfmathsetmacro\r{\i + 1}
      \draw[decorate,decoration={brace,mirror},line width=0.7pt] ($ (O) + \i*(W) + (D2) + (E) $) -- ($ (O) + \r*(W) + (D2) + (E) $);
      \draw[dotted] ($ (O) + \i*(W) + (E) $) -- ($ (O) + \i*(W) + (D2) + (E) $);
    }
    \draw[dotted] ($ (O) + 4*(W) $) -- ($ (O) + 4*(W) + (D2) $); 
    \coordinate (D3) at ($ (D2) + (0,-0.5) $);
    \foreach \i in {1,...,4} {
      \pgfmathtruncatemacro\r{\i}
      \pgfmathsetmacro\s{\i - 0.5}
      \node (S\i) at ($ (O) + \s*(W) + (D3) $) {$S/L$};
    }
  \end{tikzpicture}
  \caption{The positions $\ell$, $\ell_1$, $\ell_2$, $\ell_3$, and $\ell_4$ of $w'$ and the minimal squares between the positions.}
  \label{fig:l_positions}
  \end{figure}

  \section{Remarks on Generalizations}\label{sec:generalizations}
  It is natural to think that the square root map could be generalized to obtain a cube root map and, further, a
  $k^\text{th}$ root map. However, in \cite[Theorem~5.3.]{diss:kalle_saari} Saari proves the following reformulation of
  a result of Mignosi, Restivo, and Salemi.

  \begin{proposition}\label{prp:repetitivity_periodicity}
    If $w$ is an everywhere $\alpha$-repetitive word with $\alpha \geq \phi + 1$, where $\phi$ is the golden mean, then 
    $w$ is ultimately periodic.
  \end{proposition}

  Generalizing the square root map to a cube root map would require everywhere $3$-repetitive words. By the above such
  words must be ultimately periodic, so we expect that this direction of research would not be fruitful.

  Another way to generalize the square root map is to use abelian powers instead of ordinary powers. For abelian powers
  a result like \autoref{prp:repetitivity_periodicity} does not exist. For instance, by
  \cite[Theorem~1.9.]{2011:abelian_complexity_of_minimal_subshifts} every position in a Sturmian word begins with an
  abelian $k^{th}$ power for all $k \geq 2$. Abelian square root can be defined for e.g. optimal squareful words as we
  will see shortly. However, abelian cubes in Sturmian words do not work. Consider again the Fibonacci word $f$. The
  minimal abelian cube prefix of $T(f)$ is $10\cdot01\cdot01$. This abelian cube is followed by the factor $00$, so the
  root of the next abelian cube must begin with $00$. Hence if we define the abelian cube root of $T(f)$ to be the
  product of the roots of the abelian cubes, the resulting word begins with $1000$ which is not a factor of $f$. Thus
  by defining an abelian cube root map in this way, we lose the main property that the mapping preserves the languages
  of Sturmian words.

  In \cite{2010:everywhere_alpha-repetitive_sequences_and_sturmian_words} Saari also considers optimal abelian
  squareful words. \emph{Optimal abelian squareful words} are defined by replacing minimal squares with minimal abelian
  squares in the definition of optimal squareful words. Let $w = X_1 X_1' X_2 X_2' \cdots$ be a product of minimal
  abelian squares $X_i X_i'$. We define its \emph{abelian square root} as the word $\sqrt[ab]{w} = X_1 X_2 \cdots$. It
  follows from \cite[Theorem~18]{2010:everywhere_alpha-repetitive_sequences_and_sturmian_words} that the six minimal
  squares are products of exactly five minimal abelian squares (this is straightforward to verify directly).
  Thus if $w$ is an optimal squareful word, then $\sqrt{w} = \sqrt[ab]{w}$.
  Thus by \autoref{thm:square_root} the abelian square root of a Sturmian word $s_{x,\alpha}$ is the Sturmian word
  $s_{\psi(x),\alpha}$. Also, by \autoref{thm:preserved_periodic} there exists a minimal subshift $\Omega$ such that
  for all $w \in \Omega$ either $\sqrt[ab]{w} \in \Omega$ or $\sqrt[ab]{w}$ is periodic. Saari proves in
  \cite[Theorem~19]{2010:everywhere_alpha-repetitive_sequences_and_sturmian_words} that an optimal abelian squareful
  word must have at least five distinct minimal abelian squares, but he leaves the characterization of these sets of
  minimal abelian squares open. Thus it is possible that there exists optimal abelian squareful words which contain other
  minimal abelian squares than those given by \cite[Theorem~18]{2010:everywhere_alpha-repetitive_sequences_and_sturmian_words}. For such words the abelian square
  root map could exhibit different behavior than the square root map (if the square root map is even defined for such
  words). We have not extended our research to this direction.

  We could also generalize the special function $\psi$. Divide the distance $D$ between $x$ and $1-\alpha$ into $k$
  parts and choose the image of $x$ to be $x + \frac{t}{k}D$ among the points
  \begin{align*}
    x + \frac{1}{k}D, x + \frac{2}{k}D, \ldots, x + \frac{k-1}{k}D
  \end{align*}
  to obtain the function
  \begin{align*}
    \psi_{k,t}: \T \to \T, x \mapsto \frac{1}{k}(tx + (k - t)(1-\alpha)).
  \end{align*}
  The map $\psi_{k,t}$ is a perfectly nice function on the circle $\T$, but to make things interesting we would need to
  find a symbolic interpretation for it. We have not figured out any such interpretation for these generalized
  functions.
  
  \section{Open Problems}\label{sec:open_problems}
  In the \autoref{sec:counter_example} we saw that there are non-Sturmian words whose language is preserved under the
  square root map. However, Sturmian words satisfy an even stronger property: by \autoref{thm:square_root} for the
  Sturmian subshift $\Omega_\alpha$ of slope $\alpha$ it holds that $\sqrt{\Omega_\alpha} \subseteq \Omega_\alpha$.
  This property is not satisfied by the aperiodic and minimal subshift $\Omega_\Gamma$ of the word $\Gamma$ constructed
  in \autoref{sec:counter_example} since by \autoref{thm:preserved_periodic} there is a word in $\Omega_\Gamma$ having
  periodic square root; since $\Omega_\Gamma$ is aperiodic and minimal, it cannot contain such words. We are thus led to
  ask the following question we could not answer:

  \begin{question}
    If $\Omega$ is a subshift containing optimal squareful words satisfying $\sqrt{\Omega} \subseteq \Omega$,
    does the subshift $\Omega$ only contain Sturmian words?
  \end{question}

  Let us briefly see that if we do not require all words in $\Omega$ to be aperiodic then the above question has a
  negative answer.

  \begin{proposition}
    There exists a non-minimal non-Sturmian subshift $\Omega$ containing squareful words such that
    $\sqrt{\Omega} \subseteq \Omega$.
  \end{proposition}
  \begin{proof}[Proof Sketch]
    Let $S$ be a seed solution as in \autoref{sec:counter_example}, and let $\Gamma$ be a corresponding fixed point of
    the square root map generated by the seed $S$ as in \autoref{sec:counter_example}. Further, set
    $\Delta = S^\omega$, let $\Omega_\Delta$ be the subshift generated by $\Delta$, and let $\Omega_\Gamma$ be the
    subshift generated by $\Gamma$. If $w \in \Omega_\Gamma$, then by \autoref{thm:preserved_periodic} either $\sqrt{w}
    \in \Omega_\Gamma$ or $\sqrt{w} \in \Omega_\Delta$. Hence if we are able to show that
    $\sqrt{\Omega_\Delta} \subseteq \Omega_\Delta$, then the non-minimal and non-Sturmian subshift
    $\Omega_\Gamma \cup \Omega_\Delta$ has the desired properties.

    Let $w \in \Omega_\Delta$, so $w = T^\ell(\Delta)$ for some $0 \leq \ell < |S|$. Write $w$ as a product of minimal
    squares: $w = X_1^2 X_2^2 \cdots$. We can now argue as in the proof of \autoref{thm:preserved_periodic}. If
    $|X_1^2 \cdots X_n^2| = |S| - \ell$ for some $n \geq 1$ or $|X_1^2 \cdots X_m^2| = |S^2| - \ell$ for some
    $m \geq 1$, then using the fact that $\sqrt{\Delta} = \Delta$ it is straightforward to see that
    $\sqrt{w} \in \Omega_\Delta$. Otherwise either $\ell$ is a nicely repetitive position of $S$ or
    $\ell + |X_1^2 \cdots X_i^2| - |S|$ is a nicely repetitive position of $S$ where
    \begin{align*}
      i = \max\{j \in \{1,2,\ldots\}\colon |X_1^2 \cdots X_j^2| \leq |S^2| - \ell\}.
    \end{align*}
    In both of these cases we deduce with the help of \autoref{lem:periodic_image} that $\sqrt{w} \in \Omega_\Delta$.
  \end{proof}

  There are other interesting related questions. Consider the limit set
  \begin{align*}
    \Omega \cap \sqrt{\Omega} \cap \sqrt{\sqrt{\Omega}} \cap \ldots.
  \end{align*}
  We know very little about the limit set except in the Sturmian case when it contains the two fixed points
  $01c_\alpha$ and $10c_\alpha$. For the word $\Gamma$ of \autoref{sec:counter_example} we proved that the limit set
  contains at least two fixed points. We ask:

  \begin{question}
    When is the limit set nonempty? If it is nonempty, does it always contain fixed points? Can it contain points which
    are not fixed points?
  \end{question}

  It is a genuine possibility that the limit set is empty. Consider for instance the word $\zeta =
  \tau(\sigma^\omega(6))$, the morphic image of the fixed point of the morphism $\sigma: 6 \mapsto 656556, 5 \mapsto 5$
  under $\tau: 6 \mapsto S_6^2, 5 \mapsto S_5^2$ where $S_5 = 100$ and $S_6 = 10010$ are minimal square roots of slope
  $\alpha = [0;2,1,\ldots]$. It is straightforward to verify that $\zeta$ is optimal squareful and uniformly recurrent
  and that the returns to the factor $101$ in $\Lang(\zeta)$ are $10100$, $101(001)^2 00$ and $101(001)^4 00$. By
  considering all possible occurrences of the factor $w = \tau(56565) \in \Lang(\zeta)$ in any product of minimal
  squares of slope $\alpha$, it can be shown that the square root of the product always contains a return to the factor
  $101$ which is not in $\Lang(\zeta)$. Since the factor $w$ occurs in every point in the subshift $\Omega_\zeta$ generated
  by $\zeta$, we conclude that $\Omega_\zeta \cap \sqrt{\Omega_\zeta} = \emptyset$.

  In \autoref{sec:counter_example} we constructed infinite families of primitive solutions to \eqref{eq:square} using
  the recurrence $\gamma_{k+1} = L(\gamma_k)\gamma_k^2$. Why this construction worked was because the seed solution $S$
  and the word $L = L(S)$ satisfy $\sqrt{SS} = S$, $\sqrt{SL} = S$, $\sqrt{LS} = L$, and $\sqrt{LL} = L$, that is,
  $\sqrt{(LSS)^2} = \sqrt{LS\cdot SL \cdot SS} = LSS$. Similarly $\sqrt{(SLLLL)^2} = SLLLL$, so substituting for
  example $S = 01010010$ we obtain the primitive solution
  \begin{align*}
    S_2 S_1 S_4 S_3 S_5 S_4 S_3 S_5 S_6 S_5 S_4 S_3 S_5 S_4 S_3 = 0101001010010010100100101001001010010010 
  \end{align*}
  to \eqref{eq:square} in $\Lang(1,0)$. More solutions can be obtained with analogous constructions. Restricting to the
  languages of optimal squareful words, we ask:

  \begin{question}
    What are the primitive solutions $w$ of \eqref{eq:square} in $\Lang(a,b)$ such that $w$ or $w^2$ is not Sturmian
    and $w$ is not obtainable by the above construction?
  \end{question}

  \section*{Acknowledgments}
  The authors were supported by University of Turku Graduate School UTUGS Matti programme and by the FiDiPro grant
  (137991) from the Academy of Finland.

  We thank our supervisors Juhani Karhumäki and Luca Zamboni for suggesting that the square root map might preserve the
  language of a Sturmian word. We also thank Tero Harju for valuable comments.

  \printbibliography
\end{document}